\documentclass[11pt]{article}
\usepackage{yaonan}

\newcommand{\SPA}{{\sf SPA}}
\newcommand{\MA}{{\sf MA}}
\newcommand{\SPP}{{\sf SPP}}
\newcommand{\UIVVSPP}{{\sf UIVVSPP}}

\newcommand{\SecondPriceAuction}{\textsf{Second Price Auction}}
\newcommand{\MyersonAuction}{\textsf{Myerson Auction}}
\newcommand{\SequentialPostedPricing}{\textsf{Sequential Posted Pricing}}
\newcommand{\UIVVSequentialPostedPricing}{\textsf{UIVV Sequential Posted Pricing}}
\newcommand{\UniformIronedVirtualValueSequentialPostedPricing}{\textsf{Uniform-Ironed-Virtual-Value Sequential Posted Pricing}}

\newcommand{\DRB}{{\sf DRB}}
\newcommand{\OLP}{{\sf OLP}}
\newcommand{\OIP}{{\sf OIP}}
\newcommand{\IP}{{\sf IP}}
\newcommand{\UIVVIP}{{\sf UIVVIP}}

\newcommand{\DualityRelaxationBenchmark}{\textsf{Duality Relaxation Benchmark}}
\newcommand{\OptimalLotteryPricing}{\textsf{Optimal Lottery Pricing}}
\newcommand{\OptimalItemPricing}{\textsf{Optimal Item Pricing}}
\newcommand{\ItemPricing}{\textsf{Item Pricing}}
\newcommand{\UniformIronedVirtualValueItemPricing}{\textsf{Uniform-Ironed-Virtual-Value Item Pricing}}
\newcommand{\UIVVItemPricing}{\textsf{UIVV Item Pricing}}
\newcommand{\UIVV}{\textsf{UIVV}}

\renewcommand{\iff}{\Leftrightarrow}
\renewcommand{\implies}{\Rightarrow}
\newcommand{\Conv}{\mathrm{CE}}

\renewcommand{\P}{\mathbf{P}}

\newcommand{\NP}{\mathbf{NP}}

\newcommand{\sharpP}{\mathbf{\#P}}

\newcommand{\QPTAS}{\mathbf{QPTAS}}

\newcommand{\iron}{\text{\tt Iron}}
\newcommand{\truncate}{\text{\tt Truncate}}
\newcommand{\scale}{\text{\tt Scale}}
\newcommand{\extend}{\text{\tt Extend}}
\newcommand{\perturb}{\text{\tt Perturb}}

\newcommand{\bphi}{\boldsymbol{\phi}}

\newcommand{\bkappa}{\boldsymbol{\kappa}}

\title{Benchmark-Tight Approximation Ratio of Simple Mechanism \\ for a Unit-Demand Buyer\footnote{A preliminary version of this work with the same title appears in the 65th IEEE Annual Symposium on Foundations of Computer Science (FOCS 2024). Most proofs are omitted there but are shown in the current full version.}}

\author{
Yaonan Jin\thanks{Hong Kong University of Science and Technology. Part of this work was done while the author was at Huawei's Taylor Lab. Email: {\tt jinyaonan1996@gmail.com}}
\and
Pinyan Lu\thanks{Shanghai University of Finance and Economics, Laboratory of Interdisciplinary Research of Computation and Economics (SUFE), \& Huawei TCS Lab. Email: {\tt lu.pinyan@mail.shufe.edu.cn}}
}
\date{}

\begin{document}

\maketitle
\begin{abstract}
We study revenue maximization in the unit-demand single-buyer setting. Our main result is that {\UniformIronedVirtualValueItemPricing} guarantees a {\em tight} $3$-approximation to the {\DualityRelaxationBenchmark} [Chawla-Malec-Sivan, EC'10/GEB'15; Cai-Devanur-Weinberg, STOC'16/ SICOMP'21], breaking the barrier of $4$ since [Chawla-Hartline-Malec-Sivan, STOC'10; Chawla-Malec-Sivan, EC'10/GEB'15]. To our knowledge, this is the first {\em benchmark-tight} revenue guarantee of any simple multi-item mechanism.

Technically, all previous works employ {\MyersonAuction} as an intermediary. The barrier of $4$ follows as {\UniformIronedVirtualValueItemPricing} achieves a {\em tight} $2$-approximation to {\MyersonAuction}, which then achieves a {\em tight} $2$-approximation to {\DualityRelaxationBenchmark}. Instead, our new approach avoids {\MyersonAuction}, thus enabling the improvement. Central to our work are a {\em benchmark-based} $3$-competitive prophet inequality and its {\em fully constructive} proof. Such variant prophet inequalities shall find future applications, e.g., to Multi-Item Mechanism Design where optimal revenues are relaxed to various more accessible benchmarks.

We complement our benchmark-tight ratio with an impossibility result. All previous works and ours follow the {\em single-dimensional representative} approach introduced by [Chawla-Hartline-Kleinberg, EC'07]. Against {\DualityRelaxationBenchmark}, it turns out that this approach cannot beat our bound of $3$ for a large class of {\ItemPricing}'s.
\end{abstract}
\thispagestyle{empty}

\newpage
\setcounter{page}{1}

\section{Introduction}
\label{sec:intro}

The TCS community has tremendously advanced on Multi-Item Mechanism Design in the last two decades. Once it became clear that, even only with a single buyer, optimal multi-item mechanisms are intolerably complicated \cite{T04,BCKW15,HN19,DDT14,CDPSY18,CDOPSY22,CMPY18}, the community turned to examine the problem through the lenses of {\em approximation}.
Hitherto there is a long line of works on the design and analysis of {\em simple and approximately optimal} mechanisms \cite[and the references therein]{CHK07,CHMS10,CMS15,CD15,KW19,HN17,LY13,BILW20,Y15,HH15,RW18,CM16,CDW21,CZ17,Y18,JLQTX19,DFLSV22,MS21,DKL20,CC23,CCW23}.

\subsection{A paradigmatic model}
\label{subsec:model}

This rich literature dates back to the seminal work of Chawla, Hartline, and Kleinberg \cite{CHK07}, who studied a revenue-maximizing seller (she) with $n \geq 2$ items facing a {\em unit-demand} buyer (he) with item-wise values $v_{i}$ for $i \in [n]$ and bundle-wise values $v_{S} = \max_{i \in S} \{v_{i}\}$ for $S \subseteq [n]$.
Namely, item-wise values $v_{i} \sim F_{i}$ are {\em independently distributed}; the seller knows the priors $\bF = \{F_{i}\}_{i \in [n]}$ but does not know the outcomes $\bv = (v_{i})_{i \in [n]}$, thus aiming to maximize her expected revenue.

This revenue maximization problem can be formulated as an (infinite-dimensional) linear program\footnote{\label{footnote:LP_MIP}More precisely, this linear or mixed-integer program has a size polynomial in $\prod_{i \in [n]} |\supp(F_{i})|$, {\em the support size of the joint value distribution}. Thus, even for binary-supported distributions $|\supp(F_{i})| = 2$, it has a size exponential in the number of items $n \geq 2$.} and a feasible solution can be interpreted as follows: The seller offers the buyer a menu of lotteries $\{(\bx,\, \ell)\}$; each lottery includes its item-wise allocation probabilities $\bx = (x_{i})_{i \in [n]}$ and its price $\ell$.\footnote{Since the buyer is unit-demand, without loss of generality, a lottery $(\bx,\, \ell)$ can have a total allocation probability of at most one $\sum_{i \in [n]} x_{i} \leq 1$, i.e., exclusively allocating each item $i \in [n]$ with probability $x_{i}$.}
The buyer will purchase the utility-maximizing lottery $\argmax_{(\bx,\, \ell)} \{\bv \cdot \bx - \ell\}$, or nothing if no lottery induces a nonnegative expected utility.
The revenue-maximizing menu is called {\OptimalLotteryPricing} ({\OLP}).
Unfortunately, the associated computation problem is intractable \cite{CDOPSY22}:
First, $\Omega(2^{n})$ many lotteries can be necessary even for binary-supported distributions $|\supp(F_{i})| = 2$, hence an exponential {\em menu complexity} \cite{HN13}.
Second, no randomized polynomial-time algorithm can implement {\OptimalLotteryPricing}, unless $\P^{\sharpP} = \P^{\NP}$, even for ternary-supported distributions $|\supp(F_{i})| = 3$.

An equally fascinating problem, revenue maximization among deterministic mechanisms, can be formulated as an (infinite-dimensional) mixed-integer program\textsuperscript{\ref{footnote:LP_MIP}}, and a feasible solution can be interpreted as follows:
The seller posts item-wise prices $\bp = (p_{i})_{i \in [n]}$ and, then, the buyer purchases his favorite item $\argmax_{i \in [n]} \{v_{i} - p_{i}\}$, or nothing if every price $p_{i}$ is above value $v_{i}$.
The revenue-maximizing deterministic mechanism is called {\OptimalItemPricing} ({\OIP}).
Unfortunately, the associated computation problem is $\NP$-hard even for ternary-supported distributions $|\supp(F_{i})| = 3$, or even for i.i.d.\ distributions $\bF = \{F\}^{\otimes n}$ of support size $|\supp(F)| = \poly(n)$ \cite{CDPSY18}.

Given the above hardness results, it is of particular interest to investigate simple {\em deterministic polynomial-time} mechanisms and to analyze their revenue guarantees against {\OptimalItemPricing} and {\OptimalLotteryPricing}.\footnote{It is equally interesting to design and analyze simple {\em randomized} mechanisms. In this regard, however, the only positive progress is a quasi-polynomial time approximation scheme ($\QPTAS$) for {\OptimalLotteryPricing} \cite{KSMSW19}.}
Among all the simple mechanisms studied in the literature, only one unconditionally achieves a constant-factor revenue guarantee against either optimum.\footnote{But if we impose mild assumptions on the value distributions $\bF = \{F_{i}\}_{i \in [n]}$, more simple mechanisms also can achieve constant-factor revenue guarantees. E.g., if the value distributions satisfy the {\em regularity} condition, a standard condition from Myerson's seminal work \cite{M81}, {\OptimalItemPricing} admits a $\QPTAS$ \cite{CD15} and even the simplest mechanism \textsf{Uniform Pricing}, i.e., posting an optimal uniform price $p$ for all items, achieves a tight $\calC_{\OIP / {\sf UP}} \approx 2.6202$-approximation to {\OptimalItemPricing} and a $2\calC_{\OIP / {\sf UP}} \approx 5.2404$-approximation to {\OptimalLotteryPricing} \cite{JLQTX19}.}
It is called {\sf Uniform-Ironed-Virtual-Value Item Pricing}, or simply {\UIVVItemPricing} ({\UIVVIP}); cf.\ \Cref{fig:result} for a diagram of its revenue guarantees against either {\OptimalItemPricing} or {\OptimalLotteryPricing}, as well as other revenue gaps.

{\UIVVItemPricing} was proposed in Chawla, Hartline, and Kleinberg's original work \cite{CHK07}. In particular, they proved that it achieves a $3$-approximation to {\OptimalItemPricing}.
Afterward, Chawla, Hartline, Malec, and Sivan \cite{CHMS10} improved this upper bound to $2$, which remains the state of the art.

In a breakthrough, Chawla, Malec, and Sivan \cite{CMS15} (the conference version was in EC'10) established that {\UIVVItemPricing} achieves a $4$-approximation to {\OptimalLotteryPricing}.
Indeed, since the $\P^{\sharpP}$-hard {\OptimalLotteryPricing} is prohibitively difficult to deal with, \cite{CMS15} instead relaxed it to a more accessible benchmark, {\DualityRelaxationBenchmark} ({\DRB}).\footnote{\cite{CMS15} did not name this benchmark by {\DualityRelaxationBenchmark}. Instead, the follow-ups \cite{HH15,CDW21} gave it a {\em duality-based} new interpretation. Thus, without ambiguity, we call it {\DualityRelaxationBenchmark}.}
I.e., they proved that this benchmark revenue-surpasses {\OptimalLotteryPricing}  and that {\UIVVItemPricing} even ensures a $4$-approximation to it.
The follow-up by Cai, Devanur, and Weinberg \cite{CDW21} presented an alternative and possibly more ``universal'' proof, including a {\em duality-based} new interpretation for the benchmark.
Nonetheless, for nearly fifteen years, the original bound of $4$ from \cite{CMS15} remains state of the art and {\DualityRelaxationBenchmark} remains the only known useful relaxation of {\OptimalLotteryPricing}.\footnote{More precisely, \cite{EFFTW17a,BW19} considered another duality-based benchmark (slightly different from the one by \cite{CMS15,HH15,CDW21}). Nonetheless, {\UIVVItemPricing} turns out to have a worse or the same revenue guarantee $\geq 3$ against the \cite{EFFTW17a,BW19} benchmark, which means that benchmark is not so useful for our purpose. (Our lower-bound instances (\Cref{exp:DRB_SPP}) are {\em symmetric instances}, by which both benchmarks have the same revenue, so the \cite{EFFTW17a,BW19} benchmark can only establish a worse or the same revenue guarantee $\geq 3$.)}

Exploring ``{\OptimalLotteryPricing} vs.\ {\OptimalItemPricing}'' is also interesting, i.e., the gap between randomized and deterministic revenues, regardless of computational constraints. Unfortunately, mainly because both mechanisms are computationally intractable, nothing is known except an upper bound of $4$ (by implication) and a best-known lower bound of $1 + \frac{\ln 2}{5} \approx 1.1386$ \cite{CMS15}.

\subsection{Our contributions}
\label{subsec:contribution}

\begin{figure}[t]
\centering
\begin{tikzpicture}
    \node(n0) at (0, 3) [draw, thick] {{\DualityRelaxationBenchmark}};
    \node(n1) at (0, 0) [draw, thick] {{\OptimalLotteryPricing}};
    \node(n2) at (0, -3) [draw, thick] {{\OptimalItemPricing}};
    \node(n3) at (0, -6) [draw, thick] {{\UIVVItemPricing}};
    
    \draw[thick, dashed, ->, >=triangle 45] (n1.north) to (n0.south);
    
    \draw[thick, ->, >=triangle 45] (n3.north) to (n2.south);
    \node at (0, -4.5) [right, yshift = .6cm] {$\leq 3$ \cite{CHK07}};
    \node at (0, -4.5) [right] {$\leq 2$ \cite{CHMS10}};
    \node at (0, -4.5) [right, yshift = -.6cm] {$\geq 2$ [\Cref{thm:UIVVIP}] \OliveGreen{(tight)}};

    \draw[thick, dashed, ->, >=triangle 45] (n3.south east) to (9.5, -6 - 0.315) to
    (9.5, 3 + 0.315) to (n0.north east);
    \node at (9.5, -1.5) [right, yshift = .3cm] {$\leq 4$ \cite{CMS15,CDW21}};
    \node at (9.5, -1.5) [right, yshift = -.3cm] {$= 3$ [\Cref{thm:UIVVIP}] \OliveGreen{(tight)}};
    
    \draw[thick, ->, >=triangle 45] (n3.north east) to (9, -6 + 0.315) to (9, 0) to (n1.east);
    \node at (5, -6) [right, yshift = .3cm + .315cm] {$\geq 2$ [\Cref{thm:UIVVIP}]};
    \node at (5, -6) [right, yshift = .9cm + .315cm] {$\leq 3$ [\Cref{thm:UIVVIP}]};
    \node at (5, -6) [right, yshift = 1.5cm + .315cm] {$\leq 4$ \cite{CMS15,CDW21}};
    
    \draw[thick, ->, >=triangle 45] (n2.north) to (n1.south);
    \node at (0, -1.5) [right, yshift = -.6cm] {$\geq 1.1386$ \cite{CMS15}};
    \node at (0, -1.5) [right] {$\leq 3$ [\Cref{cor:OIP}]};
    \node at (0, -1.5) [right, yshift = .6cm] {$\leq 4$ \cite{CMS15,CDW21}};
    
    \draw[thick, dashed, ->, >=triangle 45] (n2.east) to (8.5, -3) to (8.5, 3 - 0.315) to (n0.south east);
    \node at (4.5, -3) [right, yshift = .3cm] {$\geq 2$ [\Cref{cor:OIP}]};
    \node at (4.5, -3) [right, yshift = .9cm] {$\leq 3$ [\Cref{cor:OIP}]};
    \node at (4.5, -3) [right, yshift = 1.5cm] {$\leq 4$ \cite{CMS15,CDW21}};
\end{tikzpicture}
\caption{A diagram of the previous results and our new results. The four mechanisms constitute a revenue hierarchy: {\DualityRelaxationBenchmark} $\succeq$ {\OptimalLotteryPricing} $\succeq$ {\OptimalItemPricing} $\succeq$ {\UIVVItemPricing}. ({\DualityRelaxationBenchmark} itself may lack a concrete economic meaning but remains the only known useful relaxation of {\OptimalLotteryPricing}.)
One solid arrow refers to the revenue gap between two concrete mechanisms, and one dashed arrow refers to the revenue gap between one concrete mechanism and {\DualityRelaxationBenchmark}. \\
We emphasize that, mainly because the $\P^{\sharpP}$-hard {\OptimalLotteryPricing} and the $\NP$-hard {\OptimalItemPricing} both are sophisticated, actual positive progress has been made  only for ``{\DualityRelaxationBenchmark} (resp.\ {\OptimalItemPricing}) vs.\ {\UIVVItemPricing}''. \\
For the other revenue gaps, all the previous upper bounds and our new upper bounds are implications from  ``{\DualityRelaxationBenchmark} vs.\ {\UIVVItemPricing}''.
\label{fig:result}}
\end{figure}
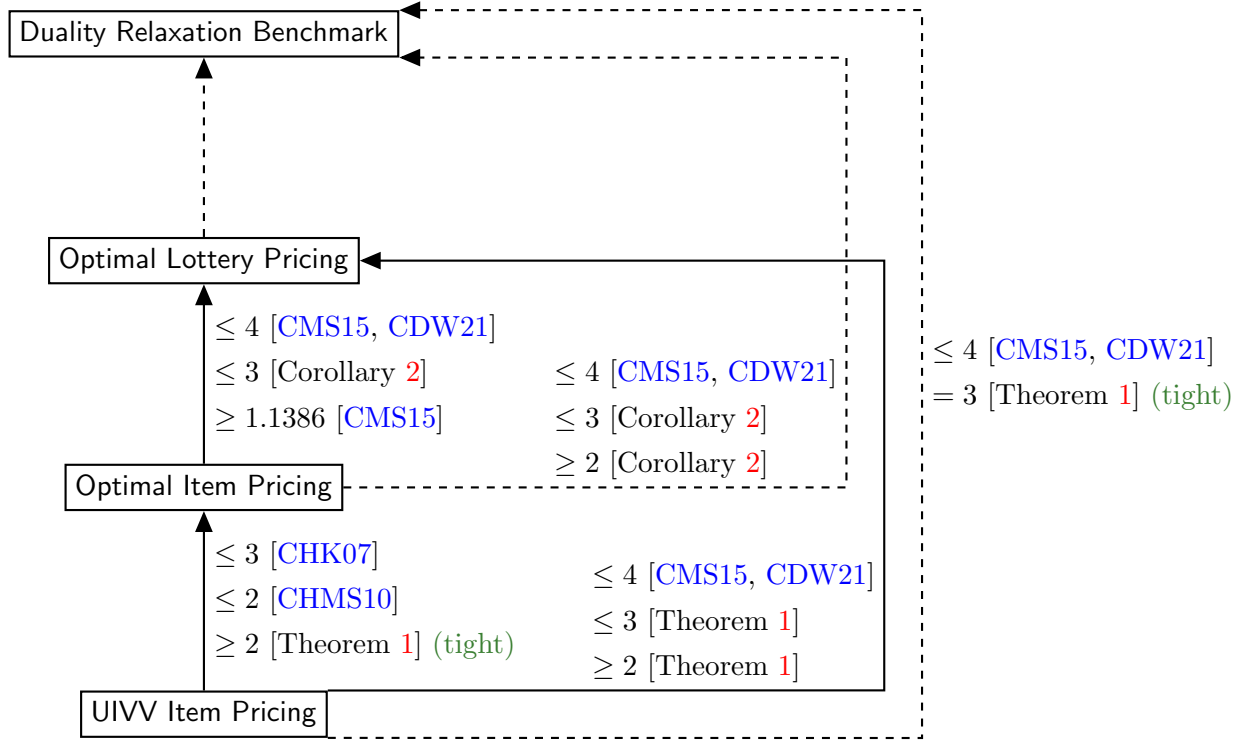

In this work, we show the {\em tight} approximation ratio $\calC_{\DRB / \UIVVIP} = 3$ of {\UIVVItemPricing} to {\DualityRelaxationBenchmark}.
To our knowledge, this is the first benchmark-tight approximation ratio of any simple multi-item mechanism -- within or beyond the unit-demand single-buyer setting.

The worst-case instances for $\calC_{\DRB / \UIVVIP} = 3$ further inspire us to improve the lower bounds on the approximation ratios of {\UIVVItemPricing} against either {\OptimalLotteryPricing} or {\OptimalItemPricing}.
Namely, we show a family of $n$-item instances by which $\calC_{\OLP / \UIVVIP} \geq \calC_{\OIP / \UIVVIP} \geq \frac{2}{1 + 1 / n}$. This lower bound can be arbitrarily close to $2$, when the number of items is large enough $n \gg 2$.
Also, it matches the upper bound $\calC_{\OIP / \UIVVIP} \leq 2$ from \cite{CHMS10}, thus closing this approximation ratio $\calC_{\OIP / \UIVVIP} = 2$.

The following theorem summarizes our improved upper and lower bounds on the revenue guarantees of {\UIVVItemPricing}; cf.\ \Cref{fig:result} for a diagram.
(The upper-bound part of Item~(iii) is credited to \cite{CHK07,CHMS10}.)

\begin{theorem}[Revenue Guarantees of {\UIVVIP}]
\label{thm:UIVVIP}
\begin{flushleft}
{\UniformIronedVirtualValueItemPricing} achieves \\
(i)~a tight $\calC_{\DRB / \UIVVIP} = 3$ approximation to {\DualityRelaxationBenchmark}, \\
(ii)~a $\calC_{\OLP / \UIVVIP} \in [2,\, 3]$ approximation to {\OptimalLotteryPricing}, and \\
(iii)~a tight $\calC_{\OIP / \UIVVIP} = 2$ approximation to {\OptimalItemPricing}.
\end{flushleft}
\end{theorem}

By implication, we also improve the upper bounds on the revenue guarantees of {\OptimalItemPricing} against either {\DualityRelaxationBenchmark} or {\OptimalLotteryPricing}, both from $4$ to $3$.
(The previous upper bounds also are implications of $\calC_{\DRB / \UIVVIP} \leq 4$ \cite{CHK07,CHMS10,CMS15,CDW21}.)
In this regard, we construct another $2$-approximation lower-bound instance for ``{\DualityRelaxationBenchmark} vs.\ {\OptimalItemPricing}''.

The following corollary concludes our improved upper and lower bounds on the revenue guarantees of {\OptimalItemPricing}; cf.\ \Cref{fig:result} for a diagram. (The lower-bound part of Item~(ii) is credited to \cite{CMS15}.)

\begin{corollary}[Revenue Guarantees of {\OIP}]
\label{cor:OIP}
\begin{flushleft}
{\OptimalItemPricing} achieves \\
(i)~a $\calC_{\DRB / \OIP} \in [2,\, 3]$ approximation to {\DualityRelaxationBenchmark} and \\
(ii)~a $\calC_{\OLP / \OIP} \in [1 + \frac{\ln 2}{5} \approx 1.1386,\, 3]$ approximation to {\OptimalLotteryPricing}.
\end{flushleft}
\end{corollary}

\noindent
{\bf Limitations on Further Improvements.}
Given our improved upper bounds of $3$, one might be optimistic with better revenue guarantees of {\UIVVItemPricing}, {\OptimalItemPricing}, or other simple deterministic mechanisms in between.
Unfortunately, en route to \Cref{thm:UIVVIP}, we also show that beating our bound of $3$ transcends the scope of known methods;
see \Cref{subsec:overview} for details.
Roughly speaking, this is because both the previous works \cite{CHK07,CHMS10,CMS15,CDW21} and this paper (possibly implicitly) apply the {\em single-dimensional representative} approach initiated by \cite{CHK07}.
We show that, against {\DualityRelaxationBenchmark}, this approach cannot establish an approximation ratio better than $3$ for any {\em ironed-virtual-value-based} {\ItemPricing}.
We note that all the mentioned works (and many others beyond the unit-demand single-buyer setting) restricted their attention to this family of mechanisms.
On the other hand, direct analysis of the $\P^{\sharpP}$-hard {\OptimalLotteryPricing} seems out of reach.

\subsection{Technical overview}
\label{subsec:overview}

Below, we sketch our new approach for the tight revenue guarantee $\calC_{\DRB / \UIVVIP} = 3$ of {\UIVVItemPricing} against {\DualityRelaxationBenchmark}, assuming that the reader has basic familiarity with the previous works, especially Myerson's virtual value theory \cite{M81}.

To begin with, recall that the previous approach \cite{CHK07,CHMS10,CMS15,CDW21} for the weaker bound $\calC_{\DRB / \UIVVIP} \leq 4$ takes three steps.\footnote{Of course, if we consider {\OptimalItemPricing} and/or {\OptimalLotteryPricing} instead, there are two other steps: \\
$\text{\OptimalItemPricing} \geq \text{\UIVVItemPricing}$ (obvious) and
$\frac{1}{4} \cdot \text{\DualityRelaxationBenchmark} \geq \frac{1}{4} \cdot \text{\OptimalLotteryPricing}$ (proved in \cite{CMS15,CDW21}).
Both equalities simultaneously hold for (say) the first instance in \Cref{footnote:tight_instance}.}
\begin{align}
    {\UIVVItemPricing}
    & ~\geq~ {\UIVVSequentialPostedPricing}
    \tag{1}\label{eq:step_1} \\
    & \phantom{~\geq~} \text{\OliveGreen{$\triangleright$ the single-dimensional representative approach (tight)}}
    \notag \\
    & ~\geq~ \tfrac{1}{2} \cdot {\MyersonAuction}
    \tag{2}\label{eq:step_2} \\
    & \phantom{~\geq~} \text{\OliveGreen{$\triangleright$ the order-oblivious prophet inequalities (tight)}}
    \notag \\
    & ~\geq~ \tfrac{1}{4} \cdot {\DualityRelaxationBenchmark}.
    \tag{3}\label{eq:step_3} \\
    & \phantom{~\geq~} \text{\OliveGreen{$\triangleright$ benchmark decomposition (tight)}}
    \notag
\end{align}

Step~\eqref{eq:step_1} applies the single-dimensional representative approach \cite{CHK07,CHMS10}, which relates the (less understood) unit-demand single-buyer setting to the (better understood) single-item multi-buyer setting.
The former single buyer's value and distribution $v_{i} \sim F_{i}$ for each former item $i \in [n]$ is reinterpreted as each latter buyer $i$'s value and distribution $v_{i} \sim F_{i}$ for the latter single item.
Therefore, {\UIVVItemPricing} shall be replaced with its single-item multi-buyer counterpart, {\UIVVSequentialPostedPricing} ({\UIVVSPP}) in the {\em order-oblivious} model (elaborated soon).

Step~\eqref{eq:step_2}, due to well-known reductions \cite{CHMS10,CFPV19} based on Myerson's virtual value theory, is equivalent to the order-oblivious prophet inequalities \cite{KS78,S84,KW19}.

Step~\eqref{eq:step_3} is evident once we write down the revenue formula of {\DualityRelaxationBenchmark} (\Cref{subsec:mechanism}).
It is basically a ``mixture'' of the revenues from {\MyersonAuction} ({\MA}) and {\SecondPriceAuction} ({\SPA}), thus being upper-bounded by the latter two's total revenue.

We emphasize that all individual steps are tight, i.e., the equalities hold for some step-specific instances,\footnote{\label{footnote:tight_instance}Here are three ``tight'' instances, respectively; for brevity, we omit calculation of revenues (elementary algebra). \\
The equality in Step~\eqref{eq:step_1}: $F_{1}(v) = F_{2}(v) = \dots = F_{n}(v) \eqdef \indicator(v \geq 1)$, i.e., all values are deterministic $\equiv 1$. \\
The equality in Step~\eqref{eq:step_2}: $F_{1}(v) \eqdef \indicator(v \geq 1)$ and $F_{2}(v) \eqdef 1 - 1 / v$ for $v \geq 1$ (a.k.a.\ the equal-revenue distribution). \\
The equality in Step~\eqref{eq:step_3}: $F_{1}(v) = F_{2}(v) = \dots = F_{n}(v) \eqdef 1 - 1 / v$ for $v \geq 1$ and $n \geq 2$ (see \Cref{exp:DRB_OIP} in \Cref{sec:example}).} which may explain why the bound of $4$ is long-standing.
In contrast, our approach will enable the improvement by unifying and then refining Steps~\eqref{eq:step_2} and \eqref{eq:step_3}.

\vspace{.1in}
\noindent
{\bf Mechanisms and Their Relation.}
First, we shall formalize {\ItemPricing} and {\SequentialPostedPricing} and recap how to relate them via the single-dimensional representative approach.

\begin{definition}[{\ItemPricing}]
\label{def:IP}
\begin{flushleft}
For a unit-demand single-buyer instance $\bF = \{F_{i}\}_{i \in [n]}$,
\begin{itemize}
 \item {\ItemPricing}: The seller posts item-wise prices $\bp = (p_{i})_{i \in [n]}$ on items $i \in [n]$. Then the buyer purchases the utility-maximizing item $\argmax_{i \in [n]} \{v_{i} - p_{i}\}$, breaking ties in favor of (one of) the highest-priced item(s),\footnote{As shown in \cite[Section~2.2]{CDPSY18}, this tie-breaking rule is without loss of generality.} or nothing if every price is above the value $p_{i} > v_{i}$, $\forall i \in [n]$. \\
 Denote by $\IP(\bF,\, \bp)$ the resulting revenue.
\end{itemize}
\end{flushleft}
\end{definition}

\begin{definition}[{\SequentialPostedPricing}]
\label{def:SPP}
\begin{flushleft}
For a single-item multi-buyer instance $\bF = \{F_{i}\}_{i \in [n]}$,
\begin{itemize}
 \item {\SequentialPostedPricing}: The seller posts buyer-wise prices $\bp = \{p_{i}\}_{i \in [n]}$ for buyers $i \in [n]$.
 Then the buyers arrive one by one $\pi(1),\, \pi(2),\, \dots,\, \pi(n)$ in a specific order $\pi \in \Pi_{n}$; the first arrival buyer $\pi(i)$ who accepts her price $v_{\pi(i)} \geq p_{\pi(i)}$ gets the item. \\
 Denote by $\SPP(\bF,\, \bp,\, \pi)$ the resulting revenue.
 
 \item The {\em order-oblivious} model:
 Once the seller determines the buyer-wise prices $\bp = \{p_{i}\}_{i \in [n]}$, \\
 a hypothetical adversary chooses (one of) the revenue-worst arrival order(s) $\pi^{*} \in \Pi_{n}$ for the seller.\ignore{ $\pi^{*} \eqdef \argmin_{\pi \in \Pi_{n}} \{\SPP(\bF,\, \bp,\, \pi)\}$.}
 Denote by $\SPP^{*}(\bF,\, \bp) \eqdef \min_{\pi \in \Pi_{n}} \{\SPP(\bF,\, \bp,\, \pi)\}$ the resulting revenue.\ignore{$\SPP^{*}(\bF,\, \bp) \eqdef \SPP(\bF,\, \bp,\, \pi^{*})$ the resulting revenue.}
\end{itemize}
\end{flushleft}
\end{definition}

The single-dimensional representative approach \cite[Theorem~4]{CHMS10} asserts that, under the same {\em deterministic} prices $\bp$, {\ItemPricing} revenue-surpasses {\SequentialPostedPricing} in the order-oblivious model.
(We note that, to validate \Cref{prop:representative} as a black-box reduction, the involved prices $\bp$ must be deterministic; see \Cref{sec:randomness} for more details. The previous work \cite{CHMS10} did not emphasize this issue.)

\begin{proposition}[{The Single-Dimensional Representative Approach \cite{CHMS10}}]
\label{prop:representative}
\begin{flushleft}
\, \\
Given the same instance $\bF$ and the same \textbf{deterministic} prices $\bp$,
{\ItemPricing} revenue-surpasses {\SequentialPostedPricing} in the order-oblivious model $\IP(\bF,\, \bp) \geq \SPP^{*}(\bF,\, \bp)\ignore{\min_{\pi \in \Pi_{n}} \{\SPP(\bF,\, \bp,\, \pi)\}}$.
\end{flushleft}
\end{proposition}

Specifically, \cite{CHK07,CHMS10} designed the uniform-ironed-virtual-value prices $\bp$ as follows.

\begin{definition}[Uniform-Ironed-Virtual-Value Prices]
\label{def:UIVV_prices}
\begin{flushleft}
For a single-item multi-buyer instance $\bF = \{F_{i}\}_{i \in [n]}$, let $\bar{\bvarphi} = \{\bar{\varphi}_{i}\}_{i \in [n]}$ be (\Cref{subsec:distribution}) the {\em increasing} ironed virtual value functions.
Given a uniform ironed-virtual-value threshold ({\UIVV}-threshold) $T > 0$, define the item-wise prices as $p_{i} \eqdef \sup \{ v \in \RR \mid \bar{\varphi}_{i}(v) < T \} = \inf \{ v \in \RR \mid \bar{\varphi}_{i}(v) \geq T \}$ for each $i \in [n]$,\footnote{We can alternatively define the prices as $p_{i} \eqdef \sup \{ v \in \RR \mid \bar{\varphi}_{i}(v) \leq T \} = \inf \{ v \in \RR \mid \bar{\varphi}_{i}(v) > T \}$. Both definitions will induce different prices $\bp = \{p_{i}\}_{i \in [n]}$ but the same bounds $\calC_{\MA / \UIVVIP} = 2$ and $\calC_{\DRB / \UIVVIP} = 3$.} \\
where the equality holds since each $\bar{\varphi}_{i}$ is an increasing function.
\end{flushleft}
\end{definition}

Now we are ready to elaborate on how to unify and refine Steps~\eqref{eq:step_2} and \eqref{eq:step_3}, thus obtaining our benchmark-tight revenue guarantee $\calC_{\DRB / \UIVVIP} = 3$.

\vspace{.1in}
\noindent
{\bf Unifying Steps~\eqref{eq:step_2} and \eqref{eq:step_3}.}
To clarify the differences between our approach and the previous works \cite{CHK07,CHMS10}, let us recap the previous proof of Step~\eqref{eq:step_2}, as follows.

\begin{theorem}[Revenue Guarantees of {\UIVVSPP} \cite{CHK07,CHMS10}]
\label{thm:MA_SPP}
\begin{flushleft}
Against {\MyersonAuction}: \\
(i)~{\UniformIronedVirtualValueSequentialPostedPricing} with a {\UIVV}-threshold $T = \frac{1}{2} \cdot \MA$ of one half of the {\MyersonAuction} revenue achieves a $\calC_{\MA / \UIVVSPP} = 2$ approximation. \\
(ii)~This ratio is optimal for \textbf{all} stopping rules in the order-oblivious model.
\end{flushleft}
\end{theorem}

\begin{proof}[Proof (The Upper-Bound Part)]
Given a specific threshold $T > 0$ and an arbitrary arrival order $\pi \in \Pi_{n}$ for {\SequentialPostedPricing}, the prices $\bp = \{p_{i}\}_{i \in [n]}$ from \Cref{def:UIVV_prices} satisfy that ``each buyer $i \in [n]$ gets allocated with the same probability, under any two values $v_{i}$ and $\tilde{v}_{i}$ with the same ironed virtual value $\bar{\varphi}_{i}(v_{i}) = \bar{\varphi}_{i}(\tilde{v}_{i})$.''
Thus, we can apply the revenue equivalence \cite{M81,CDW21} (see \Cref{prop:revenue_equivalence} for a formal statement).
I.e., the {\UIVVSequentialPostedPricing} revenue $=$ its ironed virtual welfare, and its revenue guarantee reduces to the order-oblivious prophet inequalities \cite{KS78,S84,KW19}.
More precisely, a threshold of $T = \frac{1}{2} \cdot \MA$ is well-known to $2$-approximate the optimal ironed virtual welfare $=$ the {\MyersonAuction} revenue.
\end{proof}

\Cref{thm:MA_SPP} in combination with ``$\MyersonAuction \geq \frac{1}{2} \cdot \DualityRelaxationBenchmark$'' \cite{CMS15,CDW21} gives the previous revenue guarantee $\calC_{\DRB / \UIVVSPP} \leq 4$ of {\UIVVSequentialPostedPricing} against {\DualityRelaxationBenchmark}.

An astute reader may already note that, for our sake of the revenue gap $\calC_{\DRB / \UIVVSPP}$ between {\DualityRelaxationBenchmark} and {\UIVVSequentialPostedPricing}, the ``right'' approach shall avoid the intermediary, {\MyersonAuction}, and instead attack this revenue gap directly.
This is precisely our main technical ingredient, which we call {\em benchmark-based} prophet inequalities.

\begin{theorem}[Revenue Guarantees of {\UIVVSPP}]
\label{thm:DRB_SPP}
\begin{flushleft}
Against {\DualityRelaxationBenchmark}: \\
(i)~{\UniformIronedVirtualValueSequentialPostedPricing} with a {\UIVV}-threshold $T = \frac{1}{3} \cdot \DRB$ of one third of the {\DualityRelaxationBenchmark} revenue achieves a $\calC_{\DRB / \UIVVSPP} = 3$ approximation. \\
(ii)~This ratio is optimal for \textbf{deterministic ironed-virtual-value-based} and/or \textbf{(arbitrary) uniform-ironed-virtual-value} stopping rules in the order-oblivious model.\footnote{\label{footnote:tie}By a \textbf{deterministic ironed-virtual-value-based} stopping rule, we mean: First, the seller chooses a deterministic set of {\em acceptable} ironed virtual values $\Phi_{i} \subseteq \R$ for every buyer $i \in [n]$, i.e., different possible values $v_{i} \sim F_{i}$ with the same ironed virtual value either all are acceptable $\bar{\varphi}_{i}(v_{i}) \in \Phi_{i}$ or all are unacceptable $\bar{\varphi}_{i}(v_{i}) \notin \Phi_{i}$. (So a buyer $i \in [n]$ on his arrival takes the item if and only if $\bar{\varphi}_{i}(v_{i}) \in \Phi_{i}$ and no earlier buyer had taken the item.) Then, the adversary chooses a worst-case arrival order $\pi^{*} \in \Pi_{n}$ against this stopping rule, i.e., the acceptance sets $\{\Phi_{i}\}_{i \in [n]}$. \\
By an \textbf{arbitrary uniform-ironed-virtual-value}, we mean: First, the seller chooses an arbitrary ironed-virtual-value threshold $T$ (which can be randomized) and posted prices $\bp = (p_{i})_{i \in [n]}$ such that $\bar{\varphi}_{i}(p_{i}) = T$ (which can be randomized even conditioned on $T$). Then, the adversary chooses a worst-case arrival order $\pi^{*} \in \Pi_{n}$ against the distributional information of $\bp$.}
\end{flushleft}
\end{theorem}

As mentioned, {\DualityRelaxationBenchmark} is a mixture of the revenues from {\MyersonAuction} and {\SecondPriceAuction}. This means, in our approach to \Cref{thm:DRB_SPP}, we must frequently switch between the value space and the ironed virtual value space.
To prove \Cref{thm:MA_SPP}, in contrast, the revenue equivalence allows us to concentrate solely on the ironed virtual value space and thus apply the classical prophet inequalities.
This difference incurs many technical difficulties (and may explain why the previous works \cite{CMS15,CDW21} did not investigate this revenue gap directly but leveraged {\MyersonAuction} as an intermediary).

Nonetheless, we develop many new techniques here and give a {\em fully constructive} proof of \Cref{thm:DRB_SPP}.
More concretely, we will prove more and more necessary conditions for the worst-case instances, thus gradually shrinking the optimization space and finally deriving the exact worst-case instances and the tight bound $\calC_{\DRB / \UIVVSPP} = 3$.
Such a fully constructive proof might be more structure-revealing than the nonconstructive proofs of the classical prophet inequalities.\footnote{I.e., \cite{KS78,S84,KW19} set the threshold as $T = \frac{1}{2} \cdot \MA$ (say) constructively, but proved the upper bound $\calC \leq 2$ non-constructively using the probabilistic method and gave the matching lower-bound instances separately.}

Such constructive proofs have made many success stories in Single-Item Mechanism Design. There is a long line of works on proving tight revenue/welfare guarantees of both truthful and nontruthful single-item mechanisms \cite{CGL14,AHNPY19,CFHOV21,AB20,JLTX20,JLQTX19,JLQ19,HJL20,JJLZ22,CST18,BMTT19,JL22,JL23}.
To our knowledge, however, our work is the first to adopt such techniques to Multi-Item Mechanism Design. 
In this regard, we look forward to more future works that technically bridge these two highly relevant areas.

Given \Cref{prop:representative} and the lower-bound part of \Cref{thm:DRB_SPP}, we can also conclude that the single-dimensional representative approach cannot establish a bound better than $3$ for any {\em ironed-virtual-value-based} {\ItemPricing}.
(As mentioned, all the previous works \cite{CHK07,CHMS10,CMS15,CDW21} and many others beyond the unit-demand single-buyer setting restricted their attention to this family of mechanisms.)

Further, we remark that the {\DualityRelaxationBenchmark} considered here actually is the instantiation of the benchmark from \cite{CMS15,HH15,CDW21} to the unit-demand single-buyer setting -- the original benchmark accommodates much more general settings.
It is conceivable that our main technical ingredient -- the benchmark-based prophet inequality (\Cref{thm:DRB_SPP}) -- also leaves enough room for future generalizations, which shall give improved or even benchmark-tight approximation ratios of other simple multi-item mechanisms.

\subsection{Proof overview}
In this section, we provide an outline of the proof for our main theorem.

One of our first key observations is that the {\DualityRelaxationBenchmark} for a single unit-demand buyer with independent values is monotone with respect to the stochastic dominance of value distributions; see \Cref{lem:revenue_monotonicity,fig:revenue_monotonicity} for more details. This property turns out to be very useful in our proof.
Actually, one of the reasons why Multi-Item Mechanism Design is much more difficult than Single-Item Mechanism Design is that the optimal mechanisms (in general) no longer satisfy revenue monotonicity \cite{HR15}. It is very difficult to compare to a benchmark that itself is not even monotone.
The proof for the monotonicity of {\DualityRelaxationBenchmark} is not that difficult, but also nontrivial since it is not point-wise dominance. To the best of our knowledge, this monotonicity has not been pointed out in the literature before.\footnote{Rather, approximate versions of the revenue monotonicity of (optimal) multi-item mechanisms have been established in \cite{BILW20,RW18,Y18,CCW23}.} We believe that this observation may find its applications in other problems.

Based on this revenue monotonicity, we can do a few transformations of the given instance to get instances with a worse revenue gap.
Our first reduction {\iron} that transforms a given instance to a regular instance is standard and natural, as {\UIVVSequentialPostedPricing} uses \emph{ironed} virtual values in the first place. After that, we simply assume that all value distributions are regular.
The next two transformations are crucial and also quite intuitive. They preserve the revenue of {\UIVVSequentialPostedPricing} ({\UIVVSPP}) but increase the revenue of  {\DualityRelaxationBenchmark} by its revenue monotonicity, thus inducing a larger revenue gap after the transformation. Since the revenue of {\UIVVSPP} only depends on the probability that one's value is below or above the given price but not on how much, we can modify the instance so that their virtual values are either slightly above or slightly below the given threshold. In particular, the {\truncate} reduction makes the not-buying part (below the price) only slightly below the given threshold. This will not change the revenue of {\UIVVSPP} but will increase the {\DualityRelaxationBenchmark} since the distribution after the modification stochastically dominates the original one. The {\extend} reduction modifies the buying part (above the price) so that the virtual value is only slightly above the threshold for most probability and at a very high value ($\to \infty$) for a very small probability. Again, this will not change the revenue of {\UIVVSPP}. The fact that the distribution after the modification dominates the previous one seems less obvious than the {\truncate} reduction if one just reads the above description. However, it {\em is} obvious once one looks at the revenue-quantile curve before and after the modification, as the latter revenue-quantile curve is point-wise above the original one.
Such instances after the modification are called {\em semi-linear} instances (\Cref{def:extend}) in our paper, as their revenue-quantile curves are almost linear.

We note that the revenue-quantile curve is the most convenient representation for a value distribution in our proof, and we use it extensively \cite[Chapter~3]{H13}. This is because we need to go back and forth between the value space and the virtual value space; fortunately, both values and virtual values have intuitive geometry representation in revenue-quantile curves. If we use value CDF (say) instead, it is inconvenient and less intuitive to represent virtual values.
Another notational advantage of using revenue-quantile curves rather than value distributions directly is that we can implicitly introduce the value of $\infty$, which is both mathematically sound and notational simple. This is achieved by allowing $R(0) > 0$, i.e., we have strictly positive revenue by a sale probability of $0$, which means that the price (and the value) is $\infty$.
However, this convention may cause some confusion if the reader is less familiar with it. Thus, let us point it out in this outline to avoid potential confusion since we use it very often in our proof. If one does not like infinite values, this can always be interpreted as a sufficiently high value $H \to \infty$ with a sufficiently small probability $R(0) / H \to 0^{+}$. Once the reader gets familiar with this, it is convenient to use revenue-quantile curves with $R(0) > 0$.

These transforms and the above arguments omit a subtlety; they assume that the threshold remains the same after the modifications of the instances.
This is inaccurate. Actually, the threshold is selected after the instance is given, which is set to be $T = \frac{1}{3} \cdot \DRB$ in our final {\SequentialPostedPricing}. To address this subtlety, we introduced the {\scale} transform. This reduction is controlled by a continuous parameter
$\gamma \in [0,1]$ which will reduce the benchmark revenue and the {\UIVVSPP} revenue simultaneously and continuously. But by the mean value theorem, we prove that there exist a $\gamma^* \in [0,1]$ such that the instance after the {\scale} reduction has the same benchmark revenue  as the most original instance, thus the same threshold $T = \frac{1}{3} \cdot \DRB$. This reduction remedies the mentioned issue. Thus, compared with the most original instance, the new instance (after {\scale}) has the same benchmark revenue but a lower {\UIVVSPP} revenue, thus a larger revenue gap.

Our last reduction {\perturb} transforms a semi-linear instance into a {\em linear instance} (\Cref{def:perturb}). To achieve this, we can take a different but equivalent viewpoint of the above semi-linear instances. Namely, their virtual values are exactly equal to the threshold, rather than slightly above or slightly below it. The take-it-or-leave-it decisions in {\UIVVSequentialPostedPricing} are determined by a tie-breaking rule that yields the same posted prices and take-it-or-leave-it decisions, hence the same revenue as before. From this viewpoint, it is already a linear instance. These two different viewpoints do not change the benchmark revenue, as the tie-breaking rule only affects the {\UIVVSPP} revenue. The key observation here is to consider two extreme-case tie-breaking rules, namely ``always take it'' vs.\ ``always leave it'' in ties.
Regarding the {\UIVVSPP} revenues, the worse one between these two tie-breaking rules {\em is} the worst one among all possible tie-breaks. Thus, it is sufficient to only consider these two extreme cases.
For these two families of (extreme-case) linear instances, we can easily obtain the formulae for both the benchmark revenues and the {\UIVVSPP} revenues. Thus, we can derive an upper bound on the revenue gap $\calC_{\DRB/\UIVVSPP}$ after optimizing the parameters.

We notice that, for revenue maximization in {\em uniform pricing} mechanisms, breaking ties in favor of buying (trivially) is always good. But this is not the case for {\UIVVItemPricing}, since the same virtual value may correspond to many different values/prices. (Our linear instances are typical of this phenomenon.) The so-called ``tie-breaking rule'' throughout our proof is not an actual tie-breaking rule. Rather, it refers to different choices (as the {\UIVV} posted prices) of different values that have the same ironed virtual value. It is unclear which choice, ``higher prices with lower sale probabilities'' vs.\ ``lower prices with higher sale probabilities'', is better for revenue maximization.
To our knowledge, this subtlety about different choices of {\UIVV} posted prices (or, generally, {\em ironed-value-value-based} posted prices) has not been pointed out in the literature.

Whether this subtlety has effects on revenue maximization is an important question to study.
For the {\em upper bounds}, we remark that all mentioned revenue guarantees hold for any choice of {\UIVV} posted prices (i.e., even the worst-case {\UIVV} posted prices).
For the {\em lower bounds}, however, this subtlety seems to influence {\ItemPricing} more than {\SequentialPostedPricing}.
(i)~The tight revenue guarantee of $3$ for {\SequentialPostedPricing} is quite robust to different choices of {\UIVV} posted prices; see the statement of \Cref{thm:DRB_SPP}, \Cref{footnote:tie}, and \Cref{subsec:SPP_LB} for more details.
(ii)~The revenue guarantee for {\ItemPricing} is trickier.
On the one hand, if we insist on the ``naive'' tie-breaking rule in \Cref{def:UIVV_prices} for {\UIVV} posted prices, namely the smallest/largest values that correspond to the given {\UIVV} threshold $T$, then the bound of $3$ {\em is} tight. (Also, if we follow the single-dimensional representative approach, this bound of $3$ is the best possible.)
On the other hand, if we instead adopt (say) the ``best-{\UIVV}'' posted prices, then we cannot find a matching lower-bound instance.
In this regard, it is interesting for future works to study the revenue guarantees of the ``best-{\UIVV}'' {\ItemPricing}; see \Cref{subsec:SPP_LB} for more details. 







\subsection{Further related works}
\label{subsec:related_works}

Beyond the paradigmatic unit-demand single-buyer setting, Multi-Item Mechanism Design has been systematically studied in other parallel and more general settings, which can be categorized as the single- vs.\ multi-buyer settings, and based on the underlying valuation families:
\begin{align*}
    \text{unit-demand} ~\&~
    \text{additive} \subseteq
    \text{matroid-rank} \subseteq
    \text{constrained-additive} \subseteq
    \text{XOS} \subseteq
    \text{subadditive}.
\end{align*}
See, e.g., \cite{CZ17} for their formal definitions.
In particular, ``unit-demand'' and ``additive'' are the two most basic valuation families; one is neither a subset nor a superset of the other.

This flexibility in modeling has made Multi-Item Mechanism Design a centerpiece in the intersection of Theoretical Computer Science and Mathematical Economics.
As we quote from Cai and Papadimitriou \cite{CP14}: ``{\em ...\ multi-item mechanisms have played a role akin to that enjoyed by the Traveling Salesman Problem in combinatorial optimization: A paradigmatic hard nut on which all new ideas must be tried\ ...}''
It is impossible to review such extensive literature here thoroughly.
Instead, we list some of these works in \Cref{tab:multi-item} for guidance.

\begin{table}[h]
    \centering
    \begin{tabular}{|l|l|l|}
        \hline
        \rule{0pt}{13pt} & single-buyer & multi-buyer \\ [2pt]
        \hline
        \hline
        unit-demand &
        \begin{tabular}{@{}l@{}}
        \rule{0pt}{12pt}\cite{CHK07,CHMS10,CMS15} \\
        \cite{CD15,HH15,CDW21,JLQTX19} \\ [2pt]
        \end{tabular}
        & \cite{CMS15,KW19,HH15,CDW21} \\ [2pt]
        \hline
        \rule{0pt}{13pt}additive &
        \begin{tabular}{@{}l@{}}
        \rule{0pt}{12pt}\cite{HN17,LY13,BILW20} \\
        \cite{HH15,CDW21,MS21} \\ [2pt]
        \end{tabular}
        & \cite{Y15,HH15,CDW21,DFLSV22} \\ [2pt]
        \hline
        \hline
        \rule{0pt}{13pt}matroid-rank & \cite{CM16,CZ17} & \cite{CM16,CZ17} \\ [2pt]
        \hline
        \rule{0pt}{13pt}constrained-additive & \cite{CM16,CZ17} & \cite{CZ17} \\ [2pt]
        \hline
        \rule{0pt}{13pt}XOS & \cite{RW18,CZ17} & \cite{CZ17,Y18,COZ22} \\ [2pt]
        \hline
        \rule{0pt}{13pt}subadditive & \cite{RW18,CZ17} & \cite{CZ17,DKL20,CC23,CCW23} \\ [2pt]
        \hline
    \end{tabular}
    \caption{A summary of some previous works on Multi-Item Mechanism Design.
    \label{tab:multi-item}}
\end{table}

Very recently, it was shown in \cite{CC23,CCW23} that, even in the (most general) subadditive multi-buyer setting, simple deterministic mechanisms are constant-factor approximately revenue-optimal, for which prophet inequalities again play a pivotal role.
Hence, it is now a natural time to tighten these bounds.
For this purpose, our contributions are threefold:
(i)~the benchmark-tight ratio of {\UIVVItemPricing},
(ii)~a new technique with enough room for future generalizations, benchmark-based prophet inequalities, and
(iii)~showing the limitations of the single-dimensional representative approach.
We hope all our positive and negative results will inspire future works.

\vspace{.1in}
\noindent
{\bf Organization.}
We present the notation and preliminaries in \Cref{sec:prelim}.
The upper-bound parts of \Cref{thm:UIVVIP,thm:DRB_SPP} are proved in \Cref{sec:DTM}, which basically are a benchmark-based $3$-approximation prophet inequality.
Also, the lower-bound parts of \Cref{thm:UIVVIP,thm:DRB_SPP,cor:OIP} are proved in \Cref{subsec:SPP_LB,sec:example}.
How to implement the considered simple mechanisms in polynomial time is relatively simple and is deferred to \Cref{sec:implementation}.

\section{Notation and Preliminaries}
\label{sec:prelim}

This section introduces the notation to be adopted throughout the paper and backgrounds about Bayesian mechanism design.

\vspace{.1in}
\noindent
{\bf Notation.}
Denote by $\R$ the set of all real numbers. For any pair of integers $m \geq n \geq 0$, define the sets $[n] \eqdef \{1, 2, \cdots, n\}$ and $[n: m] \eqdef \{n, n + 1, \cdots, m\}$. Denote by $\mathbbm{1}(\cdot)$ the indicator function.

\subsection{Three equivalent representations of distributions}
\label{subsec:distribution}

We will adopt three equivalent representations for an instance:
(i)~value distributions,
(ii)~revenue-quantile curves, and
(iii)~virtual value distributions.
These three representations mutually admit one-to-one correspondences, so we can always adopt the one that is most convenient for our presentation.

The first representation of an instance is its (nonnegative) $n$-dimensional product value distribution $\bF = \{F_{i}\}_{i \in [n]}$.
Without ambiguity, each one-dimensional value distribution $F_{i}$ also denotes its cumulative density function (CDF).
Given a price-$p$ item, a buyer with a random value $v_{i} \sim F_{i}$ purchases the item with probability $\Pr_{v_{i} \sim F_{i}}[v_{i} \geq p]$ rather than $\Pr_{v_{i} \sim F_{i}}[v_{i} > p]$.
In this regard, we would modify the definitions of the CDF $F_{i}(v) \eqdef \Pr_{v_{i} \sim F_{i}}[v_{i} < v]$, $\forall v \in \RR$ and the inverse CDF $F_{i}^{-1}(y) \eqdef \inf \{v \in \RR \mid F_{i}(v) \geq y\}$, $\forall y \in [0, 1]$.


\begin{figure}[t]
    \centering
    \subfloat[\label{fig:iron_curve_nonconcave}
    {\em non-concave} revenue-quantile curves]{
    \includegraphics[width = .48\textwidth]
    {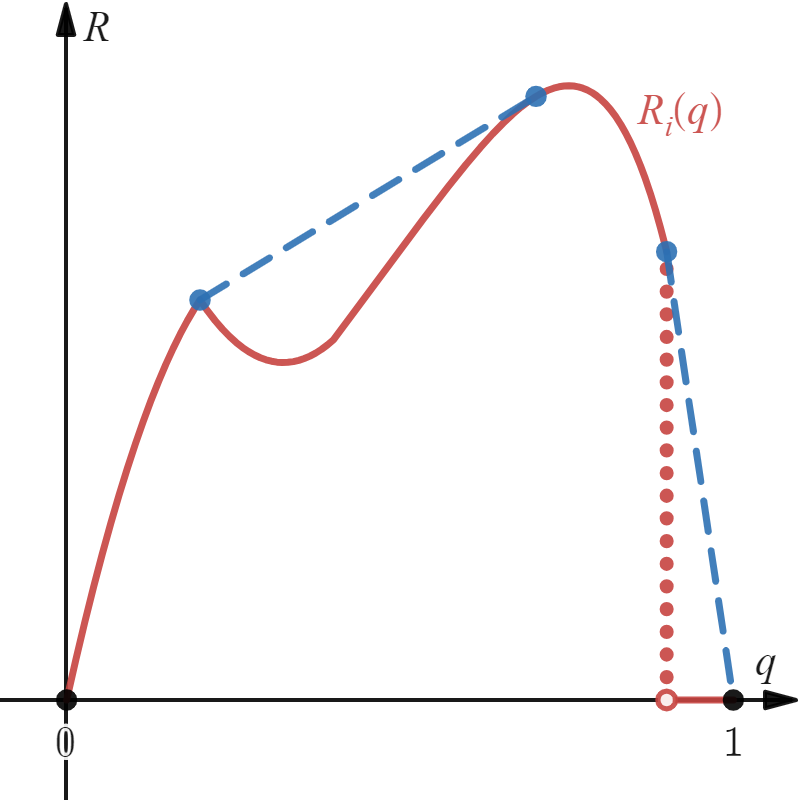}}
    \hfill
    \subfloat[\label{fig:iron_curve_concave}
    {\em concave} revenue-quantile curves]{
    \includegraphics[width = .48\textwidth]
    {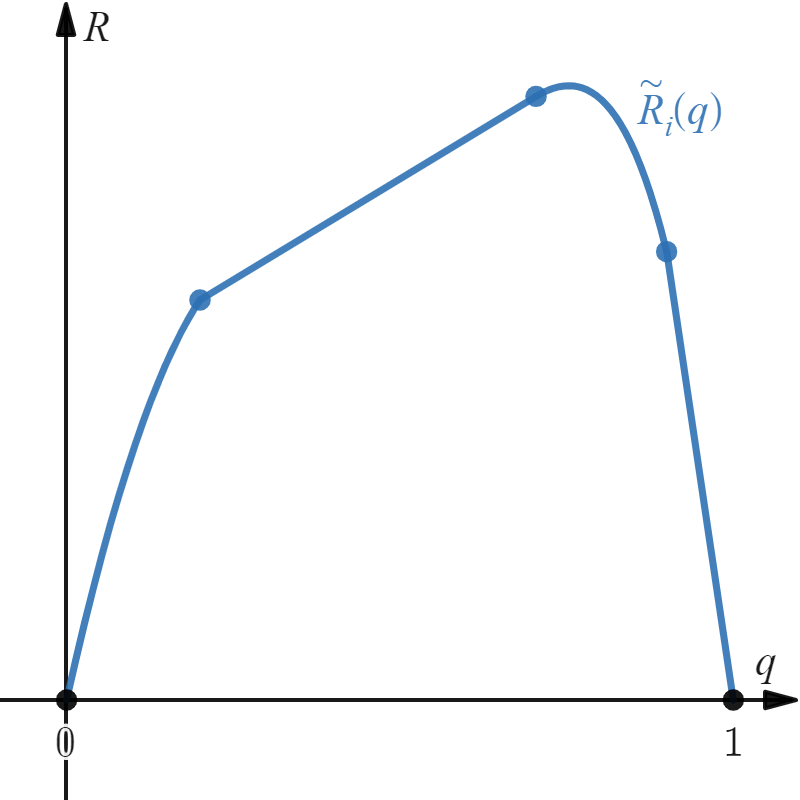}} \\
    \subfloat[\label{fig:iron_CDF_undefined}
    {\em badly-defined} virtual value CDF's]{
    \includegraphics[width = .48\textwidth]
    {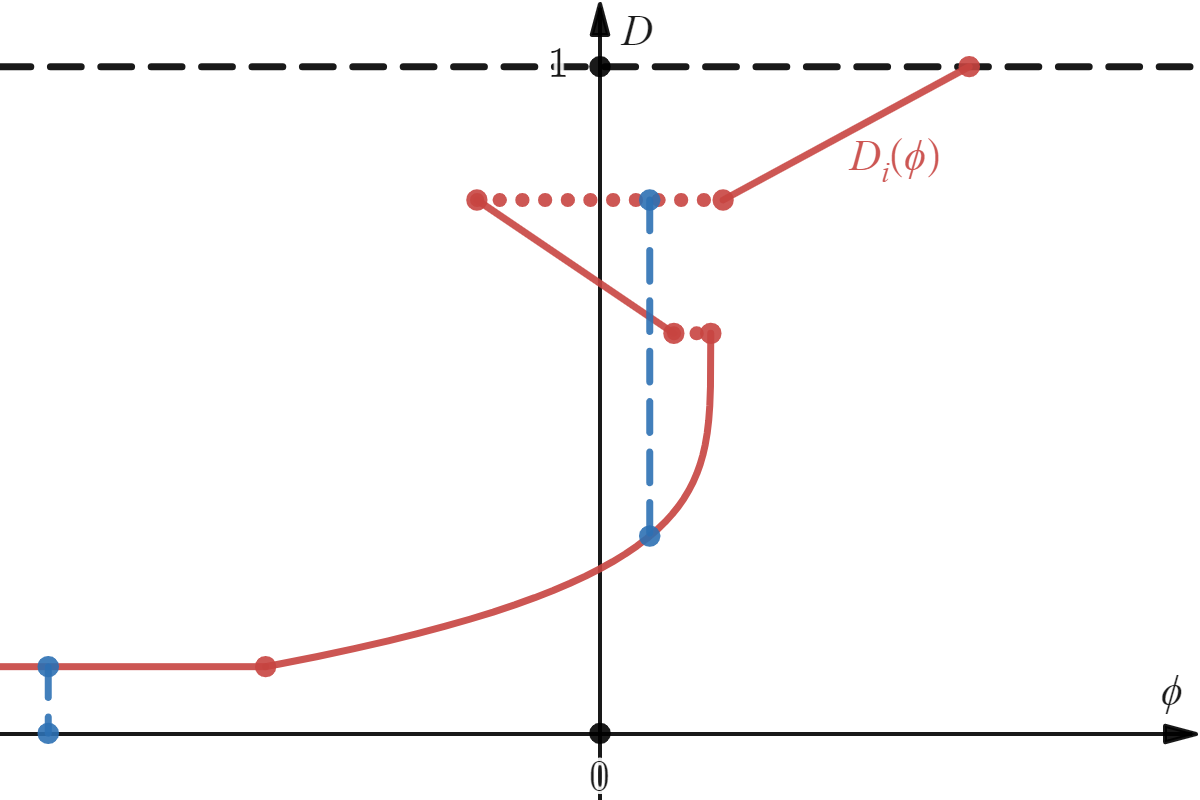}}
    \hfill
    \subfloat[\label{fig:iron_CDF_defined}
    {\em well-defined} virtual value CDF's]{
    \includegraphics[width = .48\textwidth]
    {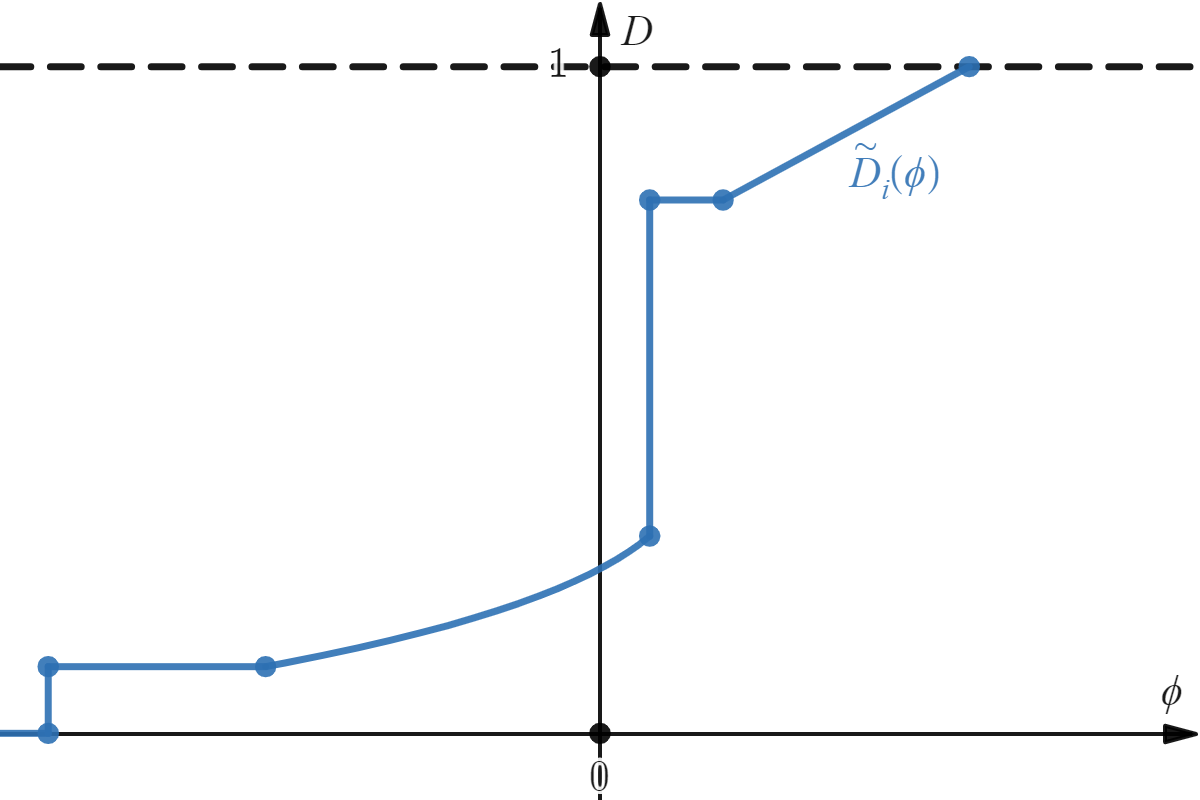}}
    \caption{Diagrams of revenue-quantile curves, virtual value CDF's, and ``ironing'' \cite{M81}.
    \label{fig:distribution}}
\end{figure}

The second representation of an instance is its revenue-quantile curves $\bR = \{R_{i}\}_{i \in [n]}$, which are determined by the value CDF's $\bF = \{F_{i}\}_{i \in [n]}$ as follows (cf.\ \Cref{fig:iron_curve_nonconcave,fig:iron_curve_concave}).
\begin{align}
    R_{i}(q) ~\eqdef~ q \cdot F_{i}^{-1}(1 - q),~ \forall q \in [0,\, 1].
    \label{eq:value_to_curve}
\end{align}
In particular, $R_{i}(0) \eqdef \lim_{q \to 0^{+}} R_{i}(q)$ when the value CDF $F_{i}$ has an unbounded support, namely $F_{i}(v) < 1$ whenever $v < +\infty$.

The third representation of an instance is its virtual value CDF's $\bD = \{D_{i}\}_{i \in [n]}$, which are determined by the revenue-quantile curves $\bR = \{R_{i}\}_{i \in [n]}$ as follows (cf.\ \Cref{fig:iron_CDF_undefined,fig:iron_CDF_defined}).
\begin{align}
    (\phi,\, D_{i}(\phi)) ~\eqdef~ (R'_{i}(q),\, 1 - q),~ \forall q \in [0,\, 1].
    \label{eq:curve_to_virtual_value}
\end{align}
(In general, the derivatives $R'_{i}(q)$ may not be decreasing functions over $q \in [0,\, 1]$, so \Cref{eq:curve_to_virtual_value} may induce badly-defined virtual value CDF's. However, in case of {\em regular distributions} (introduced below), the $R'_{i}(q)$'s are decreasing over $q \in [0,\, 1]$, so \Cref{eq:curve_to_virtual_value} induces well-defined virtual value CDF's.
Without ambiguity, we still call $\bD = \{D_{i}\}_{i \in [n]}$ virtual value CDF's.)
In particular, when $R_{i}(0) > 0$, by which the value CDF $F_{i}$ must have an unbounded support, we shall interpret $D_{i}$ as taking an infinite virtual value $R_{i}(\eps) / \eps \to +\infty$ with an infinitesimal probability mass $\eps \to 0^{+}$.

Alternatively, through the virtual value functions $\varphi_{i}(v) \eqdef v - \frac{1 - F_{i}(v)}{F'_{i}(v)}$ for $v \in \supp(F_{i})$, we can determine the virtual value CDF's $\bD = \{D_{i}\}_{i \in [n]}$ directly from the value CDF's $\bF = \{F_{i}\}_{i \in [n]}$:
\begin{align}
    (\phi,\, D_{i}(\phi)) ~\eqdef~ (\varphi_{i}(v),\, F_{i}(v)),~ \forall v \in \supp(F_{i}).
    \label{eq:value_to_virtual_value}
\end{align}
A virtual value is always upper bounded by the corresponding value, $\varphi_{i}(v) \leq v$ for $v \in \supp(F_{i})$.

Given \Cref{eq:value_to_curve,eq:curve_to_virtual_value,eq:value_to_virtual_value}, it is easy to see that all the three representations mutually admit one-to-one correspondences. I.e., we can use any one representation to reconstruct the other two.

\vspace{.1in}
\noindent
{\bf Regular Instances and Ironing.}
The family of {\em regular} value distributions (\Cref{def:regular}) was introduced in Myerson's original paper \cite{M81} and plays an important role in Bayesian mechanism design. In particular, Myerson gave a transformation called {\em ironing} (\Cref{def:ironing}) that converts irregular instances $\bF = \{F_{i}\}_{i \in [n]}$ into regular instances $\bar{\bF} = \{\bar{F}_{i}\}_{i \in [n]}$ such that, in the single-item multi-buyer setting, the optimal revenues keep the same.

\begin{definition}[Regular Instances]
\label{def:regular}
\begin{flushleft}
A value distribution $F_{i}$ is {\em regular} when its revenue-quantile curve $R_{i}(q)$ is concave over $q \in [0,\, 1]$, or equivalently, when   its virtual value CDF $D_{i}(\phi)$ is well- defined over $\phi \in \R$.
Likewise, an instance $\bF = \{F_{i}\}_{i \in [n]}$ is {\em regular} when all $F_{i}$'s are {\em regular}.
\end{flushleft}
\end{definition}

\begin{definition}[{Ironing \cite{M81}}]
\label{def:ironing}
\begin{flushleft}
Given a (possibly non-concave) revenue-quantile curve $R_{i}$, the ironed revenue-quantile curve $\bar{R}_{i} \eqdef \Conv(R_{i})$ refers to the concave envelope of $R_{i}$. Then the ironed value CDF $\bar{F}_{i}$, the ironed virtual value function $\bar{\varphi}_{i}$, and the ironed virtual value CDF $\bar{D}_{i}$ all can be reconstructed from this concave envelope $\bar{R}_{i}$. \\
In particular, an ironed virtual value is always upper bounded by the corresponding value, $\bar{\varphi}_{i}(v) = v - \frac{1 - \bar{F}_{i}(v)}{\bar{F}'_{i}(v)} \leq v$ for $v \in \supp(\bar{F}_{i})$.
\end{flushleft}
\end{definition}



\subsection{Bayesian mechanism design}
\label{subsec:mechanism}

In the single-item multi-buyer setting (\Cref{def:SPP}), in addition to {\UIVVSequentialPostedPricing}, we also study the revenue-optimal mechanism {\MyersonAuction} ({\MA}) and the welfare-optimal mechanism {\SecondPriceAuction} ({\SPA}).
\begin{flushleft}
\begin{itemize}
    \item {\MyersonAuction}: The buyer $i_{\MA} \in [n] \cup \{\emptyset\}$ with the highest nonnegative ironed virtual value (using an arbitrary but fixed tie-breaking rule) is allocated and pays the threshold value for continuing to be allocated. \\
    (\Cref{prop:revenue_equivalence}) The resulting revenue formula is $\MA(\bF) = \E_{\bv \sim \bF} \big[ \max(\bar{\bvarphi}(\bv),\, 0) \big]$.
    
    \item {\SecondPriceAuction}: The buyer $i_{\SPA} \in [n]$ with the highest value is allocated \\
    (using an arbitrary but fixed tie-breaking rule) and pays the second-highest value. \\
    The resulting revenue formula is $\SPA(\bF) = \E_{\bv \sim \bF} \big[ \max_{i \in [n]} \big\{ v_{i} \cdot \indicator(i \neq \argmax(\bv)) \big\} \big]$.
\end{itemize}
\end{flushleft}
Both mechanisms and {\UIVVSequentialPostedPricing} are truthful mechanisms and satisfy {\em revenue equivalence} \cite{M81}.
Moreover, the equality condition in \Cref{prop:revenue_equivalence} holds for {\UIVVSequentialPostedPricing} (under the particular prices $\bp$ given in \Cref{def:UIVV_prices}) and {\MyersonAuction}, but in general fails for {\SecondPriceAuction}.

\begin{proposition}[{Revenue Equivalence \cite{M81}}]
\label{prop:revenue_equivalence}
\begin{flushleft}
Given a single-item truthful mechanism $\calM = (\bx, \bpi)$, its revenue formula $\calM(\bF)$ satisfies the following:
\begin{align*}
 \calM(\bF)
 ~=~ \E_{\bv \sim \bF} \big[ \bvarphi(\bv) \cdot \bx(\bv) \big]
 ~\leq~ \E_{\bv \sim \bF} \big[ \bar{\bvarphi}(\bv) \cdot \bx(\bv) \big].
\end{align*}
In particular, the equality in the second step holds if and only if ``each buyer $i \in [n]$ is allocated with the same probability $x_{i}(v_{i}) = x_{i}(\tilde{v}_{i})$, for any two values $v_{i}$ and $\tilde{v}_{i}$ with the same ironed virtual value $\bar{\varphi}_{i}(v_{i}) = \bar{\varphi}_{i}(\tilde{v}_{i})$.''
\end{flushleft}
\end{proposition}

In the unit-demand single-buyer setting (\Cref{def:IP}), in addition to {\UIVVItemPricing}, we also study the {\DualityRelaxationBenchmark} \cite{CMS15,CDW21}, whose revenue formula is given by
\begin{align*}
    \DRB(\bF)
    & ~\eqdef~ \E_{\bv \sim \bF} \Big[\, \max_{i \in [n]} \big\{ \bar{\varphi}_{i}(v_{i}) \cdot \indicator(i = \argmax(\bv)) + v_{i} \cdot \indicator(i \neq \argmax(\bv)) \big\} \,\Big] \\
    & ~\;\equiv~ \E_{\bv \sim \bF} \Big[\, \max(\bar{\varphi}_{(1)}(v_{(1)}), \bv_{-(1)}) \,\Big] \\
    & ~\;\equiv~ \E_{\bv \sim \bF} \Big[\, \max(\bar{\varphi}_{(1)}(v_{(1)}), v_{(2)}) \,\Big]
\end{align*}
(When there are at least two items $n \geq 2$, every outcome of this benchmark is at least the second-highest value and thus is nonnegative. The last two lines, without ambiguity, simply write $(1) \in [n]$ for the highest bidder and $(2) \in [n]$ for the second-highest bidder.)
This revenue formula is a ``mixture'' of the {\MyersonAuction} revenue and the {\SecondPriceAuction} revenue. I.e., in every possible outcome $\bv \sim \bF$ we always have $\max(\MA(\bv),\, \SPA(\bv)) \leq \DRB(\bv) \leq \MA(\bv) + \SPA(\bv)$.

Also, it was proved in \cite{CMS15,CDW21} that {\DualityRelaxationBenchmark} revenue-surpasses the revenue-optimal mechanisms {\OptimalLotteryPricing} and {\OptimalItemPricing}.

\begin{proposition}[{Unit-Demand Single-Buyer Revenue Maximization \cite{CMS15,CDW21}}]
\begin{flushleft}
\, \\
Given the same instance $\bF$,
{\DualityRelaxationBenchmark} revenue-surpasses {\OptimalLotteryPricing}, which then revenue-surpasses {\OptimalItemPricing} $\DRB(\bF) \geq \OLP(\bF) \geq \OIP(\bF)$.
\end{flushleft}
\end{proposition}

\noindent
{\bf Stochastic Dominance and Revenue Monotonicity.}
We say a distribution $F_{i}$ is stochastically dominated by another distribution $\tilde{F}_{i}$, namely $F_{i} \preceq \tilde{F}_{i}$, when their value CDF's satisfy $F_{i}(v) \geq \tilde{F}_{i}(v)$ for $v \geq 0$, or equivalently, when their revenue-quantile curves satisfy $R_{i}(q) \leq \tilde{R}_{i}(q)$ for $q \in [0,\, 1]$.
Likewise, when $(F_{i} \preceq \tilde{F}_{i},\ \forall i \in [n]) \iff (R_{i} \preceq \tilde{R}_{i},\ \forall i \in [n])$, we can generalize stochastic dominance to instances $\bF \preceq \tilde{\bF} \iff \bR \preceq \tilde{\bR}$.


\begin{remark}[Stochastic Dominance]
\label{rem:revenue_monotonicity}
\begin{flushleft}
Given two regular instances $\bF$ and $\tilde{\bF}$, if their well-defined virtual value CDF's $\bD$ and $\tilde{\bD}$ satisfy stochastic dominance $\bD \preceq \tilde{\bD}$, then their value CDF's and revenue-quantile curves must also satisfy stochastic dominance $\bD \preceq \tilde{\bD} \implies \bF \preceq \tilde{\bF} \iff \bR \preceq \tilde{\bR}$.
(However, the opposite in general does not hold $\bF \preceq \tilde{\bF} \not\implies \bD \preceq \tilde{\bD}$.)
\end{flushleft}
\end{remark}

In regard to stochastic dominance of instances, single-item mechanisms often satisfy the following property called {\em revenue monotonicity} (see, e.g., \cite{DHP16,JLX23}).

\begin{proposition}[Revenue Monotonicity {\cite{DHP16,JLX23}}]
\label{prop:revenue_monotonicity}
Given any two instances $\bF \preceq \tilde{\bF}$: \\
(i)~{\MyersonAuction} satisfies revenue monotonicity $\MA(\bF) \leq \MA(\tilde{\bF})$. \\
(ii)~{\SecondPriceAuction} satisfies revenue monotonicity $\SPA(\bF) \leq \SPA(\tilde{\bF})$.
\end{proposition}

In general, multi-item (revenue-optimal) mechanisms violate revenue monotonicity \cite{HR15}.\footnote{\cite[Example~2]{HR15} gave a revenue-nonmonotone instance of a single {\em additive} buyer with independent values.}
Nonetheless, we can establish (\Cref{lem:revenue_monotonicity}) this property for {\DualityRelaxationBenchmark} in the considered setting of a single unit-demand buyer with independent values.

\begin{remark}[Revenue (Non-)Monotonicity]
It is interesting to determine how general such revenue monotonicity could be. Unfortunately, in \Cref{sec:non-monotonicity}, we show that this property fails even for a single \textit{additive} buyer.
\end{remark}

\begin{lemma}[Revenue Monotonicity]
\label{lem:revenue_monotonicity}
Given any two instances $\bF \preceq \tilde{\bF}$,
{\DualityRelaxationBenchmark} satisfies revenue monotonicity $\DRB(\bF) \leq \DRB(\tilde{\bF})$.
\end{lemma}

\begin{proof}
Any outcome of the random values $\bv \sim \bF$ (in decreasing order $v_{(1)} \ge v_{(2)} \ge \dots \ge v_{(n)}$), as we consider an arbitrary but fixed tie-breaking rule, admits a {\em strict total order} $v_{(1)} \succ v_{(2)} \succ \dots \succ v_{(n)}$.

It suffices to show $\DRB(\bF) \leq \DRB(\tilde{F}_{1} \otimes \bF_{-1})$; assuming the correctness of this equation, we can apply such arguments to the other $(n-1)$ distributions, and \Cref{lem:revenue_monotonicity} follows by induction.

We establish $\DRB(\bF) \leq \DRB(\tilde{F}_{1} \otimes \bF_{-1})$ using a coupling argument.
Given an outcome of the {\em invariant} random values $\bv_{-1} \sim \bF_{-1}$ and a uniform random quantile $q_{1} \sim \Unif[0, 1]$. The {\em interim} random values $v_{1} = F_{1}^{-1}(1 - q_{1}) \leq \tilde{v}_{1} = \tilde{F}_{1}^{-1}(1 - q_{1})$ must follow distributions $F_{1} \preceq \tilde{F}_{1}$, respectively.
Regarding the rank of $v_{1}$ in $(v_{1}, \bv_{-1})$ and the rank of $\tilde{v}_{1}$ in $(\tilde{v}_{1}, \bv_{-1})$, there are three cases -- we prove revenue monotonicity in each case:

\vspace{.1in}
\noindent
{\bf Case~1: $v_{(2)} \succ \tilde{v}_{1} \geq v_{1}$.}
Neither $v_{1}$ nor $\tilde{v}_{1}$ is the highest value in the respective value profiles.
Thus for any choice of $q_{1}$ in {\bf Case~1}, we always have
\begin{align*}
    \DRB(v_{1}, \bv_{-1})
    & ~=~ \max\Big(\bar{\varphi}_{(2)}(v_{(2)}), v_{1}, v_{(3)}, \dots, v_{(n)}\Big) \\
    & ~\leq~ \max\Big(\bar{\varphi}_{(2)}(v_{(2)}), \tilde{v}_{1}, v_{(3)}, \dots, v_{(n)}\Big)
    ~=~ \DRB(\tilde{v}_{1}, \bv_{-1})
\end{align*}

\noindent
{\bf Case~2: $\tilde{v}_{1} \succ v_{(2)} \succ v_{1}$.}
I.e., $v_{1}$ is not the highest value in $(v_{1}, \bv_{-1})$ but $\tilde{v}_{1}$ is the highest value in $(\tilde{v}_{1}, \bv_{-1})$.
As mentioned in \Cref{subsec:distribution}, an ironed virtual value is always upper bounded by the corresponding value $\bar{\varphi}_{(2)}(v_{(2)}) \leq v_{(2)}$. Thus for any choice of $q_{1}$ in {\bf Case~2}, we always have
\begin{align*}
    \DRB(v_{1}, \bv_{-1})
    & ~=~ \max\Big(\bar{\varphi}_{(2)}(v_{(2)}), v_{1}, v_{(3)}, \dots, v_{(n)}\Big) \\
    & ~\leq~ v_{(2)}
    ~\leq~ \max\Big(\bar{\tilde{\varphi}}_{1}(\tilde{v}_{1}), v_{(2)}, v_{(3)}, \dots, v_{(n)}\Big)
    ~=~ \DRB(\tilde{v}_{1}, \bv_{-1})
\end{align*}

\begin{figure}[t]
    \centering
    \includegraphics[width = .8\textwidth]
    {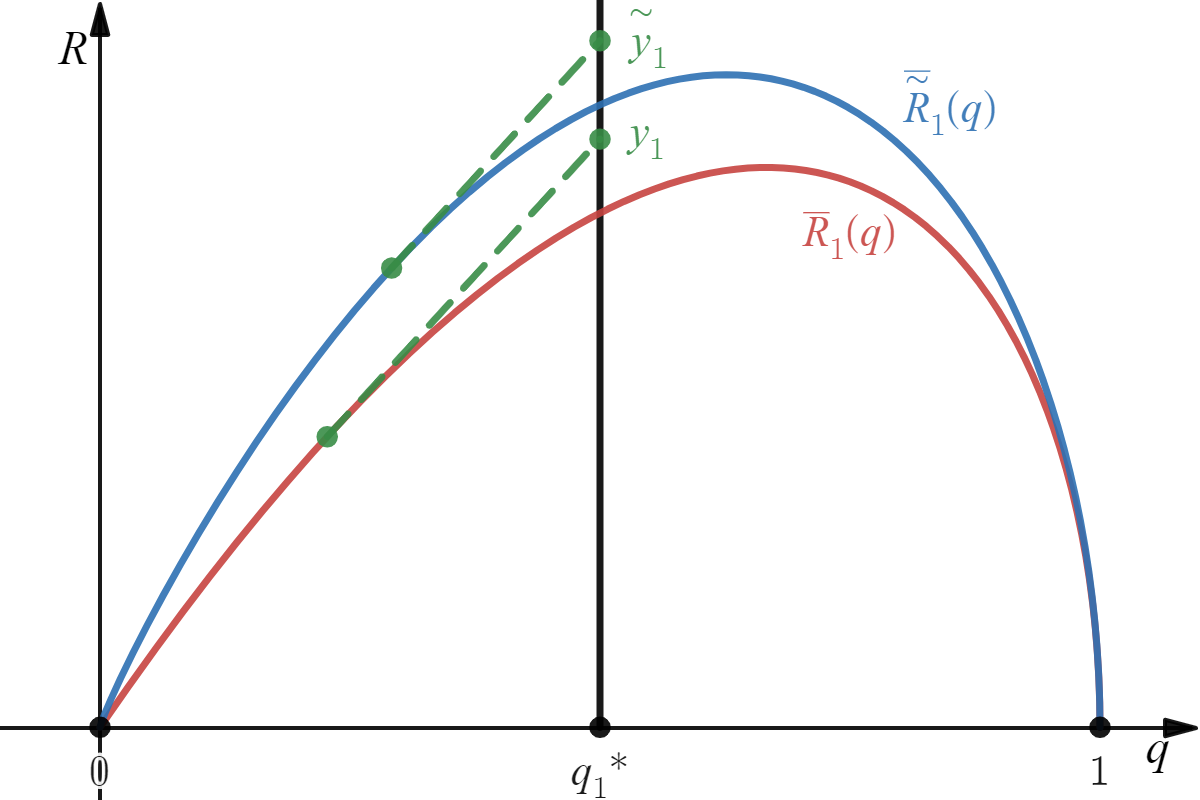}
    \caption{Diagram of {\bf Case~3} in the proof of \Cref{lem:revenue_monotonicity}, in terms of revenue-quantile curves.
    \label{fig:revenue_monotonicity}}
\end{figure}

\noindent
{\bf Case~3: $\tilde{v}_{1} \geq v_{1} \succ v_{(2)}$.}
Both $v_{1}$ and $\tilde{v}_{1}$ are the highest values in the respective value profiles. A lower quantile $q_{1}$ results in higher (or equal) interim values $v_{1} = F_{1}^{-1}(1 - q_{1})$ and $\tilde{v}_{1} = \tilde{F}_{1}^{-1}(1 - q_{1})$, so {\bf Case~3} occurs if and only if $0 \leq q_{1} \leq q_{1}^{*}$,\footnote{The boundary case $\{q_{1} = q_{1}^{*}\}$ can either be included or excluded, since it occurs with probability zero.} where the threshold quantile
\begin{align*}
    q_{1}^{*}
    ~=~ q_{1}^{*}(\bv_{-1})
    ~\eqdef~ \sup \{q \in [0, 1] \,|\, \tilde{F}_{1}^{-1}(1 - q) \succ v_{(2)} \land F_{1}^{-1}(1 - q) \succ v_{(2)}\}.
\end{align*}
Note that $q_{1}^{*} = \sup \{q \in [0, 1] \,|\, F_{1}^{-1}(1 - q) \succ v_{(2)}\}$, since $F_{1} \preceq \tilde{F}_{1}$.

As mentioned in \Cref{subsec:distribution}, $\bar{\varphi}_{1}(v_{1}) = \bar{R}'_{1}(q_{1})$ and $\bar{\tilde{\varphi}}_{1}(v_{1}) = \bar{\tilde{R}}'_{1}(q_{1})$, i.e., an ironed virtual value is precisely the derivative of the ironed revenue-quantile curve at the considered quantile.
Thus, we can deduce that
\begin{align*}
    \E_{q_{1} \sim \Unif[0, 1]} \Big[\, \DRB(v_{1}, \bv_{-1}) \cdot \indicator(\text{\bf Case~3}) \,\Bigmid\, \bv_{-1} \,\Big]
    & ~=~ \int_{0}^{q_{1}^{*}} \max\Big(\bar{R}'_{1}(q), v_{(2)}\Big) \cdot \d q \\
    & ~\leq~ \int_{0}^{q_{1}^{*}} \max\Big(\bar{\tilde{R}}'_{1}(q), v_{(2)}\Big) \cdot \d q \\
    & ~=~ \E_{q_{1} \sim \Unif[0, 1]} \big[ \DRB(\tilde{v}_{1}, \bv_{-1}) \cdot \indicator(\text{\bf Case~3}) \,\bigmid\, \bv_{-1} \big],
\end{align*}
where the inequality uses a simple geometric observation; see \Cref{fig:revenue_monotonicity} for a visual demonstration: \\
Since $\bar{R}_{1} = \Conv(R_{1})$ is a concave function over $q \in [0,\, 1]$, the integral $\int_{0}^{q_{1}^{*}} \max(\bar{R}'_{1}(q), v_{(2)}) \cdot \d q$ is equal to the $y$-coordinate (denoted by $y_{1}$) of the intersection of ``the slope-$v_{(2)}$ tangent of $\bar{R}_{1}$'' and ``the $x = q_{1}^{*}$ vertical line''. Similarly for $\bar{\tilde{R}}_{1}$; the integral $\int_{0}^{q_{1}^{*}} \max(\bar{\tilde{R}}'_{1}(q), v_{(2)}) \cdot \d q$ is equal to the counterpart $\tilde{y}_{1}$. Given the stochastic dominance $F_{1} \preceq \tilde{F}_{1} \iff R_{1} \preceq \tilde{R}_{1} \implies \bar{R}_{1} \preceq \bar{\tilde{R}}_{1}$, it is easy to infer that $y_{1} \leq \tilde{y}_{1}$.

\vspace{.1in}
\noindent
As a consequence, we can derive $\DRB(\bF) \leq \DRB(\bF_{-1}, \tilde{F}_{1})$ as follows:
\begin{align*}
 \DRB(\bF)
 & ~=~ \E_{\bv_{-1} \sim \bF_{-1}} \Big[\, \sum_{j \in \{1, 2, 3\}} \E_{q_{1} \sim \Unif[0, 1]} \Big[\, \DRB(v_{1}, \bv_{-1}) \cdot \indicator(\text{\bf Case~$j$}) \,\Bigmid\, \bv_{-1} \,\Big] \,\Big] \\
 & ~\leq~ \E_{\bv_{-1} \sim \bF_{-1}} \Big[\, \sum_{j \in \{1, 2, 3\}} \E_{q_{1} \sim \Unif[0, 1]} \Big[\, \DRB(\tilde{v}_{1}, \bv_{-1}) \cdot \indicator(\text{\bf Case~$j$}) \,\Bigmid\, \bv_{-1} \,\Big] \,\Big] \\
 & ~=~ \DRB(\bF_{-1}, \tilde{F}_{1}).
\end{align*}
As mentioned, by induction we can conclude with $\DRB(\bF) \leq \DRB(\tilde{\bF})$. This finishes the proof.
\end{proof}

\noindent
{\bf Revenue Gaps.}
Given two mechanisms $\calM_{1}$ and $\calM_{2}$ that the first revenue-surpasses the second, their revenue gap $\calC_{\calM_{1} / \calM_{2}} \geq 1$ is defined as, over all possible instances $\bF = \{F_{i}\}_{i \in [n]}$, the supremum ratio of the two revenues $\calC_{\calM_{1} / \calM_{2}} \eqdef \sup_{\bF} \big\{ \frac{\calM_{1}(\bF)}{\calM_{2}(\bF)} \big\} \geq 1$.

\section{A Benchmark-Based Prophet Inequality}
\label{sec:DTM}

In \Cref{sec:DTM}, we prove \Cref{thm:UIVVIP} (Item~(i)) and \Cref{thm:DRB_SPP}, the tight revenue guarantees of {\UIVVItemPricing} and {\UIVVSequentialPostedPricing} against {\DualityRelaxationBenchmark}.
For ease of reference, both theorems are restated below.


\begin{restate}[{\Cref{thm:UIVVIP}}]
\begin{flushleft}
Against {\DualityRelaxationBenchmark}: \\
(i)~{\UniformIronedVirtualValueItemPricing} with a {\UIVV}-threshold $T = \frac{1}{3} \cdot \DRB$ of one third of the {\DualityRelaxationBenchmark} revenue achieves a tight $\calC_{\DRB / \UIVVIP} = 3$ approximation.
\end{flushleft}
\end{restate}

\begin{restate}[{\Cref{thm:DRB_SPP}}]
\begin{flushleft}
Against {\DualityRelaxationBenchmark}: \\
(i)~{\UniformIronedVirtualValueSequentialPostedPricing} with a {\UIVV}-threshold $T = \frac{1}{3} \cdot \DRB$ of one third of the {\DualityRelaxationBenchmark} revenue achieves a $\calC_{\DRB / \UIVVSPP} = 3$ approximation. \\
(ii)~This ratio is optimal for \textbf{deterministic ironed-virtual-value-based} and/or \textbf{(arbitrary) uniform-ironed-virtual-value} stopping rules in the order-oblivious model.\textsuperscript{\ref{footnote:tie}}
\end{flushleft}
\end{restate}

The bulk of \Cref{sec:DTM} devotes to the upper-bound part of \Cref{thm:DRB_SPP}. Throughout \Cref{subsec:iron,subsec:truncate,subsec:extend,subsec:scale,subsec:perturb} we set the {\UIVV}-threshold $T = \alpha \cdot \DRB$ for a generic parameter $\alpha \in (0,\, 1)$,\footnote{\label{footnote:DTM_alpha}When $\alpha = 0$, no bounded revenue guarantee is possible: Suppose some item $i \in [n]$ has an deterministic infinitesimal value of $v_{i} \to 0^{+}$, then the adversary can choose any arrival order $\pi \in \Pi_{n}$ with $\pi(1) = i$ such that the gambler always accepts this item $i$ and gets a value of $v_{i} \to 0^{+}$. The other regime $\alpha \geq 1$ is similar.}
and step-by-step characterize the worst-case instances.
Only until \Cref{subsec:upper_bounds} we will choose $\alpha = \frac{1}{3}$ to optimize the bound $\calC_{\DRB / \UIVVSPP} = 3$. Then the upper-bound part of \Cref{thm:UIVVIP} follows as an implication of \Cref{prop:representative}.
Finally in \Cref{subsec:SPP_LB}, we show the lower-bound parts of both theorems.

\afterpage{
\begin{figure}[t]
    \centering
    \begin{mdframed}
    Reduction $\iron(\bR^{(0)})$

    \begin{flushleft}
    {\bf Input:} A (generic) instance $\bR^{(0)} = \{R_{i}^{(0)}\}_{i \in [n]}$.
    
    \vspace{.05in}
    {\bf Output:} A {\em regular} instance $\bR^{(1)} = \{R_{i}^{(1)}\}_{i \in [n]}$.

    \begin{enumerate}
        \item\label{alg:iron}
        {\bf Return} $\bR^{(1)} = \{R_{i}^{(1)}\}_{i \in [n]}$; each curve $R_{i}^{(1)} \eqdef \Conv(R_{i}^{(0)})$ is the concave envelope of $R_{i}^{(0)}$.
    \end{enumerate}
    \end{flushleft}
    \end{mdframed}
    \caption{The {\iron} reduction.
    \label{fig:alg:iron}}
\end{figure}
\begin{figure}[t]
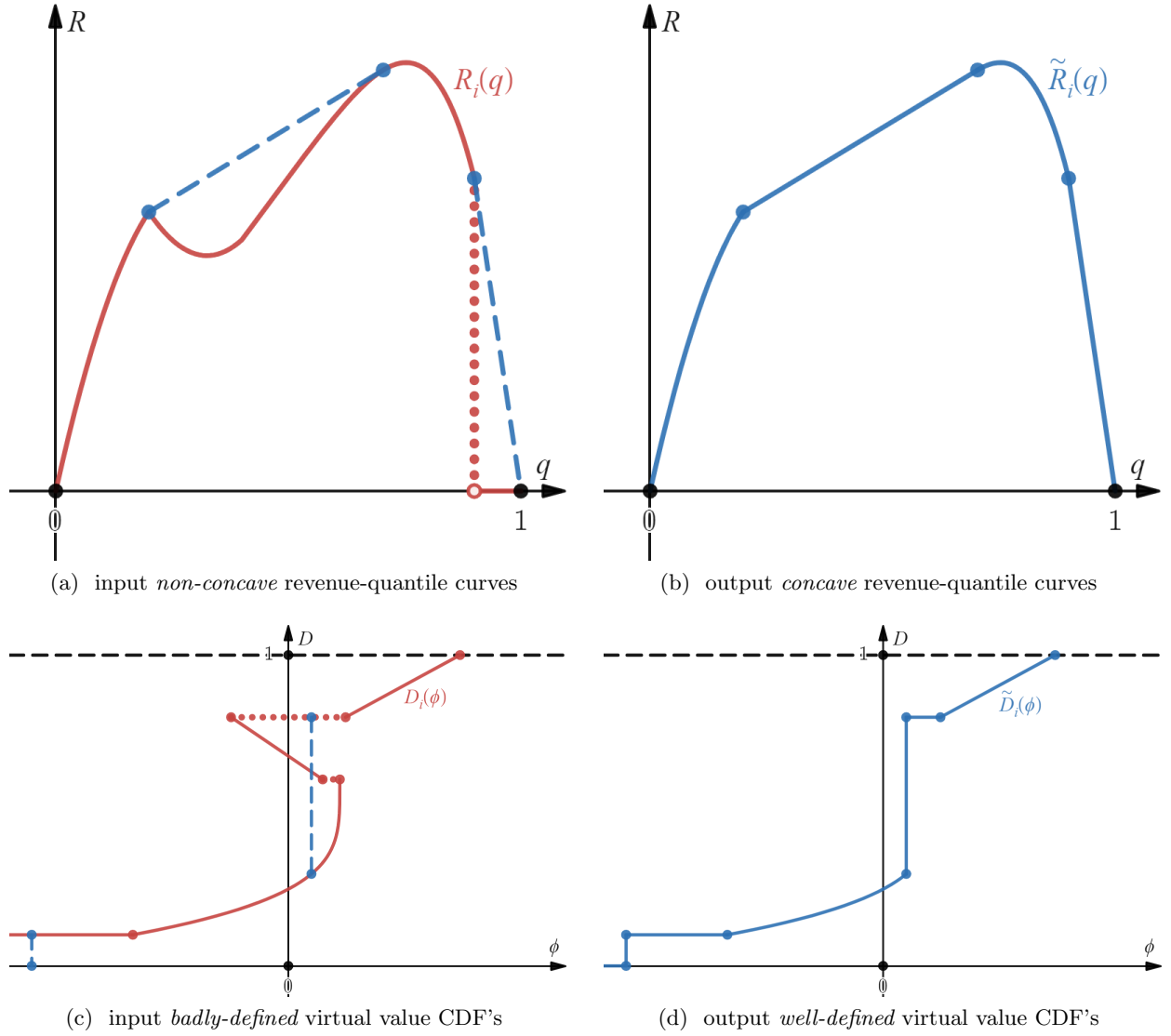

    \centering
    \subfloat[\label{fig:iron_curve_old}
    input {\em non-concave} revenue-quantile curves]{
    \includegraphics[width = .48\textwidth]
    {figure/iron_curve_old}}
    \hfill
    \subfloat[\label{fig:iron_curve_new}
    output {\em concave} revenue-quantile curves]{
    \includegraphics[width = .48\textwidth]
    {figure/iron_curve_new}} \\
    \subfloat[\label{fig:iron_CDF_old}
    input {\em badly-defined} virtual value CDF's]{
    \includegraphics[width = .48\textwidth]
    {figure/iron_CDF_old}}
    \hfill
    \subfloat[\label{fig:iron_CDF_new}
    output {\em well-defined} virtual value CDF's]{
    \includegraphics[width = .48\textwidth]
    {figure/iron_CDF_new}}
    \caption{Diagrams of {\iron} with revenue-quantile curves and virtual value CDF's.
    \label{fig:iron}}
\end{figure}
\clearpage}

\subsection{{\iron}: Reducing generic items to regular items}
\label{subsec:iron}

This section presents the {\iron} reduction (see \Cref{fig:alg:iron,fig:iron} for a description and a visual aid), based on {\bf revenue-quantile curves}, which transforms a generic instance $\bR = \bR^{(0)}$ into a {\em regular} instance $\bR^{(1)}$ (\Cref{def:regular}).
Here and later in \Cref{subsec:truncate,subsec:extend,subsec:scale}, we always use the abbreviation $T^{(0)} = \alpha \cdot \DRB(\bR^{(0)})$ to denote the {\em original} instance $\bR = \bR^{(0)}$'s {\UIVV}-threshold.

\Cref{lem:iron} summarizes the performance guarantees of the {\iron} reduction.
(Notice that this lemma uses the {\em original} instance's {\UIVV}-threshold $T^{(0)}$ for the new instance $\bR^{(1)}$.)

\begin{lemma}[{\iron}]
\label{lem:iron}
\begin{flushleft}
The reduction $\bR^{(1)} \gets \iron(\bR^{(0)})$ outputs a regular instance $\bR^{(1)}$ such that $\DRB(\bR^{(1)}) \geq \DRB(\bR^{(0)})$ and $\UIVVSPP(\bR^{(1)},\, T^{(0)}) = \UIVVSPP(\bR^{(0)},\, T^{(0)})$.
\end{flushleft}
\end{lemma}

\begin{proof}
Every output curve $R_{i}^{(1)} = \Conv(R_{i}^{(0)})$ is the concave envelope of its input counterpart $R_{i}$ (Line~\ref{alg:iron}) and thus satisfies {\bf regularity} (\Cref{def:regular}). Below, we investigate the revenues from {\DualityRelaxationBenchmark} and {\UIVVSequentialPostedPricing}.
\begin{itemize}
    \item $\DRB(\bR^{(1)}) \geq \DRB(\bR^{(0)})$.
    The {\DualityRelaxationBenchmark} revenue increases by the revenue monotonicity (\Cref{lem:revenue_monotonicity}) and that the output instance stochastically dominates the input instance $\bR^{(1)} \succeq \bR^{(0)}$ by construction (Line~\ref{alg:iron}).
    
    \item $\UIVVSPP(\bR^{(1)},\, T^{(0)}) = \UIVVSPP(\bR^{(0)},\, T^{(0)})$.
    The ironed virtual value CDF's keep the same $\bar{\bD}^{(1)} = \bar{\bD}^{(0)}$ under the reduction. Hence, under the {\UIVV}-threshold $T^{(0)}$ from \Cref{def:UIVV_prices}, the revenue equivalence (\Cref{prop:revenue_equivalence}) is applicable to both instances $\bR^{(1)}$ and $\bR^{(0)}$.
    Then an arbitrary but the same arrival order $\pi^{(1)} = \pi^{(0)} \in \Pi_{n}$ induces the same allocation rule $\bx^{(1)} = \bx^{(0)}$ and the same ironed virtual welfare
    \[
        \E_{\bphi \,\sim\, \bar{\bD}^{(1)}}\, [\, \bphi \cdot \bx^{(1)}(\bphi) \,]
        ~=~ \E_{\bphi \,\sim\, \bar{\bD}^{(1)}}\, [\, \bphi \cdot \bx^{(0)}(\bphi) \,]
        ~=~ \E_{\bphi \,\sim\, \bar{\bD}^{(0)}}\, [\, \bphi \cdot \bx^{(0)}(\bphi) \,].
    \]
    Specifically, the same worst-case arrival order $\pi^{(1)*} = \pi^{(0)*} \in \Pi_{n}$ induces the same worst-case virtual welfare, i.e., the same {\UIVVSequentialPostedPricing} revenue in the order-oblivious model
    $\UIVVSPP(\bR^{(1)},\, T^{(0)}) = \UIVVSPP(\bR^{(0)},\, T^{(0)})$. This finishes the proof.\qedhere
\end{itemize}
\end{proof}

\begin{remark}[Generalizations]
The first part of \Cref{lem:iron}, that $\DRB(\bR^{(1)}) \geq \DRB(\bR^{(0)})$, only uses the revenue monotonicity of {\DualityRelaxationBenchmark}. Therefore, when we consider a different benchmark that also satisfies the revenue monotonicity, \Cref{lem:iron} can be generalized seamlessly. This is also the case later in \Cref{lem:truncate,lem:extend}.
\end{remark}

The new instance $\bR^{(1)}$ slightly differs from our target: It would have a worse revenue gap than the {\em original} instance as desired, $\calC_{\DRB / \UIVVSPP}(\bR^{(1)},\, T^{(0)}) \geq \calC_{\DRB / \UIVVSPP}(\bR^{(0)},\, T^{(0)})$, but only if we could keep using the {\em original} instance's {\UIVV}-threshold $T^{(0)}$ for it.
The subsequent reductions in \Cref{subsec:truncate,subsec:extend} also have this shortcoming, but we will settle it in \Cref{subsec:scale} through the {\scale} reduction.\footnote{We can settle the issue directly for $\bR^{(1)}$ through a more sophisticated version of {\scale}. But the subsequent reductions in \Cref{subsec:truncate,subsec:extend} will output instances with better structures than $\bR^{(1)}$, for which the current version of {\scale} is sufficient. To ease the presentation, we defer {\scale} to \Cref{subsec:scale} and use the current version.}

Given the {\iron} reduction and \Cref{lem:iron}, we can concentrate on {\em regular} instances $\bF$, so the {\em concave} revenue-quantile curves $\bR$ and the {\em well-defined} virtual value CDF's $\bD$ are identical to the ironed counterparts $\bar{\bR}$ and $\bar{\bD}$.
Without ambiguity, later we will not distinguish those ironed counterparts from $\bR$ and $\bD$.

\afterpage{
\begin{figure}[t]
    \centering
    \begin{mdframed}
    Reduction $\truncate(\bD^{(1)},\, T^{(0)})$

    \begin{flushleft}
    {\bf Input:} The {\em regular} instance $\bD^{(1)} = \{D_{i}^{(1)}\}_{i \in [n]}$ (from \Cref{subsec:iron}) and the threshold $T^{(0)} > 0$ for the {\em original} instance $\bR = \bR^{(0)}$ (from \Cref{subsec:iron}).
    
    \vspace{.05in}
    {\bf Output:} A {\em truncated} instance $\bD^{(2)} = \{D_{i}^{(2)}\}_{i \in [n]}$.

    \begin{enumerate}
        \item\label{alg:truncate}
        {\bf Return} $\bD^{(2)} = \{D_{i}^{(2)}\}_{i \in [n]}$; each CDF $D_{i}^{(2)}(\phi) \eqdef D_{i}^{(1)}(\phi) \cdot \indicator(\phi \geq T^{(0)+})$ for $\phi \in \R$.
    \end{enumerate}
    \end{flushleft}
    \end{mdframed}
    \caption{The {\truncate} reduction.
    \label{fig:alg:truncate}}
\end{figure}
\begin{figure}[t]
    \centering
    \subfloat[\label{fig:truncate_CDF_old}
    input {\em untruncated} virtual value CDF's]{
    \includegraphics[width = .48\textwidth]
    {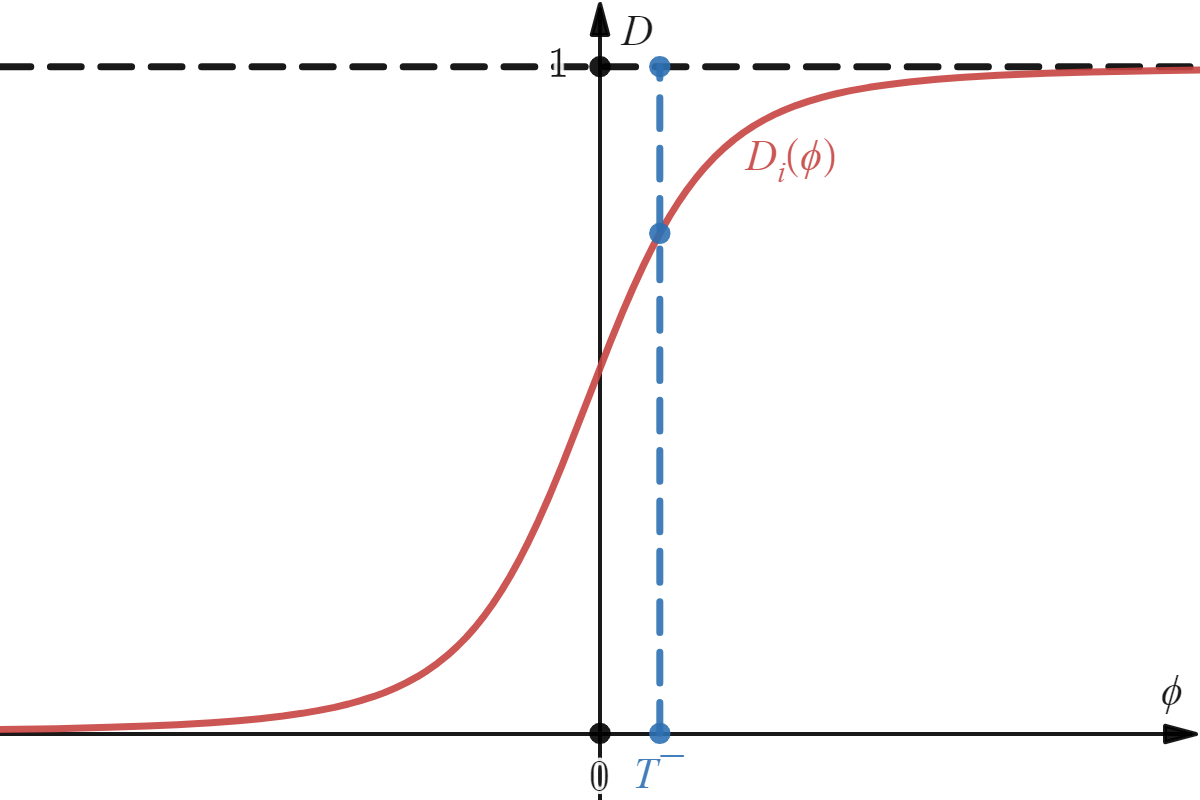}}
    \hfill
    \subfloat[\label{fig:truncate_CDF_new}
    output {\em truncated} virtual value CDF's]{
    \includegraphics[width = .48\textwidth]
    {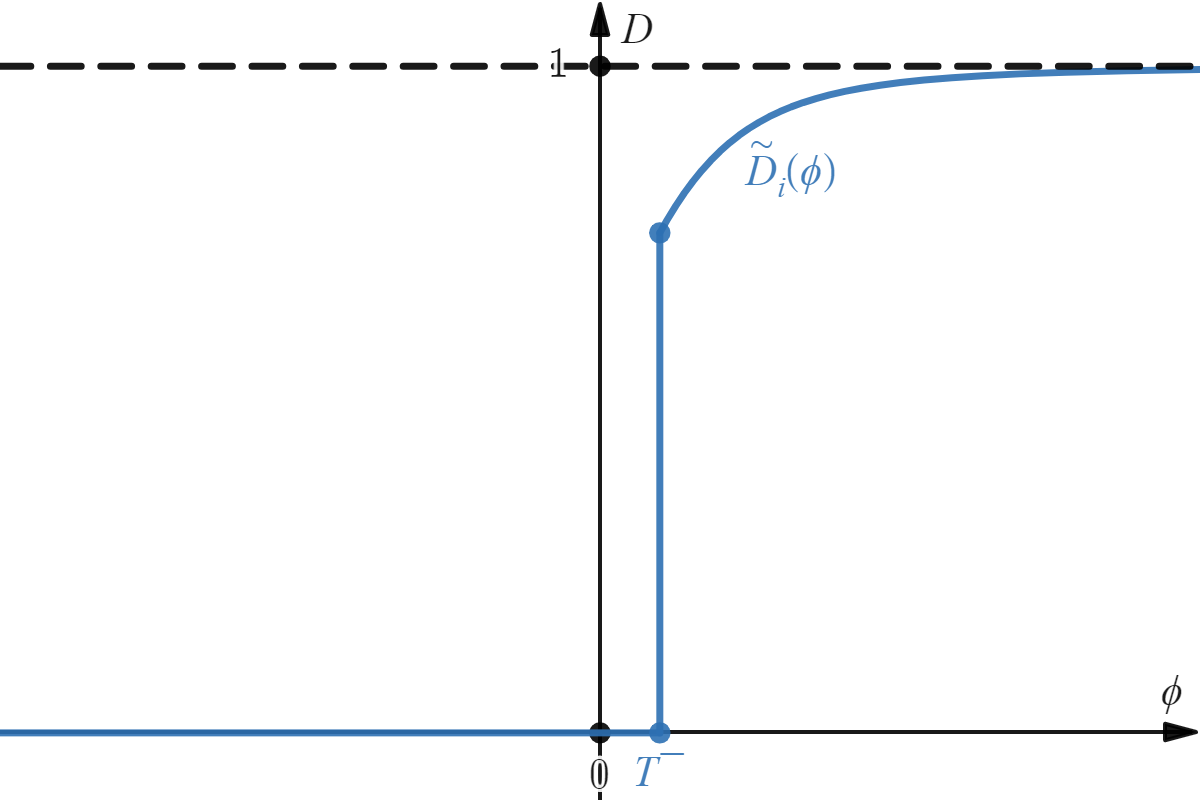}} \\
    \subfloat[\label{fig:truncate_curve_old}
    input {\em untruncated} revenue-quantile curves]{
    \includegraphics[width = .48\textwidth]
    {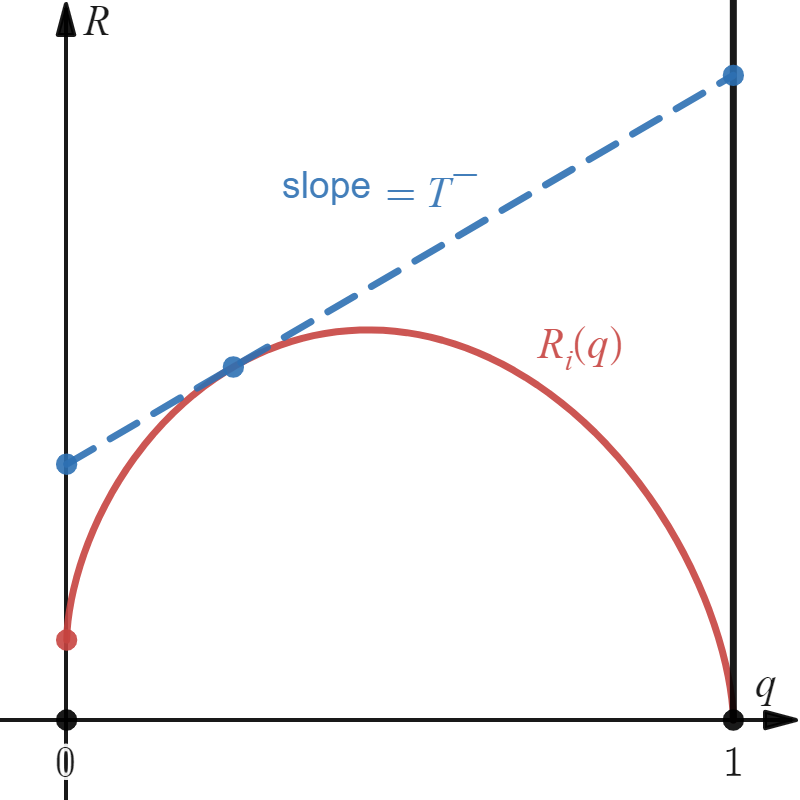}}
    \hfill
    \subfloat[\label{fig:truncate_curve_new}
    output {\em truncated} revenue-quantile curves]{
    \includegraphics[width = .48\textwidth]
    {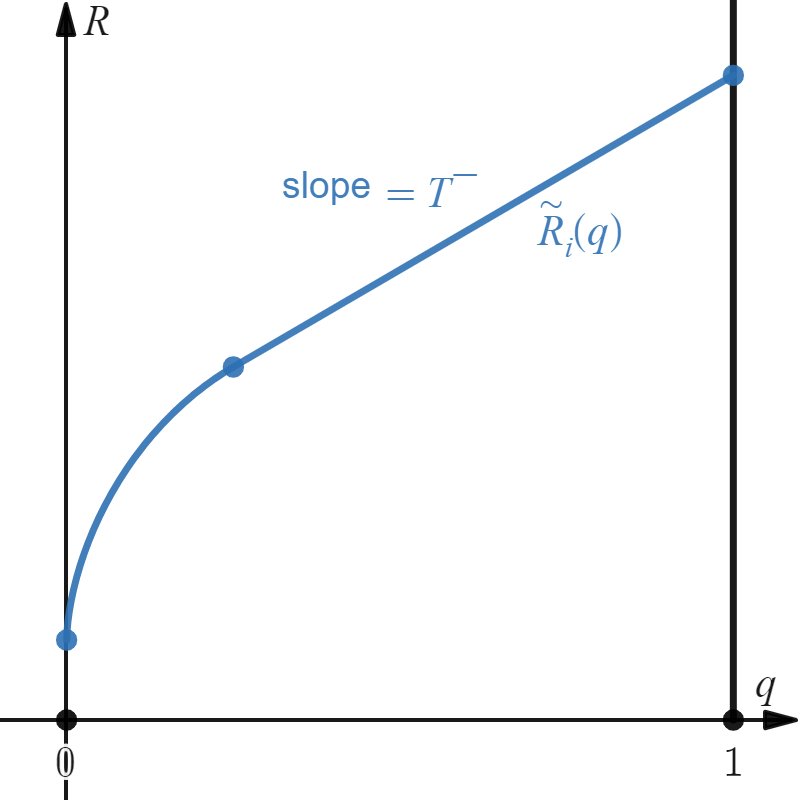}}
    \caption{Diagrams of {\truncate} with virtual value CDF's and revenue-quantile curves.
    \label{fig:truncate}}
\end{figure}
\clearpage}

\subsection{{\truncate}: Characterizing the below-threshold parts of virtual values}
\label{subsec:truncate}

This section gives the {\truncate} reduction (see \Cref{fig:alg:truncate,fig:truncate} for a description and a visual aid), based on {\bf virtual value CDF's}, which transforms a {\em regular} instance $\bD^{(1)}$ from \Cref{subsec:iron} into a {\em truncated} instance $\bD^{(2)}$ (\Cref{def:truncate}).

\Cref{lem:truncate} summarizes the performance guarantees of the {\truncate} reduction.

\begin{definition}[Truncated Instances]
\label{def:truncate}
\begin{flushleft}
A {\em regular} instance $\bF$ is further called {\em truncated} when its {\em well-defined} virtual value CDF's $\bD = \{D_{i}\}_{i \in [n]}$ all satisfy {\bf truncatedness}:
there exists a positive truncation point $\tau > 0$ such that $D_{i}(\phi) = 0$ for $\phi < \tau^{-}$ and $i \in [n]$, i.e., each $D_{i}$ is supported within $\supp(D_{i}) \subseteq \{\tau^{-}\} \cup [\tau^{+},\, +\infty]$.
Define the item-wise {\em division points} $\bkappa = \{\kappa_{i}\}_{i \in [n]}$ as the probabilities $\kappa_{i} \eqdef \Pr_{\phi_{i} \,\sim\, D_{i}}\, [\phi_{i} \geq \tau^{+}] = 1 - \Pr_{\phi_{i} \,\sim\, D_{i}}\, [\phi_{i} = \tau^{-}] \in [0,\, 1]$.

Following the definition of virtual values (\Cref{subsec:distribution}), an equivalent definition of {\bf truncatedness} is that all {\em concave} revenue-quantile curves $\bR = \{R_{i}\}_{i \in [n]}$ always have the derivatives $R'_{i}(q) \geq \tau^{-}$ for $q \in [0,\, 1]$. Thus, the {\em division points} are given by $\kappa_{i} \eqdef \sup\, \{ q \in [0,\, 1] \mid R'_{i}(q) \geq \tau^{+} \}$ for $i \in [n]$.
\end{flushleft}
\end{definition}

\begin{lemma}[{\truncate}]
\label{lem:truncate}
The reduction $\bD^{(2)} \gets \truncate(\bD^{(1)})$ outputs a truncated instance $\bD^{(2)}$ such that $\DRB(\bD^{(2)}) \geq \DRB(\bD^{(1)})$ and $\UIVVSPP(\bD^{(2)},\, T^{(0)}) = \UIVVSPP(\bD^{(1)},\, T^{(0)})$.
\end{lemma}

\begin{proof}
Each output CDF $D_{i}^{(2)}(\phi) \eqdef D_{i}^{(1)}(\phi) \cdot \indicator(\phi \geq T^{(0)-})$ for $i \in [n]$ truncates the $\leq T^{(0)-}$ part of its input counterpart $D_{i}^{(1)}$ (Line~\ref{alg:truncate}), where the truncation point $T^{(0)} > 0$ is the threshold for the {\em original} instance $\bR^{(0)}$ from \Cref{subsec:iron}.
This output instance $\bD^{(2)}$ satisfies both {\bf regularity} (since the virtual value CDF's $\{D_{i}^{(2)}\}_{i \in [n]}$ are {\em well-defined}; \Cref{subsec:distribution}) and {\bf truncatedness} (since $\tau^{(2)} = T^{(0)}$ is a valid truncation point).
Below, we examine the revenues from {\DualityRelaxationBenchmark} and {\UIVVSequentialPostedPricing}.
\begin{itemize}
    \item $\DRB(\bD^{(2)}) \geq \DRB(\bD^{(1)})$.
    The {\DualityRelaxationBenchmark} revenue increases by the revenue monotonicity (\Cref{lem:revenue_monotonicity}) and that output CDF's $D_{i}^{(2)}(\phi) \eqdef D_{i}^{(1)}(\phi) \cdot \indicator(\phi \geq T^{(0)+})$ stochastically dominate the input counterparts $D_{i}^{(1)}$ (Line~\ref{alg:truncate}); as mentioned in \Cref{rem:revenue_monotonicity}, this implies the stochastic dominance between the instances $\bD^{(2)} \succeq \bD^{(1)} \Rightarrow \bR^{(2)} \succeq \bR^{(1)} \iff \bF^{(2)} \succeq \bF^{(1)}$.
    
    \item $\UIVVSPP(\bD^{(2)},\, T^{(0)}) = \UIVVSPP(\bD^{(1)},\, T^{(0)})$.
    The reduction preserves the $\geq T^{(0)+}$ parts of the input CDF's $\bD^{(1)} = \{D_{i}^{(1)}\}_{i \in [n]}$.
    Therefore, using the same threshold $T^{(0)} > 0$ and an arbitrary but the same arrival order $\pi^{(2)} = \pi^{(1)} \in \Pi_{n}$ for both instances $\bD^{(2)}$ and $\bD^{(1)}$ induces the same allocation rule $\bx^{(2)} = \bx^{(1)}$ and the same virtual welfare
    $\E_{\bphi \,\sim\, \bD^{(2)}}\, [\, \bphi \cdot \bx^{(2)}(\bphi) \,]
    = \E_{\bphi \,\sim\, \bD^{(1)}}\, [\, \bphi \cdot \bx^{(1)}(\bphi) \,]$.
    Specifically, the same worst-case arrival order $\pi^{(2)*} = \pi^{(1)*}$ induces the same worst-case virtual welfare, i.e., (the revenue equivalence; \Cref{prop:revenue_equivalence}) the same {\UIVVSequentialPostedPricing} revenue in the order-oblivious model $\UIVVSPP(\bD^{(2)},\, T^{(0)}) = \UIVVSPP(\bD^{(1)},\, T^{(0)})$. This finishes the proof.\qedhere
\end{itemize}
\end{proof}

\begin{remark}[Generalizations]
Again, if we replace {\DualityRelaxationBenchmark} with another benchmark that also satisfies revenue monotonicity, \Cref{lem:truncate} can be generalized seamlessly.
\end{remark}

Given the {\truncate} reduction and \Cref{lem:truncate}, we can concentrate on {\em truncated} instances $\bF$.

\subsection{{\extend}: Characterizing the above-threshold parts of virtual values}
\label{subsec:extend}

This section shows the {\extend} reduction (see \Cref{fig:alg:extend,fig:extend} for a description and a visual aid), based on {\bf revenue-quantile curves}, which transforms a {\em truncated} instance $\bR^{(2)}$ from \Cref{subsec:truncate} into a {\em semi-linear} instance $\bR^{(3)}$ (\Cref{def:extend}).

\Cref{lem:extend} summarizes  the performance guarantees of the {\extend} reduction.

\begin{definition}[Semi-Linear Instances]
\label{def:extend}
\begin{flushleft}
A {\em truncated} instance $\bF$ is further called {\em semi-linear} when its {\em concave} revenue-quantile curves $\bR = \{R_{i}\}_{i \in [n]}$ all satisfy {\bf semi-linearity}: there exist a positive truncation point $\tau > 0$ and item-wise division points $\bkappa = \{\kappa_{i}\}_{i \in [n]} \in [0,\, 1]^{n}$ such that, for each $i \in [n]$, the derivative $R'_{i}(q) = \tau^{+}$ for all $q \in [0,\, \kappa_{i})$ and $R'_{i}(q) = \tau^{-}$ for all $q \in (\kappa_{i},\, 1]$.\footnote{\label{footnote:extend}The division point $\{q = \kappa_{i}\}$ refers to a zero measure and thus can be included to either interval $[0,\, \kappa_{i})$ or $(\kappa_{i},\, 1]$. A {\em semi-linear} curve $R_{i}$'s virtual value CDF $\phi_{i} \sim D_{i}$ is better interpreted as (\Cref{subsec:distribution}):
For an infinitesimal probability $\eps_{i} \to 0^{+}$, $\Pr [\phi_{i} = \tau^{-}] = 1 - \kappa_{i}$, $\Pr [\phi_{i} = \tau^{+}] = \kappa_{i} - \eps_{i}$, and $\Pr [\phi_{i} = R_{i}(0) / \eps_{i} + \tau] = \eps_{i}$. (If $\kappa_{i} = 0$, we instead consider $\Pr [\phi_{i} = \tau^{-}] = 1 - \eps_{i}$ and $\Pr [\phi_{i} = R_{i}(0) / \eps_{i} + \tau] = \eps_{i}$; then all later proofs are symmetric.)}
Notably, this condition itself implies {\bf regularity} and {\bf truncatedness} (\Cref{def:regular,def:truncate}).
\end{flushleft}
\end{definition}

\afterpage{
\begin{figure}[t]
    \centering
    \begin{mdframed}
    Reduction $\extend(\bR^{(2)},\, T^{(0)})$

    \begin{flushleft}
    {\bf Input:}
    The {\em truncated} instance $\bR^{(2)} = \{R_{i}^{(2)}\}_{i \in [n]}$ (from \Cref{subsec:truncate}) and the threshold $T^{(0)} > 0$ for the {\em original} instance $\bR = \bR^{(0)}$ (from \Cref{subsec:iron}).
    
    \vspace{.05in}
    {\bf Output:} A {\em semi-linear} instance $\bR^{(3)} = \{R_{i}^{(3)}\}_{i \in [n]}$.

    \begin{enumerate}
        \item\label{alg:extend}
        {\bf Return} $\bR^{(3)} = \{R_{i}^{(3)}\}_{i \in [n]}$; each curve $R_{i}^{(3)}(q) \eqdef T^{(0)+} \cdot (q - \kappa_{i}^{(2)}) + R_{i}^{(2)}(\kappa_{i}^{(2)})$ for $q \in [0,\, \kappa_{i}^{(2)}]$ and $R_{i}^{(3)}(q) \eqdef R_{i}^{(2)}(q) = T^{(0)-} \cdot (q - \kappa_{i}^{(2)}) + R_{i}^{(2)}(\kappa_{i}^{(2)})$ for $q \in (\kappa_{i}^{(2)},\, 1]$.
    \end{enumerate}
    \end{flushleft}
    \end{mdframed}
    \caption{The {\extend} reduction.
    \label{fig:alg:extend}}
\end{figure}
\begin{figure}[t]
    \subfloat[\label{fig:extend_curve_old}
    input {\em truncated} revenue-quantile curves]{
    \includegraphics[width = .48\textwidth]
    {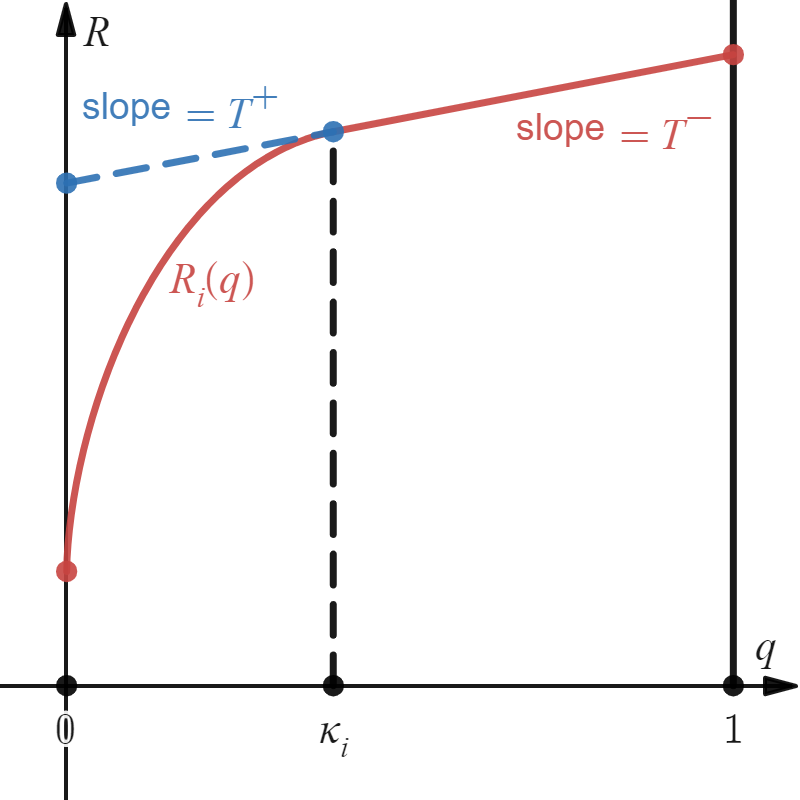}}
    \hfill
    \subfloat[\label{fig:extend_curve_new}
    output {\em semi-linear} revenue-quantile curves]{
    \includegraphics[width = .48\textwidth]
    {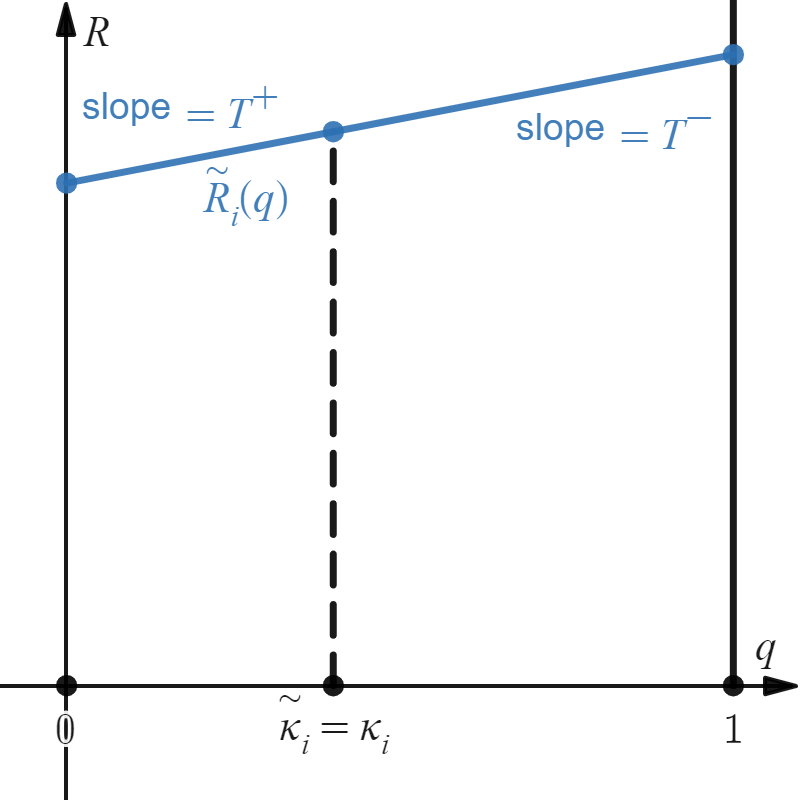}} \\
    \centering
    \subfloat[\label{fig:extend_CDF_old}
    input {\em truncated} virtual value CDF's]{
    \includegraphics[width = .48\textwidth]
    {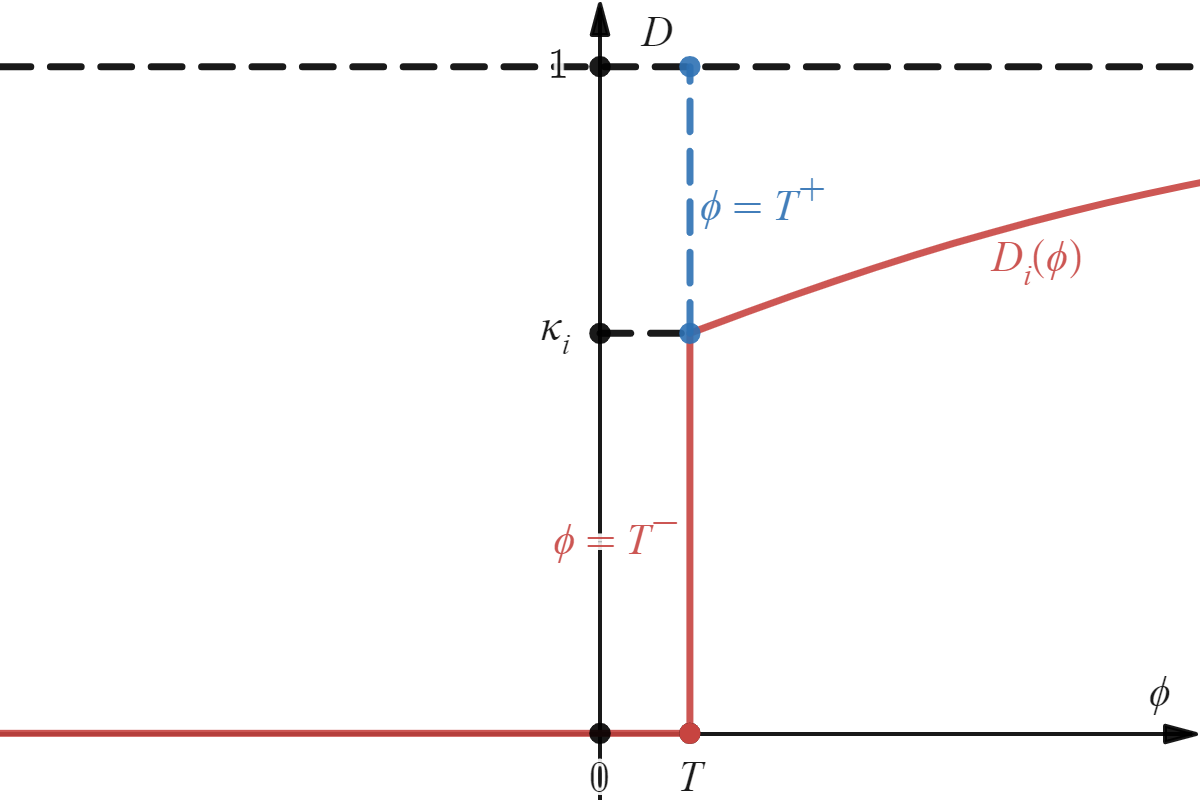}}
    \hfill
    \subfloat[\label{fig:extend_CDF_new}
    output {\em semi-linear} virtual value CDF's]{
    \includegraphics[width = .48\textwidth]
    {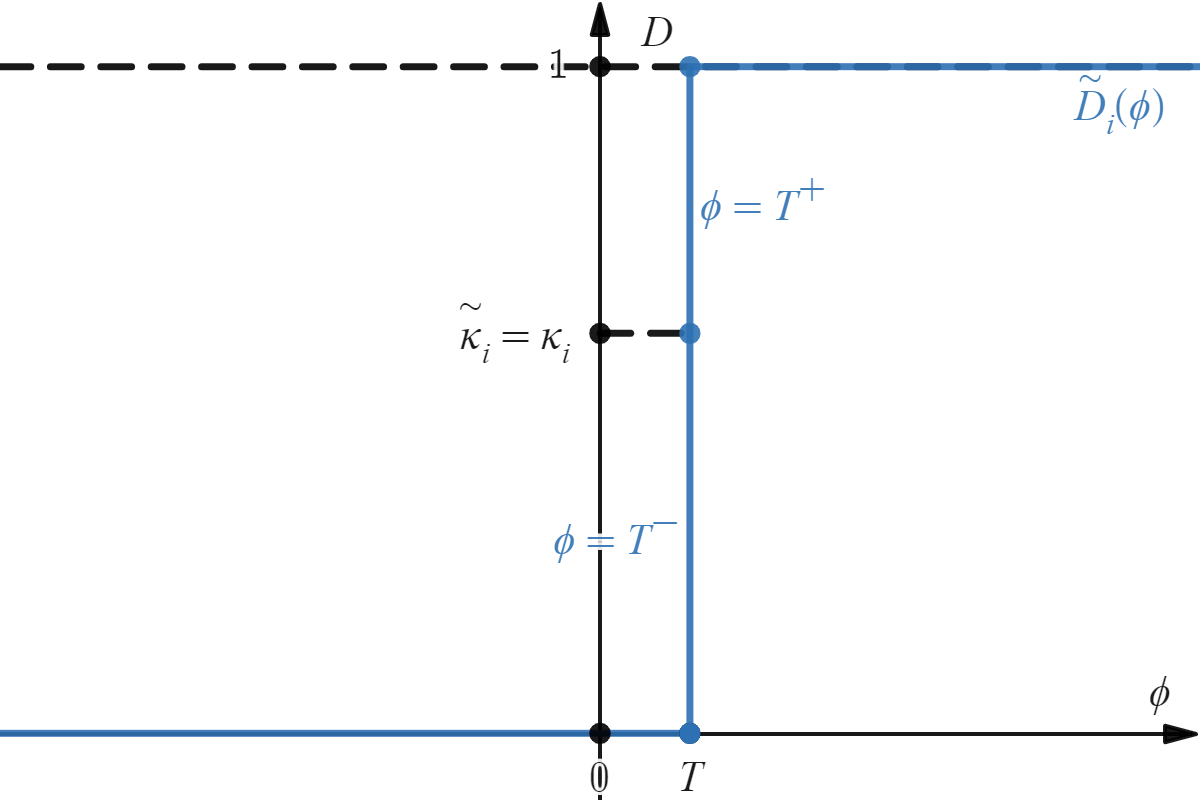}}
    \caption{Diagrams of {\extend} with virtual value CDF's and revenue-quantile curves.
    \label{fig:extend}}
\end{figure}
\clearpage}

\begin{lemma}[{\extend}]
\label{lem:extend}
\begin{flushleft}
The reduction $\bR^{(3)} \gets \extend(\bR^{(2)})$ outputs a semi-linear instance $\bR^{(3)}$ such that $\DRB(\bR^{(3)}) \geq \DRB(\bR^{(2)})$ and $\UIVVSPP(\bR^{(3)},\, T^{(0)}) = \UIVVSPP(\bR^{(2)},\, T^{(0)})$.
\end{flushleft}
\end{lemma}

\begin{proof}
Each output curve $R_{i}^{(3)}$ essentially is the slope-$T^{(0)}$ tangent of its input counterpart $R_{i}^{(2)}$; see Line~\ref{alg:extend} for the detailed construction.
It is easy to check {\bf semi-linearity}, under the same truncation point $\tau^{(3)} = \tau^{(2)} = T^{(0)} > 0$ and the same division points $\bkappa^{(3)} = \bkappa^{(2)}$. Below, we investigate the revenues from {\DualityRelaxationBenchmark} and {\UIVVSequentialPostedPricing}.
\begin{itemize}
    \item $\DRB(\bR^{(3)}) \geq \DRB(\bR^{(2)})$.
    The {\DualityRelaxationBenchmark} revenue increases by the revenue monotonicity (\Cref{lem:revenue_monotonicity}) and that the output instance stochastically dominates the input instance $\bR^{(3)} \succeq \bR^{(2)}$ by construction (Line~\ref{alg:iron}).

    \item $\UIVVSPP(\bR^{(3)},\, T^{(0)}) = \UIVVSPP(\bR^{(2)},\, T^{(0)})$. Both instances have exactly the same posted prices $R_{i}^{(3)}(\kappa_{i}^{(3)}) / \kappa_{i}^{(3)} = R_{i}^{(2)}(\kappa_{i}^{(2)}) / \kappa_{i}^{(2)}$, $\forall i \in [n]$ and sale probabilities $\kappa_{i}^{(3)} = \kappa_{i}^{(2)}$, $\forall i \in [n]$. By the definition of {\SequentialPostedPricing}, they also have exactly the same revenue.
    This finishes the proof.
    \qedhere
\end{itemize}
\end{proof}

Given the {\extend} reduction and \Cref{lem:extend}, we can concentrate on {\em semi-linear} instances $\bF$.
In terms of {\bf revenue-quantile curves}, for such an instance $\bR = \{R_{i}\}_{i \in [n]}$ with a truncation point $\tau > 0$ and item-wise division points $\bkappa = \{\kappa_{i}\}_{i \in [n]} \in [0,\, 1]$, \Cref{lem:revenue_formula} relates the revenues from ${\DualityRelaxationBenchmark} \succeq {\MyersonAuction} \succeq {\SecondPriceAuction}$.

\begin{lemma}[Revenues of Semi-Linear Instances]
\label{lem:revenue_formula}
For a semi-linear instance $\bR = \{R_{i}\}_{i \in [n]}$: \\
(i)~$\DRB(\bR) = (\sum_{i \in [n]} R_{i}(0)) + \SPA(\bR) \leq (\sum_{i \in [n]} R_{i}(0)) + \MA(\bR)$, and \\
(ii)~$\MA(\bR) = (\sum_{i \in [n]} R_{i}(0)) + \tau$.
\end{lemma}

\begin{proof}
Up to pointwise infinitesimal errors $\leq |\tau^{+} - \tau^{-}| \to 0^{+}$, (\Cref{def:extend}) each {\em semi-linear} curve $R_{i}$ takes the form of $R_{i}(q) = R_{i}(0) + \tau \cdot q$ for $q \in [0,\, 1]$ and (elementary algebra) its value CDF $F_{i}$ takes the form of $F_{i}(v) = 1 - \frac{R_{i}(0)}{v - \tau}$ for $v \geq R_{i}(0) + \tau$.

Let $h \to +\infty$ (without loss of generality, we assume $h > \tau$) and $\eps_{i} = \eps_{i}(h) \to 0^{+}$ be such that $R_{i}(\eps_{i}) / \eps_{i} = h$.
As mentioned in \Cref{subsec:distribution,footnote:extend}, $F_{i}$ is better interpreted as a truncated CDF $F_{i}^{(h)}(v) \eqdef 1 - (1 - F_{i}(v)) \cdot \indicator(v \le h)$, and the virtual value CDF $D_{i}$ as $\Pr_{\phi_{i} \,\sim\, D_{i}^{(h)}}\, [\phi_{i} = h] = \eps_{i}$ and $\Pr_{\phi_{i} \,\sim\, D_{i}^{(h)}}\, [\phi_{i} = \tau] = 1 - \eps_{i}$.
Accordingly, a ``supremum'' highest value $v_{i^{*}} = \max(\bv) = h$ yields a ``supremum'' virtual value $\varphi_{i^{*}}(v_{i^{*}}) = h$, which is always {\em above} the other values $\bv_{-i^{*}} \preceq h^{\otimes (n - 1)}$.
In contrast, a ``non-supremum'' highest value $v_{i^{*}} = \max(\bv) < h$ yields an ``infimum'' virtual value $\varphi_{i^{*}}(v_{i^{*}}) = \tau$, which is always {\em below} the other values $\bv_{-i^{*}} \succeq (R_{i}(0) + \tau)_{i \ne i^{*}}$.

By reindexing the random values $\bv = (v_{i})_{i \in [n]} \sim \bF$ in decreasing order $v_{(1)} \geq v_{(2)} \geq \dots \geq v_{(n)}$, we can formulate the {\DualityRelaxationBenchmark} revenue (up to an infinitesimal error):
\begin{align*}
    \DRB(\bR)
    & ~=~ \lim_{h \to +\infty} \E_{\bv \,\sim\, \bF^{(h)}} \Big[\, \max_{i \in [n]} \big\{ \varphi_{i}(v_{i}) \cdot \indicator(i = \argmax(\bv)) + v_{i} \cdot \indicator(i \neq \argmax(\bv)) \big\} \,\Big] \\
    & ~=~ \lim_{h \to +\infty} \E_{\bv \,\sim\, \bF^{(h)}} \Big[\, h \cdot \indicator(v_{(1)} = h) \,\Big] + \lim_{h \to +\infty} \E_{\bv \,\sim\, \bF^{(h)}} \Big[\, v_{(2)} \cdot \indicator(v_{(1)} < h) \,\Big] \\
    & ~=~ \lim_{h \to +\infty} h \cdot \Big(1 - \prod_{i \in [n]} (1 - \eps_{i})\Big) + \SPA(\bR) \\
    & ~=~ \Big(\sum_{i \in [n]} R_{i}(0)\Big) + \SPA(\bR) \\
    & ~\leq~ \Big(\sum_{i \in [n]} R_{i}(0)\Big) + \MA(\bR).
\end{align*}
The first step restates the {\DualityRelaxationBenchmark} revenue. \\
The second step uses the linearity of expectation and the above arguments. A ``supremum'' highest value $v_{(1)} = h$ yields $\varphi_{(1)}(v_{(1)}) = h \ge v_{(2)}$, so the outcome revenue $= h$. Instead, a ``non-supremum'' highest value $v_{(1)} < h$ yields $\varphi_{(1)}(v_{(1)}) = \tau \le v_{(2)}$, so the outcome revenue $= v_{(2)}$. \\
The third step: The first term uses $\Pr_{v_{i} \,\sim\, F_{i}^{(h)}}[v_{i} = h] = \eps_{i}$. And the second term uses $\SPA(\bF^{(h)}) = \E_{\bv \,\sim\, \bF^{(h)}} [v_{(2)}] = \E_{\bv \,\sim\, \bF^{(h)}} [v_{(2)} \cdot \indicator(v_{(1)} \le h)]$ (given that $F_{i}^{(h)}$'s are truncated at $h$) as well as\\
$\E_{\bv \,\sim\, \bF^{(h)}} [v_{(2)} \cdot \indicator(v_{(1)} = h)] \le \sum_{i \in [n]} \Pr_{v_{i} \,\sim\, F_{i}^{(h)}}[v_{i} = h] \cdot \E_{\bv \,\sim\, \bF^{(h)}} [v_{(2)}] = (\sum_{i \in [n]} \eps_{i}) \cdot \SPA(\bF^{(h)})$, which is infinitesimally small when $h \to +\infty \implies \eps_{i} \to 0^{+}$.

Likewise, we can formulate the {\MyersonAuction} revenue (up to an infinitesimal error):
\begin{align*}
    \MA(\bR)
    & ~=~ \lim_{h \to +\infty} \E_{\bv \,\sim\, \bF^{(h)}} \big[ \max(\bvarphi(\bv),\, 0) \big] \\
    & ~=~ \lim_{h \to +\infty} \E_{\bv \,\sim\, \bF^{(h)}} \Big[\, h \cdot \indicator(v_{(1)} = h) \,\Big]
    + \lim_{h \to +\infty} \E_{\bv \,\sim\, \bF^{(h)}} \Big[\, \tau \cdot \indicator(v_{(1)} < h) \,\Big] \\
    & ~=~ \lim_{h \to +\infty} h \cdot \Big(1 - \prod_{i \in [n]} (1 - \eps_{i})\Big)
    + \lim_{h \to +\infty} \tau \cdot \prod_{i \in [n]} (1 - \eps_{i}) \\
    & ~=~ \Big(\sum_{i \in [n]} R_{i}(0)\Big) + \tau.
\end{align*}
The first step restates the {\MyersonAuction} revenue.\\
The second step uses the linearity of expectation and the above arguments; a ``supremum'' $v_{(1)} = h$ yields a (nonnegative) virtual value of $\varphi_{(1)}(v_{(1)}) = h$, and a ``non-supremum'' highest value $v_{(1)} < h$ yields a (nonnegative) virtual value of $\varphi_{(1)}(v_{(1)}) = \tau$. \\
The third step uses $\Pr_{v_{i} \,\sim\, F_{i}^{(h)}}[v_{i} = h] = \eps_{i}$.

This finishes the proof.
\end{proof}

\begin{remark}[Generalizations]
Again, if we replace {\DualityRelaxationBenchmark} with another benchmark that also satisfies revenue monotonicity, \Cref{lem:extend} can be generalized seamlessly.
\end{remark}

\afterpage{
\begin{figure}[t]
    \centering
    \begin{mdframed}
    Reduction $\scale(\bR^{(3)},\, \bR^{(0)})$

    \begin{flushleft}
    {\bf Input:} The {\em semi-linear} instance $\bR^{(3)} = \{R_{i}^{(3)}\}_{i \in [n]}$ (from \Cref{subsec:extend}) and the {\em original} instance $\bR^{(0)} = \{R_{i}^{(0)}\}_{i \in [n]}$ (from \Cref{subsec:iron}).
    
    \vspace{.05in}
    {\bf Output:} A {\em semi-linear} instance $\bR^{(4)} = \{R_{i}^{(4)}\}_{i \in [n]}$. 

    \begin{enumerate}
        \item\label{alg:scale:construct}
        Define a family of instances $\bR^{[\gamma]} = \{R_{i}^{[\gamma]}\}_{i \in [n]}$, parameterized by $\gamma \in [0,\, 1]$, as $R_{i}^{[\gamma]}(q) = R_{i}^{(3)}(q) - \gamma \cdot R_{i}^{(3)}(0)$ for $q \in [0,\, 1]$ and $i \in [n]$.
        
        \item\label{alg:scale:output}
        {\bf Return} $\bR^{(4)} = \bR^{[\gamma^{*}]}$; the target parameter $\gamma^{*} \in [0,\, 1]$ is chosen arbitrarily from the set $\{\gamma \in [0,\, 1] \mid \DRB(\bR^{[\gamma]}) = \DRB(\bR^{(0)})\}$.
    \end{enumerate}
    \end{flushleft}
    \end{mdframed}
    \caption{The {\scale} reduction.
    \label{fig:alg:scale}}
\end{figure}
\begin{figure}[t]
    \centering
    \subfloat[\label{fig:scale_curve_old}
    input {\em semi-linear} revenue-quantile curves]{
    \includegraphics[width = .48\textwidth]
    {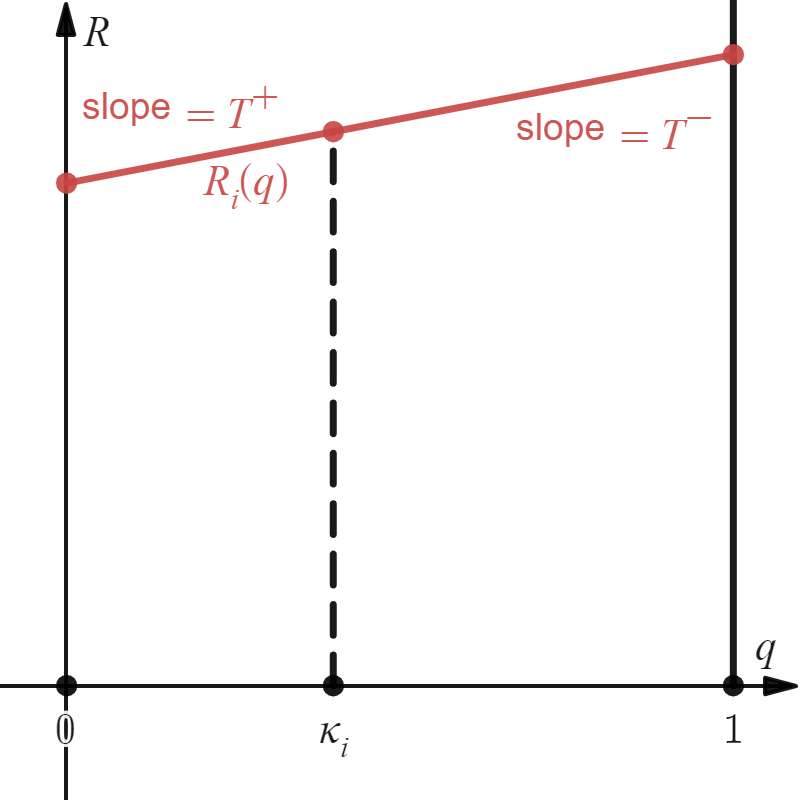}}
    \hfill
    \subfloat[\label{fig:scale_curve_new}
    output {\em semi-linear} revenue-quantile curves]{
    \includegraphics[width = .48\textwidth]
    {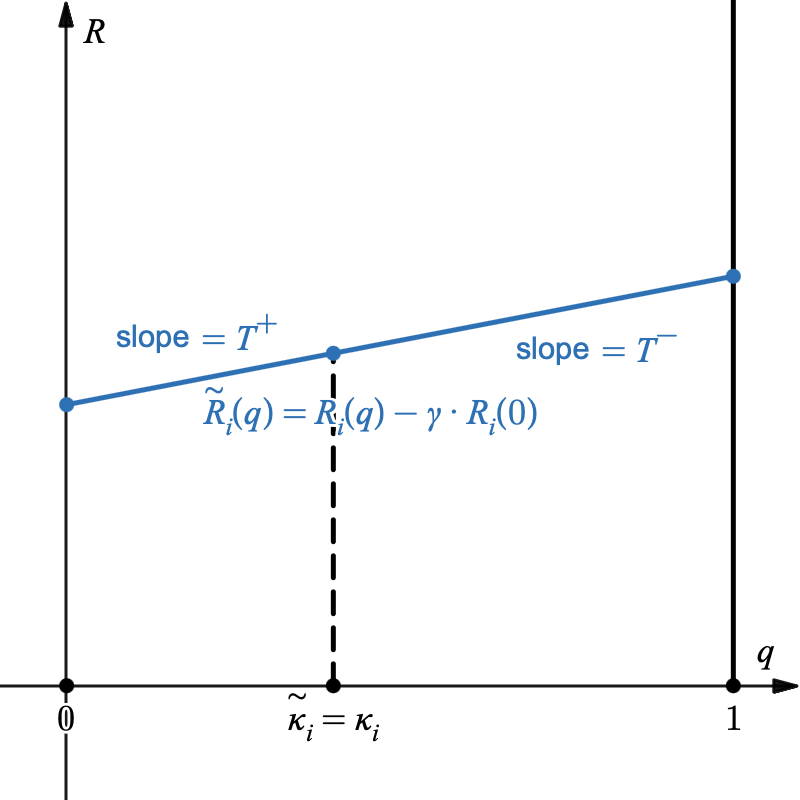}}
    \caption{Diagrams of {\scale} with revenue-quantile curves.
    \label{fig:scale}}
\end{figure}
\clearpage}

\subsection{{\scale}: Reducing generic items to semi-linear items}
\label{subsec:scale}

This section shows the {\scale} reduction (see \Cref{fig:alg:scale,fig:scale} for a description and a visual aid), based on {\bf revenue-quantile curves}, which converts the interim {\em semi-linear} instance $\bR^{(3)}$ from \Cref{subsec:extend} into another {\em semi-linear} instance $\bR^{(4)}$ that, in comparison with the {\em original} instance $\bR = \bR^{(0)}$ from \Cref{subsec:iron}, generates an equal {\DualityRelaxationBenchmark} revenue and a smaller {\UIVVSequentialPostedPricing} revenue, as desired.

\Cref{lem:scale} summarizes the performance guarantees of the {\scale} reduction.

\begin{lemma}[{\scale}]
\label{lem:scale}
\begin{flushleft}
The reduction $\bR^{(4)} \gets \scale(\bR^{(3)},\, \bR^{(0)})$ outputs a semi-linear instance $\bR^{(4)}$ such that $\DRB(\bR^{(4)}) = \DRB(\bR^{(0)})$ and $\UIVVSPP(\bR^{(4)},\, T^{(4)}) \leq \UIVVSPP(\bR^{(0)},\, T^{(0)})$.
\end{flushleft}
\end{lemma}

\begin{proof}
The {\scale} reduction first constructs a family of instances $\bR^{[\gamma]}$ parameterized by $\gamma \in [0,\, 1]$.
Each interim curve $R_{i}^{[\gamma]}(q) = R_{i}^{(3)}(q) - \gamma \cdot R_{i}^{(3)}(0)$, where $\gamma \in [0,\, 1]$ and $i \in [n]$, simply pointwise translates its input counterpart $R_{i}^{(3)}$ by a distance of $-\gamma \cdot R_{i}^{(3)}(0)$.
Obviously, each interim instance $\bR^{[\gamma]}$ for $\gamma \in [0,\, 1]$ preserves {\bf semi-linearity}. Below, we check the existence of the target parameter $\gamma^{*} \in [0,\, 1]$ (Line~\ref{alg:scale:output}) and, simultaneously, examine the revenues from {\DualityRelaxationBenchmark} and {\UIVVSequentialPostedPricing}.
\begin{itemize}
    \item $\DRB(\bR^{(4)}) = \DRB(\bR^{[\gamma^{*}]}) = \DRB(\bR^{(0)})$.
    By \Cref{lem:revenue_formula}, each interim instance $\bR^{[\gamma]}$ yields a {\DualityRelaxationBenchmark} revenue of $\DRB(\bR^{[\gamma]}) = (\sum_{i \in [n]} R_{i}^{[\gamma]}(0)) + \SPA(\bR^{[\gamma]})$.
    It is easy to see that $\DRB(\bR^{[\gamma]})$ and $\SPA(\bR^{[\gamma]})$, i.e., the {\SecondPriceAuction} revenue, as two functions in $\gamma \in [0,\, 1]$ are continuous and decreasing.
    In case of $\gamma_{\min} = 0$, we know from \Cref{lem:iron,lem:truncate,lem:extend} that $\DRB(\bR^{[\gamma_{\min}]}) = \DRB(\bR^{(3)}) \geq \DRB(\bR^{(0)})$.
    And in case of $\gamma_{\max} = 1$, we have $\DRB(\bR^{[\gamma_{\max}]}) \leq 2 \cdot (\sum_{i \in [n]} R_{i}^{[\gamma_{\max}]}(0)) + T^{(0)} = T^{(0)} = \alpha \cdot \DRB(\bR^{(0)}) < \DRB(\bR^{(0)})$, following \Cref{lem:revenue_formula} and that $\alpha \in (0,\, 1)$.
    Given these, the intermediate value theorem guarantees the existence of a target parameter $\gamma^{*} \in [0,\, 1]$ such that $\DRB(\bR^{[\gamma^{*}]}) = \DRB(\bR^{(0)})$.
    
    \item $\UIVVSPP(\bR^{(0)},\, T^{(0)}) \ge \UIVVSPP(\bR^{(4)},\, T^{(4)})$.
    To see this, we first note that
    \begin{align*}
        \UIVVSPP(\bR^{(0)},\, T^{(0)})
        & ~=~ \UIVVSPP(\bR^{(3)},\, T^{(0)}), \\
        \UIVVSPP(\bR^{(4)},\, T^{(4)})
        & ~=~ \UIVVSPP(\bR^{[\gamma^{*}]},\, T^{(0)}).
    \end{align*}
    The first equation is a combination of \Cref{lem:iron,lem:truncate,lem:extend}.
    The second equation holds because $\bR^{(4)} \equiv \bR^{[\gamma^{*}]}$ has the same {\UIVV}-threshold as the {\em original} instance $\bR = \bR^{(0)}$, namely $T^{(4)} = \alpha \cdot \DRB(\bR^{(4)}) = \alpha \cdot \DRB(\bR^{(0)}) = T^{(0)}$.
    
    Regarding $\UIVVSPP(\bR^{(3)},\, T^{(0)})$, we denote by $\pi^{(3)} \in \Pi_{n}$ the corresponding worst-case arrival order. Using the shorthand $\bar{i} = \pi^{(3)}(i)$ for $i \in [n]$, (\Cref{fig:scale_curve_old}) we can deduce that
    \begin{align*}
        \UIVVSPP(\bR^{(3)},\, T^{(0)})
        & ~=~ \sum_{i \in [n]} R_{\bar{i}}^{(3)}(\kappa_{\bar{i}}^{(3)}) \cdot \prod_{j \in [i - 1]} (1 - \kappa_{\bar{j}}^{(3)}) \\
        & ~\ge~ \sum_{i \in [n]} R_{\bar{i}}^{[\gamma^{*}]}(\kappa_{\bar{i}}^{[\gamma^{*}]}) \cdot \prod_{j \in [i - 1]} (1 - \kappa_{\bar{j}}^{[\gamma^{*}]}) \\
        & ~\ge~ \min_{\pi \in \Pi_{n}} \sum_{i \in [n]} R_{\pi(i)}^{[\gamma^{*}]}(\kappa_{\pi(i)}^{[\gamma^{*}]}) \cdot \prod_{j \in [i - 1]} (1 - \kappa_{\pi(j)}^{[\gamma^{*}]}) \\
        & ~=~ \UIVVSPP(\bR^{[\gamma^{*}]},\, T^{(0)}).
    \end{align*}
    The first step simulates the online decisions; each $i$-th arrival buyer $\bar{i}$ takes the item at price $R_{\bar{i}}^{(3)}(\kappa_{\bar{i}}^{(3)}) / \kappa_{\bar{i}}^{(3)}$ with probability $\kappa_{\bar{i}}^{(3)}$, provided that no earlier buyer had already taken it. The second step holds since (\Cref{fig:scale_curve_old}) {\scale} guarantees that $\bR^{(3)} \succeq \bR^{[\gamma^{*}]}$ and $\bkappa^{(3)} \succeq \bkappa^{[\gamma^{*}]}$. The third step simply replaces $\pi^{(3)}$ by the worst-case order regarding $\UIVVSPP(\bR^{[\gamma^{*}]},\, T^{(0)})$. The last step reuses the arguments for the first step.
    
    Combining everything finishes the proof.\qedhere
\end{itemize}
\end{proof}

Given the {\scale} reduction and \Cref{lem:scale}, we can concentrate on {\em semi-linear} instances $\bF$.

\subsection{{\perturb}: Reducing semi-linear items to linear items}
\label{subsec:perturb}

This section gives the {\perturb} reduction (see \Cref{fig:alg:perturb,fig:perturb} for a description and a visual aid), based on {\bf revenue-quantile curves}, which transforms a {\em semi-linear} instance $\bR$ from \Cref{subsec:scale} into an {\em linear} instance $\tilde{\bR}$ (\Cref{def:perturb}).

\Cref{lem:perturb} concludes the performance guarantees of the {\perturb} reduction.

\begin{definition}[Linear Instances]
\label{def:perturb}
A {\em semi-linear} instance $\bF$ (with a truncation point $\tau > 0$ and item-wise division points $\bkappa = \{\kappa_{i}\}_{i \in [n]} \in [0,\, 1]^{n}$) is further called {\em linear} when its revenue-quantile curves $\bR = \{R_{i}\}_{i \in [n]}$ all satisfy {\bf linearity}:
either $R'_{i}(q) = \tau^{+}$ for $q \in [0,\, 1]$ and $i \in [n]$ (i.e., $\bkappa = \ones$) or $R'_{i}(q) = \tau^{-}$ for $q \in [0,\, 1]$ and $i \in [n]$ (i.e., $\bkappa = \zeros$).
\end{definition}

\afterpage{
\begin{figure}[t]
    \centering
    \begin{mdframed}
    Reduction $\perturb(\bR)$

    \begin{flushleft}
    {\bf Input:} The {\em semi-linear} instance $\bR = \{R_{i}\}_{i \in [n]}$ (from \Cref{subsec:scale}) with the threshold $T > 0$ and the division points $\bkappa = \{\kappa_{i}\}_{i \in [n]}$.
    \hfill \OliveGreen{$\triangleright$ $R_{i}(q) = T^{+} \cdot (q - \kappa_{i}) + R_{i}(\kappa_{i})$ for $q \in [0,\, \kappa_{i}]$.}
    
    \hfill \OliveGreen{$\triangleright$ $R_{i}(q) = T^{-} \cdot (q - \kappa_{i}) + R_{i}(\kappa_{i})$ for $q \in (\kappa_{i},\, 1]$.}
    
    \vspace{.05in}
    {\bf Output:} A {\em linear} instance $\tilde{\bR} = \{\tilde{R}_{i}\}_{i \in [n]}$.

    \begin{enumerate}
        \item\label{alg:perturb_positive}
        Let $\bR^{+} = \{R_{i}^{+}\}_{i \in [n]}$ be $R_{i}^{+}(q) \eqdef T^{+} \cdot (q - \kappa_{i}) + R_{i}(\kappa_{i})$ for $q \in [0,\, 1]$ and $i \in [n]$.
        
        \item\label{alg:perturb_negative}
        Let $\bR^{-} = \{R_{i}^{-}\}_{i \in [n]}$ be $R_{i}^{-}(q) \eqdef T^{-} \cdot (q - \kappa_{i}) + R_{i}(\kappa_{i})$ for $q \in [0,\, 1]$ and $i \in [n]$.
        
        \item\label{alg:perturb_output}
        {\bf Return} $\argmin\, \{\, \UIVVSPP(\tilde{\bR},\, T) \,\mid\, \tilde{\bR} \in \{ \bR^{+},\, \bR^{-} \} \,\}$.
    \end{enumerate}
    \end{flushleft}
    \end{mdframed}
    \caption{The {\perturb} reduction.
    \label{fig:alg:perturb}}
\end{figure}
\begin{figure}[t]
    \centering
    \subfloat[\label{fig:perturb_curve_old}
    input {\em semi-linear} revenue-quantile curves]{
    \includegraphics[width = .4\textwidth]
    {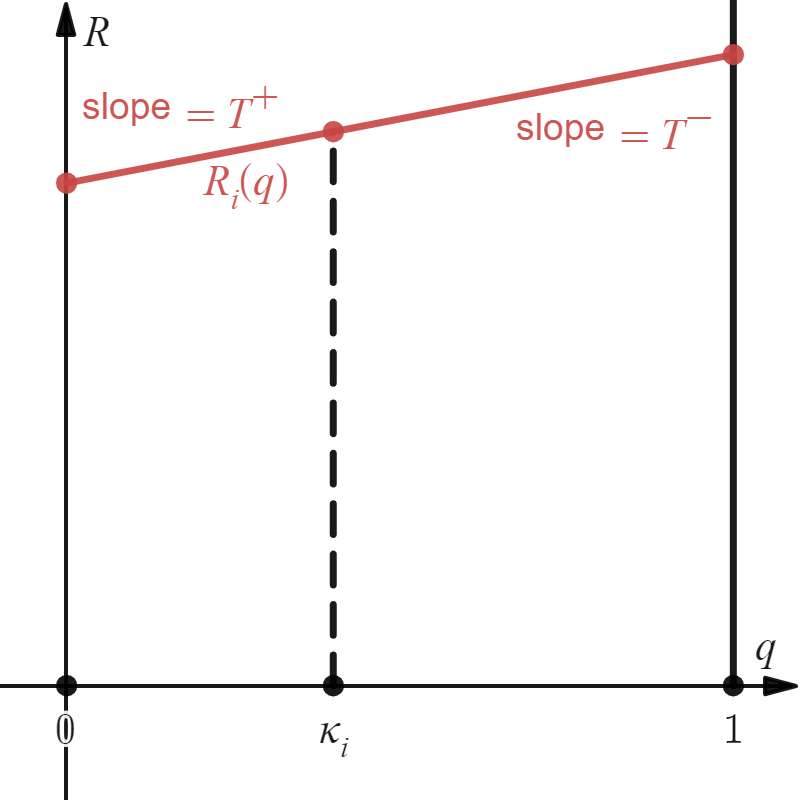}} \\
    \subfloat[\label{fig:perturb_curve_positive}
    ``positive'' {\em linear} revenue-quantile curves]{
    \includegraphics[width = .4\textwidth]
    {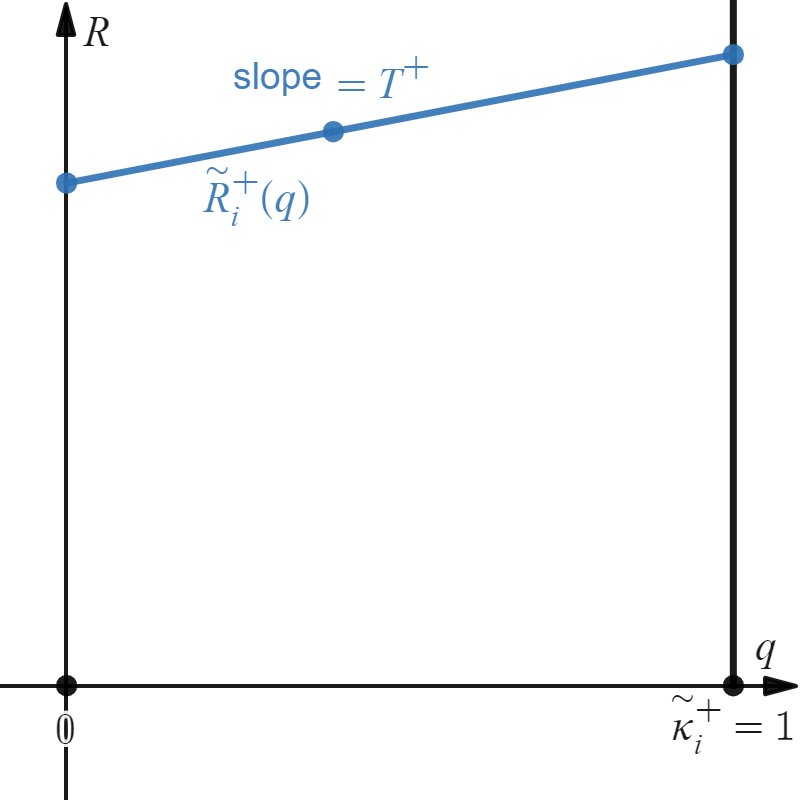}}
    \hfill
    \subfloat[\label{fig:perturb_curve_negative}
    ``negative'' {\em linear} revenue-quantile curves]{
    \includegraphics[width = .4\textwidth]
    {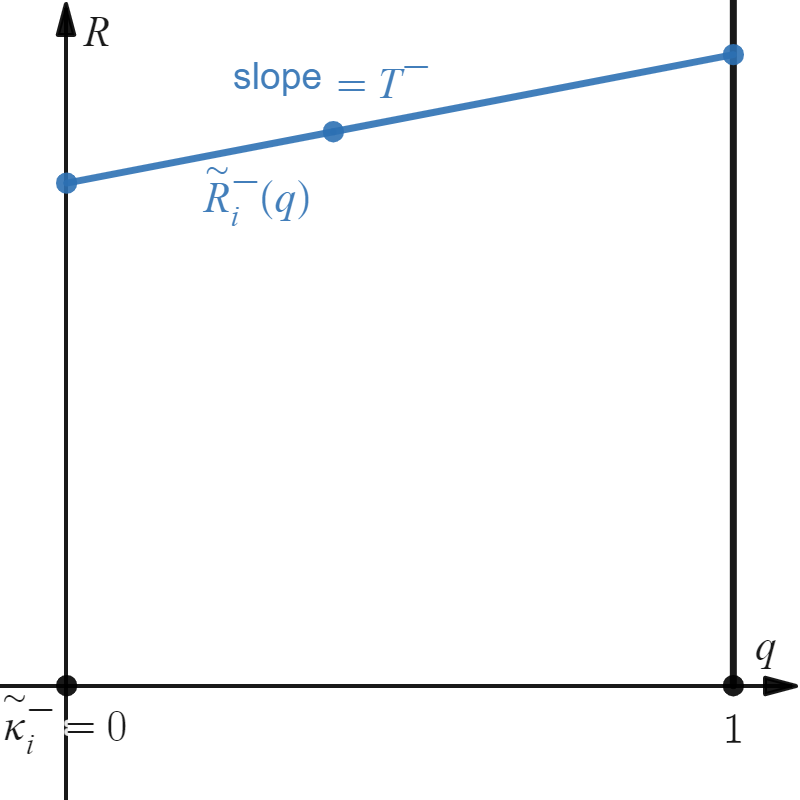}}
    \caption{Diagrams of the {\perturb} reduction with revenue-quantile curves.
    \label{fig:perturb}}
\end{figure}
\clearpage}

\begin{lemma}[{\perturb}]
\label{lem:perturb}
\begin{flushleft}
The reduction $\tilde{\bR} \gets \perturb(\bR)$ outputs a linear instance $\tilde{\bR}$ such that $\DRB(\tilde{\bR}) = \DRB(\bR)$ and $\UIVVSPP(\tilde{\bR},\, \tilde{T}) \leq \UIVVSPP(\bR,\, T)$.
\end{flushleft}
\end{lemma}

\begin{proof}
The output instance $\tilde{\bR}$ is chosen between two candidates (Line~\ref{alg:perturb_output}), the ``positive'' candidate $\bR^{+} = \{R_{i}^{+}\}_{i \in [n]}$ that $R_{i}^{+}(q) \eqdef T^{+} \cdot (q - \kappa_{i}) + R_{i}(\kappa_{i})$ and the ``negative'' candidate $\bR^{-} = \{R_{i}^{-}\}_{i \in [n]}$ that $R_{i}^{-}(q) \eqdef T^{-} \cdot (q - \kappa_{i}) + R_{i}(\kappa_{i})$.
Both candidates are identical to the input instance $\bR$ in revenue-quantile curves up to pointwise infinitesimal errors $\leq |T^{+} - T^{-}| \to 0^{+}$,
(\Cref{lem:revenue_formula}) hence the same {\DualityRelaxationBenchmark} revenue $\DRB(\bR^{+}) = \DRB(\bR^{-}) = \DRB(\bR)$ and the same {\UIVV}-threshold $= T$.
Now we can easily check {\bf linearity} for both candidates and thus the output instance $\tilde{\bR}$, under the same threshold $= T$ and the respective division points $\bkappa^{+} = \ones$ and $\bkappa^{-} = \zeros$. Below, we investigate the revenues from {\DualityRelaxationBenchmark} and {\UIVVSequentialPostedPricing}.
\begin{itemize}
    \item $\DRB(\tilde{\bR}) = \DRB(\bR)$.
    Both candidates $\bR^{+}$ and $\bR^{-}$ preserve the {\DualityRelaxationBenchmark} revenue and the threshold $= T$; so does the output instance $\DRB(\tilde{\bR}) = \DRB(\bR)$ and $\tilde{T} = T$.
    
    \item $\UIVVSPP(\tilde{\bR},\, \tilde{T}) \leq \UIVVSPP(\bR,\, T)$.
    By construction (Line~\ref{alg:perturb_output}) and given that $\tilde{T} = T$, this equation reduces to $\min\, \{\, \UIVVSPP(\bR^{+},\, T),\, \UIVVSPP(\bR^{-},\, T) \,\} \leq \UIVVSPP(\bR,\, T)$.
    For the input instance $\bR$, let $\pi^{*} \in \Pi_{n}$ and $\bx^{*}$ be its worst-case arrival order and its resulting allocation rule.
    Using the shorthand $\bar{i} = \pi^{*}(i)$ for $i \in [n]$, the input instance $\bR$ yields a worst-case virtual welfare
    \begin{align}
        \E_{\bphi \,\sim\, \bD} \big[\, \bphi \cdot \bx^{*}(\bphi) \,\big]
        & ~=~ \sum_{i} R_{\bar{i}}(\kappa_{\bar{i}}) \cdot \prod_{j < i} (1 - \kappa_{\bar{j}})
        \notag \\
        & ~=~ \sum_{i} (R_{\bar{i}}(0) + T^{+} \cdot \kappa_{\bar{i}}) \cdot \prod_{j < i} (1 - \kappa_{\bar{j}})
        \notag \\
        & ~=~ \sum_{i} (R_{\bar{i}}(0) + T \cdot \kappa_{\bar{i}}) \cdot \prod_{j < i} (1 - \kappa_{\bar{j}})
        \notag \\
        & ~=~ \Big(\sum_{i} R_{\bar{i}}(0) \cdot \prod_{j < i} (1 - \kappa_{\bar{j}})\Big)
        + T \cdot \underbrace{\Big(1 - \prod_{i} (1 - \kappa_{\bar{i}})\Big)}_{\kappa}
        \label{eq:OOSPP}
    \end{align}
    The first step simulates the online decisions; each $i$-th arrival buyer $\bar{i}$ takes the item at price $R_{\bar{i}}(\kappa_{\bar{i}}) / \kappa_{\bar{i}}$ with probability $\kappa_{\bar{i}}$, provided that no earlier buyer had already taken it.
    The second step follows from {\bf semi-linearity} (\Cref{def:extend}).
    The third step drops the pointwise infinitesimal errors $|T^{+} - T| \cdot \kappa_{\bar{i}} \to 0^{+}$.
    And the last step is elementary algebra.
    
    We would reuse $\pi^{*}$ and $\bx^{*}$ for both candidates $\bR^{+}$ and $\bR^{-}$ (regardless whether $\pi^{*}$ remains the revenue-worst arrival order). Since $\bkappa^{+} = \ones$ and $\bkappa^{-} = \zeros$ (Lines~\ref{alg:perturb_positive} and \ref{alg:perturb_negative}), the ``positive'' candidate $\bR^{+}$ likewise yields a virtual welfare of
    \begin{align}
        \E_{\bphi \,\sim\, \bD^{+}} \big[\, \bphi \cdot \bx^{*}(\bphi) \,\big]
        & ~=~ \Big(\sum_{i} R_{\bar{i}}(0) \cdot \prod_{j < i} (1 - \kappa_{\bar{j}}^{+})\Big)
        + T \cdot \Big(1 - \prod_{i} (1 - \kappa_{\bar{i}}^{+})\Big)
        \notag \\
        & ~=~ R_{\bar{1}}(0) + T,
        \label{eq:VUSPP_positive:1}
    \end{align}
    and the ``negative'' candidate $\bR^{-}$ likewise yields a virtual welfare of
    \begin{align}
        \E_{\bphi \,\sim\, \bD^{-}} \big[\, \bphi \cdot \bx^{*}(\bphi) \,\big]
        & ~=~ \sum_{i} \Big(R_{\bar{i}}(0) \cdot \prod_{j < i} (1 - \kappa_{\bar{j}}^{-})\Big)
        + T \cdot \Big(1 - \prod_{i} (1 - \kappa_{\bar{i}}^{-})\Big)
        \notag \\
        & ~=~ \sum_{i} R_{\bar{i}}(0).
        \label{eq:VUSPP_negative:1}
    \end{align}
    Let $\kappa \eqdef 1 - \prod_{j \in [n]} (1 - \kappa_{j}) \in [0, 1]$.
    Now we can check the claimed equation as follows:
    \begin{align*}
        & \min\, \Big\{\, \UIVVSPP(\bR^{+},\, T),\, \UIVVSPP(\bR^{-},\, T) \,\Big\} \\
        & ~\leq~ \kappa \cdot \UIVVSPP(\bR^{+},\, T) + (1 - \kappa) \cdot \UIVVSPP(\bR^{-},\, T) \\
        & ~=~ \kappa \cdot \min_{\pi \,\in\, \Pi_{n}}\, \Big\{\, \E_{\bphi \,\sim\, \bD^{+}} \big[\, \bphi \cdot \bx(\bphi) \,\big] \,\Big\} + (1 - \kappa) \cdot \min_{\pi \,\in\, \Pi_{n}}\, \Big\{\, \E_{\bphi \,\sim\, \bD^{-}} \big[\, \bphi \cdot \bx(\bphi) \,\big] \,\Big\} \\
        & ~\leq~ \kappa \cdot \E_{\bphi \,\sim\, \bD^{+}} \big[\, \bphi \cdot \bx^{*}(\bphi) \,\big]
        + (1 - \kappa) \cdot \E_{\bphi \,\sim\, \bD^{-}} \big[\, \bphi \cdot \bx^{*}(\bphi) \,\big] \\
        & ~=~ \kappa \cdot \Big(R_{\bar{1}}(0) + T\Big)
        + (1 - \kappa) \cdot \Big(\sum_{i} R_{\pi(i)}(0)\Big) \\
        & ~=~ R_{\bar{1}}(0)
        + (1 - \kappa) \cdot \Big(\sum_{i > 1} R_{\bar{i}}(0)\Big)
        + \kappa \cdot T \\
        & ~\leq~ \Big(\sum_{i} R_{\bar{i}}(0) \cdot \prod_{j < i} (1 - \kappa_{\bar{j}})\Big)
        + \kappa \cdot T \\
        & ~=~ \E_{\bphi \,\sim\, \bD} \big[\, \bphi \cdot \bx^{*}(\bphi) \,\big]
        ~=~ \UIVVSPP(\bR,\, T).
    \end{align*}
    The second step and the last step both apply the revenue equivalence (\Cref{prop:revenue_equivalence}).
    And the other steps are elementary algebra; for the last second step in particular, we notice that $1 - \kappa = \prod_{j \in [n]} (1 - \kappa_{\bar{j}}) \leq \prod_{j < i} (1 - \kappa_{\bar{j}})$ for each $i > 1$.
    This finishes the proof.\qedhere
\end{itemize}
\end{proof}

Given the {\perturb} reduction and \Cref{lem:perturb}, in the rest of \Cref{sec:DTM} we can concentrate on {\em linear} instances $\tilde{\bR}$.

\subsection{Proof of the upper bounds}
\label{subsec:upper_bounds}

Now we are ready to prove the upper-bound parts of \Cref{thm:UIVVIP} (Item~(i)) and \Cref{thm:DRB_SPP}, both of which are restated below for ease of reference.
(The following proof of \Cref{thm:DRB_SPP} is nonconstructive. As a complement, we characterize the worst-case instances for every $n \geq 3$ in \Cref{sec:worst_case}.)

\begin{restate}[{\Cref{thm:UIVVIP}}]
\begin{flushleft}
Against {\DualityRelaxationBenchmark}: \\
(i)~{\UniformIronedVirtualValueItemPricing} with a {\UIVV}-threshold $T = \frac{1}{3} \cdot \DRB$ of one third of the {\DualityRelaxationBenchmark} revenue achieves a tight $\calC_{\DRB / \UIVVIP} = 3$ approximation.
\end{flushleft}
\end{restate}

\begin{restate}[{\Cref{thm:DRB_SPP}}]
\begin{flushleft}
Against {\DualityRelaxationBenchmark}: \\
(i)~{\UniformIronedVirtualValueSequentialPostedPricing} with a {\UIVV}-threshold $T = \frac{1}{3} \cdot \DRB$ of one third of the {\DualityRelaxationBenchmark} revenue achieves a $\calC_{\DRB / \UIVVSPP} = 3$ approximation. \\
(ii)~This ratio is optimal for \textbf{deterministic ironed-virtual-value-based} and/or \textbf{(arbitrary) uniform-ironed-virtual-value} stopping rules in the order-oblivious model.\textsuperscript{\ref{footnote:tie}}
\end{flushleft}
\end{restate}

\begin{proof}[Proof of \Cref{thm:DRB_SPP} (The Upper-Bound Part)]
Following \Cref{subsec:iron,subsec:truncate,subsec:extend,subsec:scale,subsec:perturb}, it suffices to consider {\em linear} instances $\bR$.
Such an instance has a {\UIVV}-threshold $T = \alpha \cdot \DRB(\bR) $. 
Following \Cref{lem:revenue_formula}, the {\DualityRelaxationBenchmark} revenue is given by
\[
    \DRB(\bR ) \leq 2 \sum_{i\in [n]} R_i(0) + T.
\]
By elementary algebra, we deduce that 
\[
    \sum_{i\in [n]} R_i(0)
    \geq \frac{\DRB(\bR ) - T}{2}
    = \frac{1 - \alpha}{2} \DRB(\bR).
\]

Moreover, for the {\UIVVSequentialPostedPricing} revenue, there are two possibilities.
If $\bR $ is a ``positive'' {\em linear} instance, we know from \Cref{eq:VUSPP_positive:1} that
\[
    \UIVVSPP(\bR,\, T) \geq T = \alpha \cdot \DRB(\bR).
\]
If $\bR $ is a ``negative'' {\em linear} instance, we know from \Cref{eq:VUSPP_negative:1} that
\[
    \UIVVSPP(\bR,\, T) = \sum_{i\in [n]} R_i(0) \geq\frac{1-\alpha}{2}\DRB(\bR ).
\]
Towards the best revenue guarantee, we must have $\UIVVSPP(\bR,\, T) \geq \frac{1}{3}\DRB(\bR)$ in both cases, by choosing $\alpha = \frac{1}{3}$.
This accomplishes the upper-bound part of \Cref{thm:DRB_SPP}.
\end{proof}



\begin{proof}[Proof of \Cref{thm:UIVVIP} (The Upper-Bound Part of Item~(i))]
This follows from \Cref{prop:representative} (i.e., the single-dimensional representative approach) and the upper-bound part of \Cref{thm:DRB_SPP}.
\end{proof}

\subsection{The lower-bound analyses}
\label{subsec:SPP_LB}

This section accomplishes \Cref{thm:UIVVIP} (Item~(i)) and \Cref{thm:DRB_SPP} by showing matching lower-bound instances.
(Unlike the upper-bound parts, the lower-bound parts of both theorems are incomparable. I.e., \Cref{thm:UIVVIP} considers a stronger mechanism, {\ItemPricing}, but a restricted pricing scheme, uniform-ironed-virtual-value posted prices.)

First, we prove a matching lower bound for \Cref{thm:DRB_SPP}, under \textbf{deterministic ironed-virtual-value-based} stopping rules, using even ``i.i.d.'' instances (\Cref{exp:DRB_SPP}).


\begin{example}[{Lower-Bound Instances for $\calC_{\DRB / \UIVVSPP} = \calC_{\DRB / \UIVVIP} = 3$}]
\label{exp:DRB_SPP}
\begin{flushleft}
There are $n \geq 2$ i.i.d.\ items $\bF = \{F_{i}\}_{i \in [n]}$ with a common revenue-quantile curve $R_{i}(q) = R(q) \eqdef \frac{1}{n} + q$ for $q \in [0,\, 1]$ and (elementary algebra) a common value CDF $F_{i}(v) = F(v) \eqdef 1 - \frac{1 / n}{v - 1}$ for $v \in [\frac{n + 1}{n},\, +\infty]$.
\end{flushleft}
\end{example}

\begin{proof}[Proof of \Cref{thm:DRB_SPP} (The Lower-Bound Part)]
It is easy to check that $\bF$ is a {\em linear} instance (\Cref{def:perturb}). Following \Cref{lem:revenue_formula}, it yields a {\DualityRelaxationBenchmark} revenue of
\begin{align*}
    \DRB(\bF) = n R(0) + \SPA(\bF) = 1 + \int_{0}^{+\infty} (1 - S(v)) \cdot \d v = 1 + 2 = 3.
\end{align*}
The second step formulates the {\SecondPriceAuction} revenue using the second-highest value CDF $S(v) \eqdef (F(v))^{n} + n \cdot (F(v))^{n - 1} \cdot (1 - F(v))$.
And the other steps are elementary algebra.

With a \textbf{\em deterministic ironed-virtual-value-based} stopping rule, the seller first chooses the item-wise ironed-virtual-value acceptance sets $\Phi_{i} \subseteq \R$ for $i \in [n]$. Then the adversary chooses a worst-case arrival order $\pi \in \Pi_{n}$ for this stopping rule.
For \Cref{exp:DRB_SPP}, (\Cref{footnote:extend}) each identical item's (ironed) virtual value $\phi_{i} \sim D_{i}$ only has two possibilities: $\phi_{\sf high} = \frac{R(\eps)}{\eps} = \frac{1}{n \eps} + 1$ with an infinitesimal probability $\eps \to 0^{+}$ and $\phi_{\sf low} = 1$ with probability $(1 - \eps) \to 1^{-}$. There are two cases:

\vspace{.05in}
\noindent
{\bf Case~1:}
$\exists i \in [n]$ such that $\phi_{\sf low} \in \Phi_{i}$, some acceptance set includes the low virtual value $\phi_{\sf low} = 1$. \\
The adversary can make such an item $i$ arrive first $\pi(1) = i$.
Then the seller will accept the low virtual value $\phi_{\sf low}$ with probability $\geq 1 - \eps$ (mainly from item $\pi(1) = i$) and the high virtual value $\phi_{\sf high}$ with probability $\leq \eps$, hence an expected virtual welfare $\leq \phi_{\sf low} \cdot 1 + \phi_{\sf high} \cdot \eps = 1 + \frac{1}{n} + \eps$.

\vspace{.05in}
\noindent
{\bf Case~2:}
$\nexists i \in [n]$ such that $\phi_{\sf low} \in \Phi_{i}$, no acceptance set includes the low virtual value $\phi_{\sf low} = 1$. \\
Regardless of the arrival order $\pi \in \Pi_{n}$ chosen by the adversary, the seller obtains an expected virtual welfare $\leq \phi_{\sf high} \cdot n \eps = 1 + n \eps$.

\vspace{.05in}
\noindent
So the considered stopping rule (the revenue equivalence; \Cref{prop:revenue_equivalence}) yields a revenue of at most $\leq \lim_{\eps \to 0^{+}} \max\, \{\, 1 + \frac{1}{n} + \eps,\, 1 + n \cdot \eps \,\} = 1 + \frac{1}{n}$, which can be arbitrarily close to $1^{+}$ for a large enough the number of items $n \gg 2$.
This finishes the proof.
\end{proof}

Second, we prove a matching lower bound for \Cref{thm:DRB_SPP}, under \textbf{(arbitrary) uniform-ironed-virtual-value} stopping rules; the following \Cref{exp:DRB_SPP_uniform} slightly modifies \Cref{exp:DRB_SPP}.

\begin{example}[{Lower-Bound Instances for $\calC_{\DRB / \UIVVSPP} = 3$}]
\label{exp:DRB_SPP_uniform}
\begin{flushleft}
There are $n \geq 2$ items $\bF = \{F_{i}\}_{i \in [n]}$ with revenue-quantile curves $R_{i}(q) \eqdef \frac{1}{n} + (1 + \frac{i}{n^{2}}) \cdot q$ for $q \in [0,\, 1]$ and (elementary algebra) value CDF's $F_{i}(v) \eqdef 1 - \frac{1 / n}{v - (1 + i / n^{2})}$ for $v \in [1 + \frac{1}{n} + \frac{i}{n^{2}},\, +\infty]$.
\end{flushleft}
\end{example}

\begin{proof}[Proof of \Cref{thm:DRB_SPP} (The Lower-Bound Part)]
This instance $\bF$ stochastically dominates the one in \Cref{exp:DRB_SPP}, so we know from \Cref{lem:revenue_monotonicity} that $\DRB(\bF) \geq 3$.

By an \textbf{arbitrary uniform-ironed-virtual-value}, the seller chooses a threshold $T$ (which can be randomized) and posted prices $\bp = (p_{i})_{i \in [n]}$ such that $\bar{\varphi}_{i}(p_{i}) = T$ (which can be randomized even conditioned on $T$). Then, the adversary chooses a worst-case arrival order $\pi^{*} \in \Pi_{n}$ against the distributional information of $\bp$.
Consider the anti-lexicographic order $F_{n}, F_{n - 1}, \dots, F_{1}$ (which is no worse than $\pi^{*}$) and an outcome $(T, \bp)$.
There are two cases about the item $F_{n}$:

\vspace{.05in}
\noindent
{\bf Case~1:}
$F_{n}$ gets accepted with probability $= 1$.
Then, since $F_{n}$'s value is supported on $[1 + \frac{2}{n}, +\infty]$, its outcome posted price is at most $p_{n} \leq 1 + \frac{2}{n}$, and so is the revenue $\leq 1 + \frac{2}{n}$.

\vspace{.05in}
\noindent
{\bf Case~2:}
$F_{n}$ gets accepted with probability $< 1$.
Then, since $F_{n}$'s ironed virtual value is supported on $\{1 + \frac{1}{n}, +\infty\}$, the outcome threshold is at least $T \geq 1 + \frac{1}{n}$.
Every other item $F_{i}$ for $i \in [n - 1]$ has an infinite posted price $p_{i} = +\infty$, since $\bar{\varphi}_{i}(v) = 1 + \frac{i}{n^{2}} < 1 + \frac{1}{n} \leq T$ for $v < +\infty$.
So the revenue $= \sum_{i = 1}^{n} p_{i} (1 - F_{i}(p_{i})) \prod_{j > i} F_{j}(p_{j}) = p_{n} (1 - F_{n}(p_{n})) + \sum_{i = 1}^{n - 1} R_{i}(0) F_{n}(p_{n}) = 1 + \frac{2 (1 - F_{n}(p_{n}))}{n} \leq 1 + \frac{2}{n}$, where the third step uses $F_{n}(p_{n}) = 1 - \frac{1 / n}{p_{n} - (1 + 1 / n)} \implies p_{n} = 1 + \frac{1}{n}(1 + \frac{1}{1 - F_{n}})$.

\vspace{.05in}
\noindent
To summarize, $\UIVVSPP(\bF) \leq 1 + \frac{2}{n}$, which can be arbitrarily close to $1^{+}$ for a large enough the number of items $n \gg 2$.
This finishes the proof.
\end{proof}


Finally, based on \Cref{exp:DRB_SPP}, we establish the lower-bound part of \Cref{thm:UIVVIP} (Item~(i)) that {\UIVVItemPricing} cannot achieve an approximation ratio better than $3$.

\begin{proof}[Proof of \Cref{thm:UIVVIP} (The Lower-Bound Part of Item~(i))]
The {\DualityRelaxationBenchmark} revenue is $\DRB(\bF) = 3$ and each identical item's (ironed) virtual value $\phi_{i} \sim D_{i}$ has two possibilities: $\phi_{\sf high} = \frac{1}{n \eps} + 1$ with probability $\eps \to 0^{+}$ and $\phi_{\sf low} = 1$ with probability $(1 - \eps) \to 1^{-}$.
By \Cref{def:UIVV_prices}, {\UIVVItemPricing} posts on all items $i \in [n]$
either a uniform high price $p_{\sf high} = \frac{R(\eps)}{\eps} = \frac{1}{n \eps} + 1$, hence a revenue of $p_{\sf high} \cdot (1 - (F(p_{\sf high}))^{n}) \leq \frac{p_{\sf high}}{p_{\sf high} - 1} = 1 + n \eps$,
or a uniform low price $p_{\sf low} = \frac{R(1)}{1} = \frac{n + 1}{n}$, hence a revenue of $p_{\sf low} \cdot 1 = \frac{n + 1}{n}$.

In sum, the revenue gap is at least $\geq \lim_{\eps \to 0^{+}} \frac{3}{\max\, \{\, 1 + n \eps,\, 1 + 1 / n \,\}} = \frac{3}{1 + 1 / n}$, which can be arbitrarily close to $3$ for a large enough number of items $n \gg 2$.
This finishes the proof.
\end{proof}

\begin{remark}[Tightness of \Cref{thm:UIVVIP}]
If we insist on \Cref{def:UIVV_prices} and use the {\UIVV} posted prices $p_{i} = \sup \{ v \in \RR \mid \bar{\varphi}_{i}(v) < T \}$ or $p_{i} = \sup \{ v \in \RR \mid \bar{\varphi}_{i}(v) \leq T \}$, then \Cref{exp:DRB_SPP} indicates that the revenue guarantee of $3$ is tight.
However, if we consider instead the following ``best-{\UIVV}'' posted prices $\bp^{*} = \bp^{*}(\bF)$, we cannot find a matching lower-bound instance.\footnote{Since $\bar{\varphi}_{i}$'s are monotonically increasing functions, such ``best-{\UIVV}'' posted prices $\bp^{*} = \bp^{*}(\bF)$ are well-defined.}
\begin{align*}
    \bp^{*}(\bF) \eqdef \argmax_{\bp} \{\IP(\bF,\, \bp) \mid p_{i} \in \bar{\varphi}_{i}^{-1}(T),\ \forall i \in [n]\}.
\end{align*}
E.g., let us rethink \Cref{exp:DRB_SPP}. That instance is symmetric $\bF = \{F\}^{\otimes n}$ and has a common ironed virtual value function $\bar{\varphi}(v) = 1$ for $v \in [\frac{n + 1}{n}, +\infty)$ and $\bar{\varphi}(+\infty) = +\infty$; thus the ``best-{\UIVV}'' posted prices $\bp^{*}$ are exactly the posted prices in {\OptimalItemPricing}.
It is not difficult to see that the $\bp^{*}$ for \Cref{exp:DRB_SPP} takes the form of $p_{i}^{*} = \frac{n + 1}{n}$ for exactly one item $i \in [n]$ and $p_{j}^{*} = +\infty$ for every other item $j \in [n] \setminus \{i\}$ and the resulting revenue $\UIVVIP(\bF,\, \bp^{*}) = \OIP(\bF) = 2$. Namely, the revenue guarantee against {\DualityRelaxationBenchmark} is $\frac{3}{2}$ (rather than $3$).
\end{remark}

\section{More Lower-Bound Instances}
\label{sec:example}

In this section, we prove the lower-bound parts of \Cref{thm:UIVVIP} (Item~(ii) and Item~(iii)) and \Cref{cor:OIP} (Item~(i)), both of which are restated below for ease of reference. Our lower-bound instances are inspired by \Cref{exp:DRB_SPP} from \Cref{subsec:SPP_LB}.

\begin{restate}[{\Cref{thm:UIVVIP}}]
\begin{flushleft}
{\UniformIronedVirtualValueItemPricing} achieves \\
(ii)~no better than a $\calC_{\OLP / \UIVVIP} \geq 2$ approximation to {\OptimalLotteryPricing}, and \\
(iii)~no better than a $\calC_{\OIP / \UIVVIP} \geq 2$ approximation to {\OptimalItemPricing}.
\end{flushleft}
\end{restate}

\begin{restate}[{\Cref{cor:OIP}}]
\begin{flushleft}
{\OptimalItemPricing} achieves \\
(i)~no better than a $\calC_{\DRB / \OIP} \geq 2$ approximation to {\DualityRelaxationBenchmark}.
\end{flushleft}
\end{restate}

First, by reconsidering \Cref{exp:DRB_SPP} from \Cref{subsec:SPP_LB}, we obtain a lower bound of $\calC_{\OLP / \UIVVIP} \geq \calC_{\OIP / \UIVVIP} \geq 2$ on the revenue guarantee of {\UIVVItemPricing} against either {\OptimalLotteryPricing} or {\OptimalItemPricing}.

\begin{restate}[{\Cref{exp:DRB_SPP}}]
\begin{flushleft}
\textnormal{There are $n \geq 2$ i.i.d.\ items $\bF = \{F_{i}\}_{i \in [n]}$ with a common revenue-quantile curve $R_{i}(q) = R(q) \eqdef \frac{1}{n} + q$ for $q \in [0,\, 1]$ and (elementary algebra) a common value CDF $F_{i}(v) = F(v) \eqdef 1 - \frac{1 / n}{v - 1}$ for $v \in [\frac{n + 1}{n},\, +\infty]$.}
\end{flushleft}
\end{restate}

\begin{proof}[Proof of \Cref{thm:UIVVIP} (The Lower-Bound Parts of Item~(ii) and Item~(iii))]
We have shown in \Cref{subsec:SPP_LB} that the {\UIVVItemPricing} revenue is at most $\UIVVIP(\bF) \leq \frac{n + 1}{n}$. Below we further show that the {\OptimalItemPricing} revenue is at least $\OIP(\bF) \geq 2$.

For a large enough $t \to +\infty$, let us consider the item-wise prices $\bp^{[t]} = (p_{i})_{i \in [n]}$ that $p_{i} = t$ for each item $i \in [n - 1]$ and $p_{n} = \frac{n + 1}{n}$ for the index-$n$ item.
In this way, the buyer is always willing to purchase the index-$n$ item $\{v_{n} - \frac{n + 1}{n} \geq 0\}$ but prefers another item $\argmax_{i \in [n - 1]} \{v_{i}\}$ if and only if that item maximizes his utility $\{\max(\bv_{-n}) - t \geq v_{n} - \frac{n + 1}{n}\}$.
By elementary algebra, the resulting {\ItemPricing} revenue with $t \to +\infty$ is equal to
\begin{align*}
    \lim_{t \to +\infty} \IP(\bF,\, \bp^{[t]})
    & ~=~ \lim_{t \to +\infty} \Big(t \cdot \int_{v = \frac{n + 1}{n}}^{+\infty} \Big(1 - (F(v + t - \tfrac{n + 1}{n}))^{n - 1}\Big) \cdot \d F(v) \\
    & \phantom{~=~ \lim_{t \to +\infty} \Big(} + \tfrac{n + 1}{n} \cdot \Big(1 - \int_{v = \frac{n + 1}{n}}^{+\infty} \Big(1 - (F(v + t - \tfrac{n + 1}{n}))^{n - 1}\Big) \cdot \d F(v)\Big)\Big) \\
    & ~=~ \lim_{t \to +\infty} \Big((t - \tfrac{n + 1}{n}) \cdot \int_{v = \frac{n + 1}{n}}^{+\infty} \Big(1 - \Big(1 - \tfrac{1 / n}{v + t - \frac{n + 1}{n} - 1}\Big)^{n - 1}\Big) \cdot \tfrac{1 / n}{(v - 1)^{2}} \cdot \d v\Big) + \tfrac{n + 1}{n} \\
    & ~=~ \lim_{t \to +\infty} \Big((t - \tfrac{n + 1}{n}) \cdot \int_{v = \frac{n + 1}{n}}^{+\infty} (1 \pm o_{t}(1)) \cdot \tfrac{n - 1}{n \cdot (v + t - \frac{n + 1}{n} - 1)} \cdot \tfrac{1 / n}{(v - 1)^{2}} \cdot \d v\Big) + \tfrac{n + 1}{n} \\
    & ~=~ \lim_{t \to +\infty} \Big((1 \pm o_{t}(1)) \cdot t \cdot \tfrac{n - 1}{n \cdot t} \cdot \int_{v = \frac{n + 1}{n}}^{+\infty} \tfrac{1 / n}{(v - 1)^{2}} \cdot \d v\Big) + \tfrac{n + 1}{n} \\
    & ~=~ \tfrac{n - 1}{n} + \tfrac{n + 1}{n}
    ~=~ 2.
\end{align*}
The revenues from {\OptimalLotteryPricing} and {\OptimalItemPricing} clearly are at least this amount $\OLP(\bF) \geq \OIP(\bF) \geq \lim_{t \to +\infty} \IP(\bF,\, \bp^{[t]}) \geq 2$.
Thus, either revenue gap is at least $\geq \frac{2}{1 + 1 / n}$, which can be arbitrarily close to $2$ for a large enough number of items $n \gg 2$.
This finishes the proof.
\end{proof}

Moreover, by considering the next \Cref{exp:DRB_OIP}, we obtain a lower bound of $\calC_{\DRB / \OIP} \geq 2$ on the revenue gap between {\DualityRelaxationBenchmark} and {\OptimalItemPricing}.

\begin{example}[{Lower-Bound Instances for $\calC_{\DRB / \OIP} \geq 2$}]
\label{exp:DRB_OIP}
\begin{flushleft}
There are $n \geq 2$ i.i.d.\ items $\bF = \{F_{i}\}_{i \in [n]}$ with a common revenue-quantile curve $R_{i}(q) = R(q) \eqdef 1$ for $q \in [0,\, 1]$ and (elementary algebra) a common value distribution $F_{i}(v) = F(v) \eqdef 1 - 1 / v$ for $v \in [1,\, +\infty]$.
\end{flushleft}
\end{example}

\begin{proof}[Proof of \Cref{cor:OIP} (The Lower-Bound Part of Item~(i))]
We can easily check that $\bF$ is a {\em linear} instance (\Cref{def:perturb}).
By \Cref{lem:revenue_formula}, the {\MyersonAuction} revenue is $\MA(\bF) = n R(0) = n$, and the {\DualityRelaxationBenchmark} revenue is
\begin{align*}
    \DRB(\bF) = n R(0) + \SPA(\bF) = n + n = 2n.
\end{align*}
The second step formulates the {\SecondPriceAuction} revenue using the second-highest value CDF $S(v) \eqdef (F(v))^{n} + n \cdot (F(v))^{n - 1} \cdot (1 - F(v))$.
We thus have $\OIP(\bF) \leq \MA(\bF) = n$ since, as shown in \cite[Lemma~5]{CHK07}, {\MyersonAuction} revenue-surpasses {\OptimalItemPricing} on the same instance.
(On the other hand, we also have $\OIP(\bF) \geq \lim_{t \to +\infty} \IP(\bF,\, \bp^{[t]}) = n$ by using the item-wise prices $\bp^{[t]} = (t,\, t,\, \dots,\, t)$.)
This finishes the proof.
\end{proof}

\section*{Acknowledgements}
We are grateful to Jason Hartline, Yingkai Li, and Jinzhao Wu for invaluable discussions and to Dongxu Wang for help with mathematical experiments.

\begin{flushleft}
\bibliographystyle{alphaurl}
\bibliography{main}
\end{flushleft}

\appendix

\section{The Single-Dimensional Representative Approach}
\label{sec:randomness}

\begin{restate}[{\Cref{prop:representative}}]
\begin{flushleft}
Given the same instance $\bF$ and the same \textbf{deterministic} prices $\bp$,
{\ItemPricing} yields a higher revenue than {\SequentialPostedPricing} in the order-oblivious model $\IP(\bF,\, \bp) \geq \SPP^{*}(\bF,\, \bp)\ignore{\min_{\pi \in \Pi_{n}} \{\SPP(\bF,\, \bp,\, \pi)\}}$.
\end{flushleft}
\end{restate}

We emphasize that, to validate \Cref{prop:representative} as a black-box reduction, the involved prices $\bp$ must be deterministic.
Namely, randomness of the prices $\bp$ can benefit {\SequentialPostedPricing}\footnote{In contrast, randomness of the prices $\bp$ cannot benefit {\ItemPricing}, because randomized prices $\bp$ cannot yield a better expected revenue than their revenue-optimal outcomes $\argmax_{\tilde{\bp} \in \supp(\bp)} \IP(\bF,\, \tilde{\bp})$.}
and may make it surpass the corresponding {\ItemPricing} (even if the value distributions $\bF$ are i.i.d.\ regular \cite{M81}).
Here is a simple example.

\begin{example}[The Effects of Randomness]
\label{exp:randomness}
\begin{flushleft}
There are two i.i.d.\ random values $v_{i}$ for $i \in \{1,\, 2\}$; each value is equal to $2$ with probability $\frac{1}{2}$ and otherwise is uniformly distributed over $[1,\, 2)$.
It is easy to check that this is a regular instance (cf.\ \Cref{def:regular}).
Now consider two i.i.d.\ random prices $\Pr[p_{i} = 1] = \Pr[p_{i} = 2] = \frac{1}{2}$ for $i \in \{1,\, 2\}$. By elementary algebra, we have
\begin{align*}
    \E_{\bp} [\, \IP(\bF,\, \bp) \,]
    & ~=~ 2 \cdot \Pr\ignore{_{\bv,\, \bp}} [\, \bv = \bp = \mathbf{2} \,]
    + 1 \cdot (1 - \Pr\ignore{_{\bv,\, \bp}} [\, \bv = \bp = \mathbf{2} \,] - \Pr\ignore{_{\bv,\, \bp}} [\, \bv \prec \bp = \mathbf{2} \,]) \\
    & ~=~ 2 \cdot \tfrac{1}{16} + 1 \cdot (1 - \tfrac{1}{16} - \tfrac{1}{16})
    ~=~ 1.
\end{align*}
Since the i.i.d.\ random prices $p_{1},\, p_{2}$ are order-symmetric to the i.i.d.\ random values $v_{1},\, v_{2}$, without loss of generality we can consider the order $\pi = (1,\, 2)$. Then, also by elementary algebra, we have
\begin{align*}
    \E_{\bp} [\, \SPP^{*}(\bF,\, \bp) \,]
    & ~=~ 1 \cdot \Pr [\, p_{1} = 1 \,]
    + 2 \cdot \Pr [\, v_{1} = p_{1} = 2 \,] \\
    & \phantom{~=~} + (1 \cdot \Pr [\, p_{2} = 1 \,]
    + 2 \cdot \Pr [\, v_{2} = p_{2} = 2 \,]) \cdot \Pr[\, v_{1} < p_{1} = 2 \,]
    \hspace{1.1cm} \\
    & ~=~ 1 \cdot \tfrac{1}{2} + 2 \cdot \tfrac{1}{4} + (1 \cdot \tfrac{1}{2} + 2 \cdot \tfrac{1}{4}) \cdot \tfrac{1}{4}
    ~=~ \tfrac{5}{4}.
\end{align*}
In sum, these randomized prices $\bp$ make {\SequentialPostedPricing} revenue-surpass {\ItemPricing}.
\end{flushleft}
\end{example}

Let us rethink \Cref{exp:randomness} carefully: We can apply \Cref{prop:representative} to every possible outcome of the prices $\in \{(1,\, 1),\, (1,\, 2),\, (2,\, 1),\, (2,\, 2)\}$. However, since the arrival order $\pi \in \Pi_{n}$ for {\SequentialPostedPricing} chosen by the adversary relies on the outcome, over the randomness of the prices $\bp$ we cannot conclude with $\E_{\bp} [ \IP(\bF,\, \bp) ] \geq \E_{\bp} [ \SPP^{*}(\bF,\, \bp) ]$ (which indeed is false).

\section{Polynomial-Time Implementation of Simple Mechanisms}
\label{sec:implementation}

To implement {\UniformIronedVirtualValueItemPricing} under the prices $\bp = \{p_{i}\}_{i \in [n]}$ determined in \Cref{def:UIVV_prices,thm:DRB_SPP}, we assume that the value distributions $\bF = \{F_{i}\}_{i \in [n]}$ are discrete, with support sizes $m_{i} \eqdef |\supp(F_{i})|$ for $i \in [n]$.
This is a standard assumption when computation is involved (see, e.g., \cite{DDT14,CDPSY18,CDOPSY22,CMPY18}). Let $m \eqdef \max_{i \in [n]} \{ m_{i} \}$.

Specifically, we use the prices $p_{i} \eqdef \sup \{ v \in \RR \mid \bar{\varphi}_{i}(v) < T \}$ for $i \in [n]$, where the $\bar{\varphi}_{i}$'s are the increasing ironed virtual value functions and the {\UIVV}-threshold $T = \frac{1}{3} \cdot \DRB(\bF)$ is one third of the {\DualityRelaxationBenchmark} revenue. We can compute these prices $\{p_{i}\}_{i \in [n]}$ as follows.
\begin{flushleft}
\begin{itemize}
 \item Compute the revenue-quantile curves $\{R_{i}\}_{i \in [n]}$ in time $O(n m)$.
 
 \item Compute the ironed revenue-quantile curves $\bar{R}_{i} = \Conv(R_{i})$, which are $(\leq m)$-piecewise linear concave functions, in time $O(n m \log m)$ by the Graham scan algorithm \cite{G72}.
 
 \item Compute the ironed virtual value functions $\{\bar{\varphi}_{i}\}_{i \in [n]}$ and the ironed virtual value CDF's $\{\bar{D}_{i}\}_{i \in [n]}$, which are $(\leq m)$-piecewise constant functions, in time $O(n m)$.
 
 \item Compute the auxiliary value CDF's $\{F_{-i}\}_{i \in [n]}$, where each CDF $F_{-i}(v) \eqdef \prod_{j \neq i} F_{j}(v)$ is a $(\leq n m)$-piecewise constant function, in time $O(n^{2} m)$.
 
 \item Compute, based on $\{F_{i}\}_{i \in [n]}$, $\{F_{-i}\}_{i \in [n]}$, and $\{\bar{\varphi}_{i}\}_{i \in [n]}$, the {\DualityRelaxationBenchmark} revenue $\DRB(\bF) = \E_{\bv \sim \bF} \big[ \max_{i \in [n]} \big\{ \bar{\varphi}_{i}(v_{i}) \cdot \indicator(i = \argmax(\bv)) + v_{i} \cdot \indicator(i \neq \argmax(\bv)) \big\} \big]$ \\
 in time $O(\sum_{i \in [n]} |F_{-i}| \cdot |F_{i}|) = O(n^{2} m^{2})$ by enumerating all events $\{\argmax(\bv) = i\} \wedge \{v_{i} \in \supp(F_{i})\} \wedge \{\max(\bv_{-i}) \in \supp(F_{-i})\}$.
 
 \item Compute, based on $\{\bar{\varphi}_{i}\}_{i \in [n]}$ and $T = \frac{1}{3} \cdot \DRB(\bF)$, the prices $\{p_{i}\}_{i \in [n]}$ in time $O(n \log m)$ by binary search.
\end{itemize}
\end{flushleft}
In sum, we can compute the prices $\{p_{i}\}_{i \in [n]}$ in time $O(n^{2} m^{2})$.

\section{Revenue Non-Monotonicity of {\DualityRelaxationBenchmark}}
\label{sec:non-monotonicity}

Below, we show that the {\DualityRelaxationBenchmark} \cite{CDW21} does not satisfy revenue monotonicity in both the \textit{additive single-buyer} case (\Cref{sec:non-monotonicity:additive}) and the \textit{unit-demand multi-buyer} case (\Cref{sec:non-monotonicity:unit-demand-multi-buyer}).

\subsection{The Additive Single-Buyer Case}
\label{sec:non-monotonicity:additive}

Given $v_{i} \sim F_{i}$ for $i \in [n]$, an {\em additive} buyer has a value of $v(S) = \sum_{i \in S} v_{i}$ for every bundle $S \subseteq [n]$. The revenue formula of the {\DualityRelaxationBenchmark} \cite{CDW21} is given by
\begin{align*}
    \DRB(\bF)
    ~\eqdef~ \E_{\bv \sim \bF} \Big[\, \sum_{i \in [n]} \big(\bar{\varphi}_{i}(v_{i}) \cdot \indicator(i = \argmax(\bv)) + v_{i} \cdot \indicator(i \ne \argmax(\bv))\big) \,\Big].
\end{align*}
The following \Cref{lem:additive} shows that, in general, the {\DualityRelaxationBenchmark} does not satisfy the revenue monotonicity in the \textit{additive single-buyer} case.

\begin{lemma}[Revenue Non-Monotonicity]
\label{lem:additive}
\begin{flushleft}
There exist two ``additive single-buyer'' instances $\bF \preceq \tilde{\bF}$ such that {\DualityRelaxationBenchmark} violates revenue monotonicity $\DRB(\bF) > \DRB(\tilde{\bF})$.
\end{flushleft}
\end{lemma}

\begin{example}[Revenue Non-Monotonicity]
\label{exp:additive}
There are two items; when they have the same value, we break ties in favor of the first item ``$\{v_{1} = v_{2}\} \implies \{v_{1} \succ v_{2}\}$''. Let $\eps \to 0^{+}$.

In the first instance $\bF$, the first item $F_{1}$ has a deterministic value of $1$, and the second item $F_{2}$ follows the distribution $\Pr_{v_{2} \sim F_{2}}[v_{2} = 1] = 1 - \eps$ and $\Pr_{v_{2} \sim F_{2}}[v_{2} = 1 / \eps] = \eps$. Given $v_{2} \sim F_{2}$, the outcome benchmark $\DRB(\bv)$ evaluates to
\begin{align*}
    \DRB(\bv) ~=~
    \begin{cases}
        1 + v_{2} = 2, & v_{2} = 1 \\
        1 + \bar{\varphi}_{2}(v_{2}) = 1 + 1 / \eps, & v_{2} = 1 / \eps
    \end{cases}.
\end{align*}
In expectation over $v_{2} \sim F_{2}$, the benchmark $\DRB(\bF)$ evaluates to
\begin{align*}
    \DRB(\bF)
    ~=~ \lim_{\eps \to 0^{+}} \Big(2 \cdot (1 - \eps)
    + (1 + 1 / \eps) \cdot \eps\Big)
    ~=~ 3.
\end{align*}

In the second instance $\tilde{\bF}$, the first item keeps the same $\tilde{F}_{1} = F_{1}$, while the second item follows the truncated equal-revenue distribution $\tilde{F}_{2}(v) \eqdef 1 - 1 / v \cdot \indicator(v \le 1 / \eps)$ for $v > 1$. Given $\tilde{v}_{2} \sim \tilde{F}_{2}$, the outcome benchmark $\DRB(\tilde{\bv})$ evaluates to
\begin{align*}
    \DRB(\tilde{\bv}) ~=~ 1 + \bar{\tilde{\varphi}}_{2}(\tilde{v}_{2}) ~=~
    \begin{cases}
        1, & v_{2} \in (1, 1 / \eps) \\
        1 + 1 / \eps, & v_{2} = 1 / \eps
    \end{cases}.
\end{align*}
In expectation over $\tilde{v}_{2} \sim \tilde{F}_{2}$, the benchmark $\DRB(\tilde{\bF})$ evaluates to
\begin{align*}
    \DRB(\tilde{\bF})
    ~=~ \lim_{\eps \to 0^{+}} \Big(1 \cdot (1 - \eps)
    + (1 + 1 / \eps) \cdot \eps\Big)
    ~=~ 2.
\end{align*}
We conclude the proof of \Cref{lem:additive} by noting that $\bF \preceq \tilde{\bF}$ and $\DRB(\bF) > \DRB(\tilde{\bF})$.
\end{example}

\subsection{The Unit-Demand Multi-Buyer Case}
\label{sec:non-monotonicity:unit-demand-multi-buyer}

Consider $m \ge 2$ unit-demand buyers with independent item-wise values $v_{ji} \sim F_{ji}$ for $ji \in [m] \times [n]$; let $\bv_{j} \eqdef (v_{ji})_{i \in [n]}$ for $j \in [m]$. The revenue formula of the {\DualityRelaxationBenchmark} \cite{CDW21} concerns a buyer-item bipartite matching $M \in \calM([m] \times [n])$ and is given by
\begin{align*}
    \DRB(\bF)
    ~\eqdef~ \E_{\bv \sim \bF} \Big[\, \max_{M \in \calM([m] \times [n])} \sum_{ji \in M} \big(\bar{\varphi}_{ji}(v_{ji}) \cdot \indicator(i = \argmax(\bv_{j})) + v_{ji} \cdot \indicator(i \ne \argmax(\bv_{j}))\big) \,\Big].
\end{align*}
The following \Cref{lem:unit-demand-multi-buyer} shows that, in general, the {\DualityRelaxationBenchmark} does not satisfy the revenue monotonicity in the \textit{unit-demand multi-buyer} case.


\begin{lemma}[Revenue Non-Monotonicity]
\label{lem:unit-demand-multi-buyer}
\begin{flushleft}
There exist two ``unit-demand multi-buyer'' instances $\bF \preceq \tilde{\bF}$ such that {\DualityRelaxationBenchmark} violates revenue monotonicity $\DRB(\bF) > \DRB(\tilde{\bF})$.
\end{flushleft}
\end{lemma}

\begin{example}[Revenue Non-Monotonicity]
\label{exp:unit-demand-multi-buyer}
There are two buyers and two items; we break ties such that ``$\{v_{21} = v_{22}\} \implies \{v_{21} \succ v_{22}\}$'' (but otherwise arbitrarily). Let $\eps \to 0^{+}$.

In the first instance $\bF$:
For the first buyer, the first item $F_{11}$ has a deterministic value of $2$, and the second item $F_{12}$ has a deterministic value of $0$.
For the second buyer, the first item $F_{21}$ has a deterministic value of $1$, and the second item $F_{22}$ follows the distribution $\Pr_{v_{22} \sim F_{22}}[v_{22} = 1] = 1 - \eps$ and $\Pr_{v_{22} \sim F_{22}}[v_{22} = 1 / \eps] = \eps$.
Given $v_{22} \sim F_{22}$, the outcome benchmark $\DRB(\bv)$ evaluates to
\begin{align*}
    \DRB(\bv) ~=~
    \begin{cases}
        2 + v_{22} = 3, & v_{22} = 1 \\
        2 + \bar{\varphi}_{22}(v_{22}) = 2 + 1 / \eps, & v_{22} = 1 / \eps
    \end{cases}.
\end{align*}
In expectation over $v_{22} \sim F_{22}$, the benchmark $\DRB(\bF)$ evaluates to
\begin{align*}
    \DRB(\bF)
    ~=~ \lim_{\eps \to 0^{+}} \Big(3 \cdot (1 - \eps)
    + (2 + 1 / \eps) \cdot \eps\Big)
    ~=~ 4.
\end{align*}

In the second instance $\tilde{\bF}$:
For the first buyer, both items keep the same $F_{11} = \tilde{F}_{11}$ and $F_{12} = \tilde{F}_{12}$.
For the second buyer, the first item keeps the same $\tilde{F}_{21} = F_{21}$, while the second item follows the truncated equal-revenue distribution $\tilde{F}_{22}(v) \eqdef 1 - 1 / v \cdot \indicator(v \le 1 / \eps)$ for $v > 1$. Given $\tilde{v}_{22} \sim \tilde{F}_{22}$, the outcome benchmark $\DRB(\tilde{\bv})$ evaluates to
\begin{align*}
    \DRB(\tilde{\bv}) ~=~ 2 + \bar{\tilde{\varphi}}_{22}(\tilde{v}_{22}) ~=~
    \begin{cases}
        2, & v_{22} \in (1, 1 / \eps) \\
        2 + 1 / \eps, & v_{22} = 1 / \eps
    \end{cases}.
\end{align*}
In expectation over $\tilde{v}_{22} \sim \tilde{F}_{22}$, the benchmark $\DRB(\tilde{\bF})$ evaluates to
\begin{align*}
    \DRB(\tilde{\bF})
    ~=~ \lim_{\eps \to 0^{+}} \Big(2 \cdot (1 - \eps)
    + (2 + 1 / \eps) \cdot \eps\Big)
    ~=~ 3.
\end{align*}
We conclude the proof of \Cref{lem:unit-demand-multi-buyer} by noting that $\bF \preceq \tilde{\bF}$ and $\DRB(\bF) > \DRB(\tilde{\bF})$.
\end{example}

\section{The Exact Worst Cases}
\label{sec:worst_case}

Following \Cref{subsec:perturb}, below we sketch how to characterize the worst cases among {\em linear} instances, based on {\bf revenue-quantile curves}. (The arguments below can be informal; the goal is to reveal the structures of the exact worst cases.)

Recall that $n \geq 2$. First, we transform an {\em $n$-item linear} instance $\bR$ into an {\em $(n + 1)$-item linear} instance $\bR \otimes R_{T}$, where the new item $R_{T}(q) \eqdef T \cdot q$ for $q \in [0,\, 1]$ has a deterministic value of $T = \alpha \cdot \DRB(\bR)$ and the same division point $\kappa_{T} = 1$ or $\kappa_{T} = 0$ as the input $\bkappa = \ones$ or $\bkappa = \zeros$.




\begin{lemma}
\begin{flushleft}
Given an $n$-item linear instance $\bR$, the $(n + 1)$-item instance $\bR \otimes R_{T}$ is a linear instance that $\DRB(\bR \otimes R_{T}) = \DRB(\bR)$ and $\UIVVSPP(\bR \otimes R_{T},\, T) \leq \UIVVSPP(\bR,\, T)$.
\end{flushleft}
\end{lemma}

\begin{proof}[Proof (Sketched)]
The new instance $\bR \otimes R_{T}$ must satisfy {\bf linearity} (\Cref{def:perturb}).
The benchmark revenue keeps the same $\DRB(\bR \otimes R_{T}) = \DRB(\bR)$; for each of the $n \geq 2$ given item $R_{i}$, its value is supported on $[R_{i}(1), +\infty] = [R_{i}(0) + T, +\infty]$ and its virtual value is supported on $\{T, +\infty\}$, so the new item $R_{T}$ never contributes the benchmark revenue.
When $\bkappa = \ones$ and $\kappa_{T} = 1$ (\Cref{eq:VUSPP_positive:1}):
$\UIVVSPP(\bR \otimes R_{T},\, T) = R_{T}(0) + T = T \leq \min_{i \in [n]} \{ R_{i}(0) \} + T = \UIVVSPP(\bR,\, T)$.
Moreover, when $\bkappa = \zeros$ and $\kappa_{T} = 0$ (\Cref{eq:VUSPP_negative:1}):
$\UIVVSPP(\bR \otimes R_{T},\, T) = \sum_{i \in [n]} R_{i}(0) + R_{T}(0) = \sum_{i \in [n]} R_{i}(0) = \UIVVSPP(\bR,\, T)$.
\end{proof}

Because $\DRB(\bR \otimes R_{T}) = \DRB(\bR)$, the given {\UIVV} threshold $T = \alpha \cdot \DRB(\bR)$ remains the new instance $\bR \otimes R_{T}$'s {\UIVV} threshold.




Next, we convert an {\em $(n + 1)$-item linear} instance $\bR \otimes R_{T}$ into an {\em ``almost symmetric'' $(n + 1)$-item linear} instance $\tilde{\bR} \otimes R_{T} = \{\tilde{R}\}^{\otimes n} \otimes R_{T}$, where the symmetric items $\tilde{R}(q) \eqdef \frac{\sum_{i \in [n]} R_{i}(0)}{n} + T \cdot q$ for $q \in [0,\, 1]$ and the division points keep the same $(\tilde{\bkappa},\, \kappa_{T}) = (\bkappa,\, \kappa_{T})$.



\begin{lemma}
\begin{flushleft}
Given a linear instance $\bR \otimes R_{T}$, the almost symmetric instance $\tilde{\bR} \otimes R_{T}$ is a linear instance that $\DRB(\tilde{\bR} \otimes R_{T}) \geq \DRB(\bR \otimes R_{T})$ and $\UIVVSPP(\tilde{\bR} \otimes R_{T},\, T) = \UIVVSPP(\bR \otimes R_{T},\, T)$.
\end{flushleft}
\end{lemma}


    


\begin{proof}[Proof (Sketched)]
The new instance $\tilde{\bR} \otimes R_{T}$ satisfies {\bf linearity} (\Cref{def:perturb}).
The benchmark revenue can only increase; following :
\begin{align*}
    \DRB(\tilde{\bR} \otimes R_{T})
    & ~=~ \Big(n \tilde{R}(0) + R_{T}(0)\Big) + \SPA(\tilde{\bR} \otimes R_{T}) \\
    & ~=~ n \tilde{R}(0) + \SPA(\tilde{\bR}) \\
    & ~=~ n \tilde{R}(0) + \MA(\tilde{\bR}) \\
    & ~=~ 2n \tilde{R}(0) + T \\
    & ~=~ 2 \cdot \Big(\sum_{i \in [n]} R_{i}(0) + R_{T}(0)\Big) + T
    ~\geq~ \DRB(\bR \otimes R_{T}).
\end{align*}
The second step: $R_{T}$'s deterministic value $\equiv T$ is below $\tilde{\bR}$'s value supports $[\frac{\sum_{i \in [n]} R_{i}(0)}{n} + T,\, +\infty]$, thus no effect in {\SecondPriceAuction}.
The third step: $\tilde{\bR}$ are symmetric and have positive virtual values $\in \{T,\, +\infty\}$, so {\MyersonAuction} and {\SecondPriceAuction} are revenue-equivalent (\Cref{prop:revenue_equivalence}).
All the other steps follow from \Cref{lem:revenue_formula}.

That $\UIVVSPP(\tilde{\bR} \otimes R_{T},\, T) = \UIVVSPP(\bR \otimes R_{T},\, T)$ is obvious, whether $\tilde{\bR} \otimes R_{T}$ and $\bR \otimes R_{T}$ are ``positive'' linear instances (\Cref{eq:VUSPP_positive:1}) or ``negative'' linear instances (\Cref{eq:VUSPP_negative:1}).
\end{proof}

The new instance $\tilde{\bR} \otimes R_{T}$ slightly differs from our target. But reusing the {\scale} reduction, we can easily modify it into a desired instance $\hat{\bR} \otimes R_{T} = \{\hat{R}\}^{\otimes n} \otimes R_{T}$ that
$\DRB(\hat{\bR} \otimes R_{T}) = \DRB(\bR \otimes R_{T})$ (hence the same {\UIVV} threshold $\hat{T} = T$) and $\UIVVSPP(\hat{\bR} \otimes R_{T},\, \hat{T}) \leq \UIVVSPP(\bR \otimes R_{T},\, T)$.
\begin{itemize}
    \item {\DualityRelaxationBenchmark}:
    $\hat{\bR} \otimes R_{T}$ has a {\UIVV} threshold $T = \alpha \cdot \DRB$ and a benchmark revenue $\DRB = 2n\hat{R}(0) + T$, which means $T = \frac{2\alpha}{1 - \alpha} \cdot n\hat{R}(0)$ and $\DRB = \frac{2}{1 - \alpha} \cdot n\hat{R}(0)$.

    \item {\UIVVSequentialPostedPricing}:
    $\UIVVSPP = T = \frac{2\alpha}{1 - \alpha} \cdot n\hat{R}(0)$ for a ``positive'' {\em linear} instance (\Cref{eq:VUSPP_positive:1}) and $\UIVVSPP = n\hat{R}(0)$ for a ``negative'' {\em linear} instance (\Cref{eq:VUSPP_negative:1}).
    Clearly, the worst case refers to the smaller revenue $= \min\{ \frac{2\alpha}{1 - \alpha},\, 1 \} \cdot n R(0)$.
\end{itemize}
To conclude, a parameter $\alpha = (0,\, 1)$ for {\UIVVSequentialPostedPricing} induces a revenue guarantee of $\calC(\alpha) = \frac{2}{\min\{ 2\alpha,\, 1 - \alpha \}}$, which is minimized $\calC_{\DRB / \UIVVSPP} = \calC(\alpha^{*}) = 3$ at $\alpha^{*} = \frac{1}{3}$.
For completeness, we explicitly show the worst-case instances below.

\begin{example}[{Worst-Case Instances for $\calC_{\DRB / \UIVVSPP} = 3$}]
\begin{flushleft}
There are $(n + 1) \geq 3$ items $\{\hat{F}\}^{\otimes n} \otimes F_{T}$ with revenue-quantile curves $\hat{R}(q) \eqdef \frac{1}{n} + q$ for $q \in [0,\, 1]$ and $R_{T}(q) \eqdef q$ for $q \in [0, 1]$ and (elementary algebra) value CDF's $\hat{F}(v) \eqdef 1 - \frac{1 / n}{v - 1}$ for $v \in [\frac{n + 1}{n},\, +\infty]$ and $F_{T}(v) \eqdef \indicator(v > 1)$.
\end{flushleft}
\end{example}

\end{document}